\documentclass[12pt]{article}
\pdfoutput=1
\usepackage[nosort]{cite}
\usepackage[height=8.85in,width=6.45in]{geometry}
\usepackage{pifont}
\usepackage{amsthm}
\usepackage[mathscr]{euscript}

\usepackage{times}
\usepackage{color}
\usepackage{pdflscape}
\usepackage{tikz}
\usetikzlibrary{arrows.meta, positioning, calc}
\usepackage{xfrac}
\usepackage{booktabs}
\usepackage[utf8]{inputenc}
\usepackage{amsmath}
\usepackage{amssymb}
\usepackage{mathtools}
\numberwithin{equation}{section}
\usepackage{slashed}
\usepackage{braket}
\usepackage[svgnames]{xcolor}
\usepackage{url}
\usepackage[colorlinks,citecolor=DarkGreen,linkcolor=FireBrick,linktocpage]{hyperref}
\usepackage{cite}
\usepackage{graphicx}
\usepackage{float}
\usepackage{tikz-cd}
\usepackage{stmaryrd}
\usepackage{courier}
\usepackage{dashbox}
\usepackage{caption}
\usepackage{subcaption}
\usepackage{enumitem}
\usepackage{footmisc}
\usepackage{tcolorbox}
\usepackage{fix-cm}

\tcbuselibrary{theorems, skins, breakable}

\newcounter{mainresultcounter}

\newtcolorbox{mainresult}[1][]{
  enhanced,
  colback=blue!5!white,
  colframe=blue!60!black,
  fonttitle=\bfseries,
  title=Main Result~\refstepcounter{mainresultcounter}\themainresultcounter,
  attach boxed title to top left={yshift=-2mm, xshift=4mm},
  boxed title style={
    colback=blue!60!black,
    colframe=blue!60!black,
    rounded corners
  },
  rounded corners,
  drop shadow,
  #1
}

\newtcbtheorem[number within=section]{claim}{Claim}{
  enhanced,
  colback=violet!5!white,
  colframe=violet!50!black,
  fonttitle=\bfseries,
  attach boxed title to top left={yshift=-2mm, xshift=4mm},
  boxed title style={
    colback=violet!50!black,
    colframe=violet!50!black,
    rounded corners
  },
  rounded corners,
}{cl}

\newtcbtheorem[use counter from=claim]{example}{Example}{
  enhanced,
  colback=teal!5!white,
  colframe=teal!50!black,
  fonttitle=\bfseries\itshape,
  attach boxed title to top left={yshift=-2mm, xshift=4mm},
  boxed title style={
    colback=teal!50!black,
    colframe=teal!50!black,
    rounded corners
  },
  rounded corners,
}{ex}

\usepackage{dsfont}
\usepackage{array}
\usepackage{multirow}
\usepackage{rotating}

\newcolumntype{H}{>{\setbox0=\hbox\bgroup}c<{\egroup}@{}}

\renewcommand{\title}[1]{\vbox{\center\LARGE{#1}}\vspace{5mm}}
\renewcommand{\author}[1]{\vbox{\center#1}\vspace{5mm}}
\newcommand{\address}[1]{\vbox{\center\em#1}}

\makeatletter
\newsavebox{\@brx}
\newcommand{\llangle}[1][]{\savebox{\@brx}{\(\m@th{#1\langle}\)}%
  \mathopen{\copy\@brx\kern-0.5\wd\@brx\usebox{\@brx}}}
\newcommand{\rrangle}[1][]{\savebox{\@brx}{\(\m@th{#1\rangle}\)}%
  \mathclose{\copy\@brx\kern-0.5\wd\@brx\usebox{\@brx}}}
\makeatother

\def\d{\mathsf{d}}
\def\CC{{\mathcal C}}

\def\Vec{\mathop{\mathrm{Vec}}\nolimits}

\newcommand{\Rep}{\mathrm{Rep}}
\newcommand{\Ver}{\mathrm{Ver}}

\newtheorem{theorem}{Theorem}[section]
\newtheorem{definition}[theorem]{Definition}

\newtheorem{conjecture}[theorem]{Conjecture}

\begin{document}

\begin{titlepage}

\title{Hypergroup Symmetry in Relative Quantum Field Theories and Chiral Algebras}

\author{Terry Gannon${}^1$ and Brandon C.\ Rayhaun${}^{2}$}

        \address{${}^{1}$Department of Mathematics, University of Alberta, Edmonton, Alberta, Canada\\
        ${}^{2}$School of Natural Sciences, Institute for Advanced Study, Princeton, NJ, USA}

\abstract

A QFT is said to be \emph{relative} if it lives at the boundary of a topological QFT in one higher dimension. We develop a general framework for working with noninvertible symmetries of relative theories in two spacetime dimensions, extending several well-known results for absolute QFTs. We emphasize various new features which arise in the relative setting, including the role of topological surfaces of the bulk, and the appearance of hypergroups and certain generalizations of tube algebras known as \emph{dome algebras}.  Our formalism is particularly well-suited for studying rational chiral algebras, where it predicts that finite symmetries are in explicit one-to-one correspondence with conformal embeddings of finite index. 

We describe several implications of our framework for absolute theories. First, we explain how to  ``glue'' together symmetries of the left- and right-moving chiral algebras of a 2D CFT to produce topological line defects of the full theory. Second, we derive a precise correspondence between boundary conditions of a 2D CFT and symmetries of its chiral algebra. This correspondence has several structural corollaries: in diagonal rational CFTs, we demonstrate that the topological line defects of the theory act transitively on its boundary conditions, and further that the identity Cardy state has the smallest $g$-function amongst all boundary conditions, including those which only preserve Virasoro symmetry.

We conclude by illustrating our results in a variety of examples. For instance, we show that, if there exists a rational chiral algebra with central charge $c=8$ whose modular tensor category is the Drinfeld center of the Haagerup fusion category, then it must arise as the fixed points of a rank-2 hypergroup acting on the $SU(3)_1\otimes (E_{6})_1$ chiral algebra.

\end{titlepage}

\eject

\setcounter{tocdepth}{3}
\tableofcontents

\clearpage 

\section{Introduction}

A relative quantum field theory (QFT) in $d$ spacetime dimensions is by definition a theory which lives at the boundary of a nontrivial topological QFT (TQFT) in one dimension higher \cite{Freed:2012bs}. A well-known family of examples is the 6D $\mathcal{N}=(2,0)$ superconformal field theories \cite{Witten:1995zh,Strominger:1995ac}. Two dimensional (rational) chiral algebras \cite{Zamolodchikov:1985wn} --- a.k.a.\ (strongly rational) vertex operator algebras (VOAs), the structure formed by the sector of holomorphic local operators in a 2D (rational) conformal field theory --- are another rich source of relative theories \cite{Moore:1988qv,Witten:1988hf}. 

The goal of this paper is to develop a physicist-friendly theory of generalized symmetries for chiral algebras and other relative QFTs in $d=2$ spacetime dimensions, extending the beautiful edifice already erected for absolute 2D QFTs \cite{Gaiotto:2014kfa,Bhardwaj:2017xup,Chang:2018iay}.\footnote{See \cite{Petkova:2000ip,Fuchs:2002cm,Frohlich:2004ef,Frohlich:2006ch} for earlier appearances of generalized symmetries in absolute 2D QFTs and \cite{Shao:2023gho,Schafer-Nameki:2023jdn} for reviews. See also \cite{Kong:2017etd,Kong:2019byq,Kong:2019cuu,Bhardwaj:2021mzl,LawrieYuZhang2024IntermediateDefectGroups,FrancoYu2024GeneralizedSym2D} for an incomplete collection of previous studies of generalized symmetries in relative QFTs.} We will see that many qualitatively new features arise in the relative setting. Indeed, we will provide physical interpretations of  algebraic structures which have hitherto not made prominent appearances in the physics literature: namely, hypergroups \cite{Bischoff:2016jmy,Bischoff:2022fxf,Rie22,Dong:2025ttr}, which are spiritually fusion rings with non-integer structure constants, and certain generalizations of tube algebras \cite{ocneanu1994chirality} known as \emph{dome algebras} \cite{Green:2023ork}.\footnote{See however \cite{Kaidi:2024wio} for a recent physics paper where hypergroups arise.}

In a companion paper \cite{grmath}, we will dive deeper into various mathematical aspects whose details are suppressed here. Applications to the classification of symmetries and boundary conditions of $c=1$ conformal field theories will be presented elsewhere.

\subsection{Towards a noninvertible symmetry-enriched Moore-Seiberg theory}

One of our motivations is to take some small steps towards realizing a noninvertible symmetry-enriched generalization of Moore-Seiberg theory \cite{Moore:1988qv,Witten:1988hf}. 

To set the stage, recall the Chern-Simons/Wess-Zumino-Witten (CS/WZW) correspondence, which asserts in particular that the chiral algebra of a 2D WZW model can be thought of as living on the boundary of a corresponding bulk 3D topological CS theory \cite{Witten:1988hf}. It is now appreciated that this is just a particular example of a more general phenomenon in which chiral algebras of arbitrary 2D rational CFTs describe gapless edge modes of 3D TQFTs. Although a given TQFT generally supports several inequivalent chiral algebra boundaries (see e.g.\ \cite[Table 1]{Rayhaun:2023pgc}), it turns out that any one of these chiral boundaries completely determines the bulk. We will call this the \emph{bulk-boundary correspondence}.

\begin{figure}
\centering
    \input{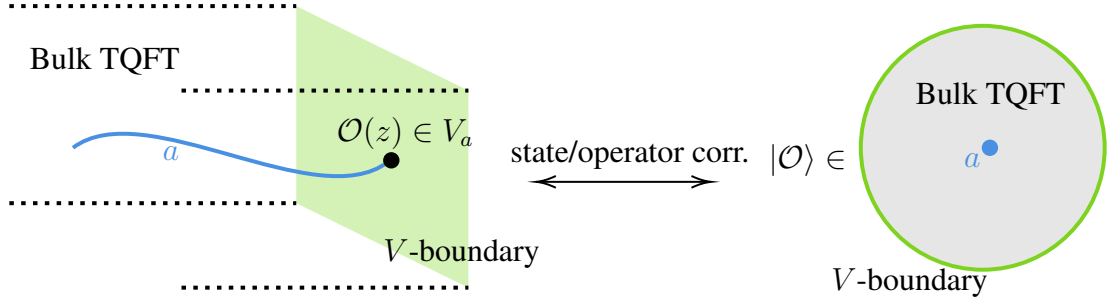}
    \caption{The Hilbert space $V_a$ consists of boundary local operators $\mathcal{O}(z)$ on which the bulk anyon $a$ can terminate. By the state/operator correspondence, it can also be viewed as the Hilbert space that the bulk TQFT assigns to a disk which is punctured at the center by the anyon $a$.}\label{fig:bulkboundary}
\end{figure}
There is a mathematical shadow of the bulk-boundary correspondence which goes through tensor categories. Indeed, starting from the boundary chiral algebra $V$, one can construct the category $\Rep(V)$, whose objects are representations of $V$ and whose morphisms are intertwiners between representations. This category carries a lot of structure: in particular, \cite{Moore:1988qv,Huang:2005gs} showed that it forms what is known today as a modular tensor category (MTC). This MTC is sometimes called the Moore-Seiberg data of $V$.

One can also extract a category $\mathcal{B}$ from the bulk TQFT. Its objects by definition are anyons, or topological line operators, and its morphisms are topological point junctions. Again, this category admits the structure of an MTC \cite{turaev1992modular,Kitaev:2005hzj,turaev2016quantum}. Furthermore, the category of anyons in the bulk TQFT is equivalent to the representation category of the boundary chiral algebra, 
\begin{align}\label{eqn:bulkboundary}
    \mathcal{B}\cong \Rep(V).
\end{align}
In fact, the equivalence can be made explicit, 
\begin{align}
    a\mapsto V_a,
\end{align}
by assigning to a bulk topological line $a\in\mathcal{B}$ the Hilbert space $V_a\in\Rep(V)$ of boundary local operators on which the line $a$ can end. Equivalently, by the state/operator correspondence, $V_a$ can be identified with the Hilbert space of states which the TQFT assigns to a punctured disk, see Figure \ref{fig:bulkboundary}. Equation \eqref{eqn:bulkboundary} is the claimed mathematical shadow of the bulk-boundary correspondence. 

The discussion so far is summarized by the following diagram:

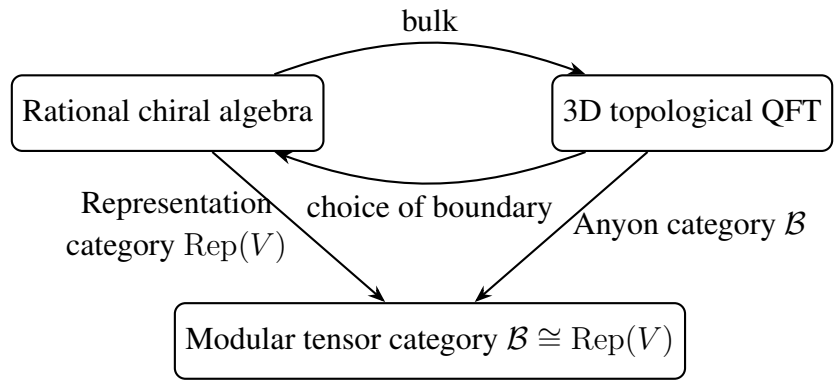
\begin{center}
\begin{tikzpicture}[
    node distance = 2cm,
    box/.style = {
        rectangle,
        rounded corners,
        draw=black,
        thick,
        minimum width=3cm,
        minimum height=1cm,
        text centered,
        align=center
    },
    arrow/.style = {
        ->,
        -{Stealth},
        thick,
        align=center
    }
]

\node[box] (A) {Rational chiral algebra};
\node[box, right=3cm of A] (B) {3D topological QFT};
\node[box, below=2.5cm of $(A)!0.5!(B)$] (C) {Modular tensor category 
$\mathcal{B}\cong \Rep(V)$};

\draw[arrow, bend left=20]  (A) to node[above] {bulk} (B);
\draw[arrow, bend left=20]  (B) to node[below] {choice of boundary} (A);

\draw[arrow] (A) to node[left]  {Representation \\ category $\Rep(V)$} (C);
\draw[arrow] (B) to node[right] {Anyon category $\mathcal{B}$} (C);

\end{tikzpicture}
\end{center}

\begin{figure}
    \centering
     \input{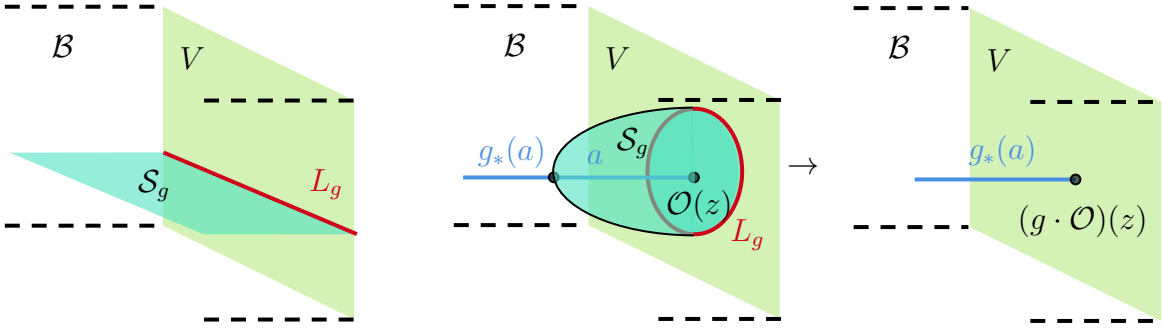}
     \caption{Left: the representation of an  automorphism $g$ of a chiral algebra. Right: the action of an automorphism $g$ on a (non-genuine) local operator $\mathcal{O}(z)$ in a module $V_a$ produces a local operator $(g\cdot \mathcal{O})(z)$ in $V_{g_\ast (a)}$.}\label{fig:autV}
 \end{figure}

Suppose that $V$ further admits an action by a (say, finite) group $G$ of automorphisms. It turns out that the $G$-symmetry of $V$ induces a $G$-symmetry enrichment of the bulk in the sense of \cite{Barkeshli:2014cna}, and hence $V$ can be thought of as living on the boundary of a $G$ symmetry-enriched topological order. We call this the \emph{symmetry-enriched bulk/boundary correspondence}.

Let us briefly give some intuition for why the bulk admits a symmetry enrichment once one specifies a $G$-symmetry of the boundary. Symmetries are implemented by topological defects \cite{Gaiotto:2014kfa}, and in our setup, we can model the $G$-symmetry of $V$ via bulk topological surface operators $\mathcal{S}_g$ terminating on boundary topological line operators $L_g$ for each $g\in G$, as in the left of Figure \ref{fig:autV}. The necessity of incorporating the (possibly trivial) bulk surface $\mathcal{S}_g$ can be motivated by the fact that automorphisms of $V$ generally permute around its representations, $V_a \mapsto V_{g_\ast(a)}$. This permutation is precisely what $\mathcal{S}_g$ implements (see the right of Figure \ref{fig:autV}). The data of a topological surface operator for each $g\in G$ is part of the data of a $G$ symmetry-enrichment of the bulk 3D TQFT, and it turns out that the rest of the data needed to determine a $G$-enrichment is also fixed by the boundary. This is a kind of ``relative'' generalization of the familiar fact that a $G$-symmetry of an absolute QFT fixes a $G$-SPT in one dimension higher via its anomaly.

Again, there is a mathematical shadow of the symmetry-enriched bulk/boundary correspondence which goes through tensor categories. Starting from the boundary side, one can construct the category
\begin{align}
    \Rep_G(V):= \bigoplus_{g\in G} \Rep_g(V),
\end{align}
where $\Rep_g(V)$ is the category of $g$-twisted modules. The formal definition of a $g$-twisted module can be found in e.g.\ \cite{dong1998twisted}, but the intuitive idea is that genuine boundary local operators in $V$ incur a $g$-action when they are transported around a (non-genuine) local operator in a $g$-twisted module, as we will see in a moment. Clearly, $\Rep_G(V)$ contains the Moore-Seiberg data $\Rep(V)\subset \Rep_G(V)$ as a subcategory, and in fact, it is known that it forms what is called a $G$-crossed braided extension of $\Rep(V)$ \cite{Mcrae:2019pol}. This is the sense in which we think of $\Rep_G(V)$ as a symmetry-enrichent of the Moore-Seiberg data of $V$.

\begin{figure}
     \centering
     \input{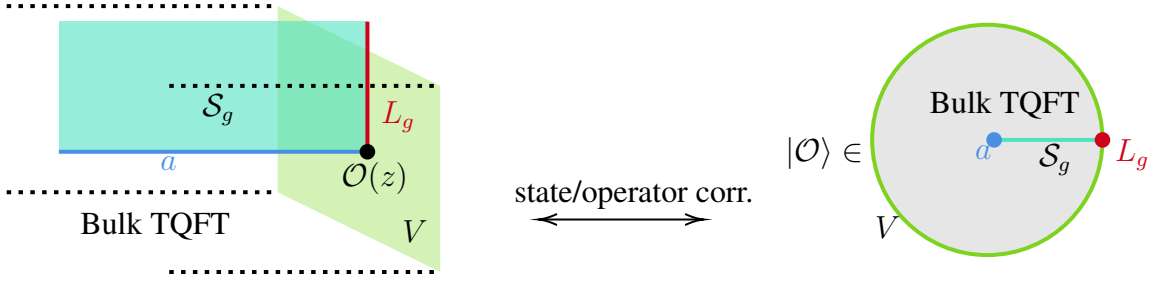}
     \caption{The twisted module $V_a$ consists of (non-genuine) boundary local operators $\mathcal{O}(z)$ which live at the junction of a $g$-twisted anyon $a$ and the boundary topological line junction $L_g$. By the state/operator correspondence, $V_a$ can also be thought of as the Hilbert space that the bulk TQFT assigns to a disk with boundary condition described by $V$ and with insertions of topological defects as shown on the right.}\label{fig:gtwistedmodule}
 \end{figure}

One can similarly extract a symmetry-enriched category from the bulk. Indeed, instead of considering just the category of genuine anyons, one can form the category $\mathcal{B}_g$ of $g$-twisted anyons (i.e.\ topological lines which live at the boundary of the surface $\mathcal{S}_g$) and build the category 
\begin{align}
    \mathcal{B}_G := \bigoplus _{g\in G} \mathcal{B}_g,
\end{align}
which again forms a $G$-crossed braided extension of $\mathcal{B}$ \cite{Barkeshli:2014cna}.
Again, $\mathcal{B}_G$ is equivalent to the category of twisted representations of the boundary chiral algebra 
\begin{align}\begin{split}
    \mathcal{B}_G&\cong \Rep_G(V) \\
    a&\mapsto V_a,
\end{split}
\end{align}
via an explicit equivalence which assigns to a $g$-twisted anyon $a$ a corresponding $g$-twisted module $V_{a}$, defined as the Hilbert space of local operators which sit at the junction of $a$ and $L_g$ as in Figure \ref{fig:gtwistedmodule} (cf.\ Figure \ref{fig:bulkboundary}). The line junction $L_g$ is responsible for implementing the monodromy that a genuine boundary local operator experiences when it is dragged around an insertion $\mathcal{O}(z)\in V_a$.

We therefore have the following diagram in the $G$-enriched setting:

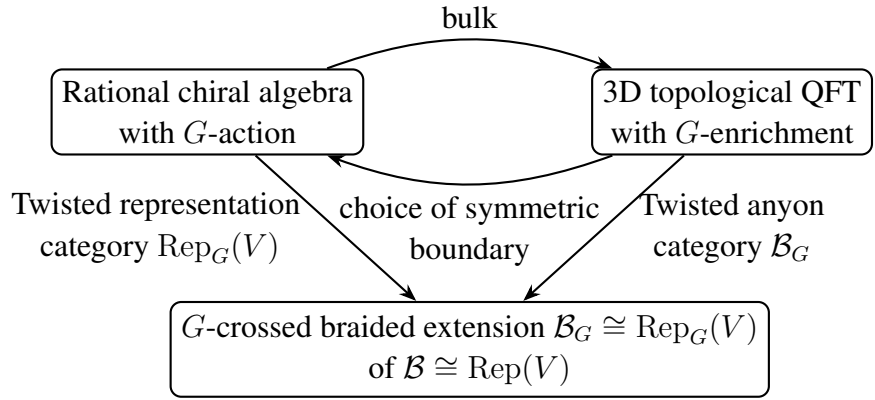
\begin{center}
\begin{tikzpicture}[
    node distance = 2cm,
    box/.style = {
        rectangle,
        rounded corners,
        draw=black,
        thick,
        minimum width=3cm,
        minimum height=1cm,
        text centered,
        align=center
    },
    arrow/.style = {
        ->,
        -{Stealth},
        thick,
        align=center
    }
]

\node[box] (A) {Rational chiral algebra \\ 
with $G$-action};
\node[box, right=3cm of A] (B) {3D topological QFT \\ with $G$-enrichment};
\node[box, below=2.5cm of $(A)!0.5!(B)$] (C) {$G$-crossed braided extension 
$\mathcal{B}_G\cong \Rep_G(V)$ \\ of $\mathcal{B}\cong \Rep(V)$};

\draw[arrow, bend left=20]  (A) to node[above] {bulk} (B);
\draw[arrow, bend left=20]  (B) to node[below] {choice of symmetric \\ boundary} (A);

\draw[arrow] (A) to node[left]  {Twisted representation~~~ \\ category $\Rep_G(V)$~~~} (C);
\draw[arrow] (B) to node[right] {~~~Twisted anyon \\ ~~~category $\mathcal{B}_G$} (C);

\end{tikzpicture}
\end{center}

The question we take up in this paper is what replaces these concepts when the boundary (and hence also the bulk) is enriched with a \emph{noninvertible} symmetry. This is what we mean when we say we are after a noninvertible symmetry-enriched Moore-Seiberg theory. 

As anticipated earlier, we advocate that the correct generalization arises when one replaces groups with hypergroups, which we will define shortly. (See also \cite{Bischoff:2016jmy,Rie22} for closely related ideas.) In particular, in the next subsection, we will define a \emph{hypergroup-graded extension} of the Moore-Seiberg data of a rational chiral algebra $V$ whenever $V$ possesses a noninvertible symmetry. We expect, in analogy with the group-like case, that this category admits additional structure coming from something resembling a generalized braiding, but we leave the determination of what exactly a ``hypergroup-crossed braided tensor category'' is to future work.

\subsection{Hypergroups and boundary topological line operators}

To lay the groundwork for a noninvertible generalization of the discussion from the previous subsection, let us give a reinterpretation of the categories $\mathcal{B}_G$ and $\Rep_G(V)$ in the case that $V$ admits an invertible $G$-action. 

\begin{figure}
\centering 
\input{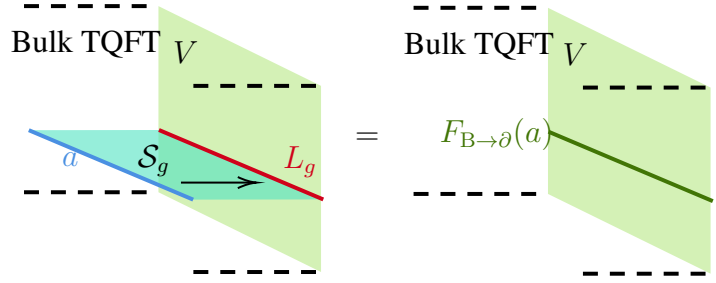}
\caption{The bulk to boundary functor $F_{B\to \partial}$.}\label{fig:FBtopartial}
\end{figure}

Starting with $\mathcal{B}_G$, note that there is a bulk-to-boundary functor $F_{\mathrm{B}\to\partial}$ which takes in a twisted-sector anyon $a$ (i.e.\ an anyon attached to some surface $\mathcal{S}_g$) and produces a boundary topological line operator $F_{\mathrm{B}\to\partial}(a)$, as in Figure \ref{fig:FBtopartial}.\footnote{This functor turns out to be fully faithful, meaning in particular that if $a$ is an elementary twisted anyon, then $F_{\mathrm{B}\to\partial}(a)$ is also an elementary boundary topological line operator. Note that this is not typically the case for gapped boundary conditions, where one often finds that bulk topological lines split into a direct sum when they are brought to the boundary.} Using $F_{\mathrm{B}\to\partial}$, we can equivalently think of $\mathcal{B}_G$ as a fusion category of topological line operators supported on the $V$-boundary condition of the bulk TQFT. We will abusively conflate $\mathcal{B}_G$ with its image under $F_{\mathrm{B}\to\partial}$ in what follows.

Similarly, instead of thinking of the $g$-twisted module $V_a \in \Rep_G(V)$ as the Hilbert space of local operators which sit at the junction of $a$ and $L_g$, as in the left of Figure \ref{fig:Vareinterp}, one can instead think of $V_a$ as the Hilbert space of boundary local operators on which the line $F_{\mathrm{B}\to\partial}(a)$ can end, as in the right of Figure \ref{fig:Vareinterp}.

\begin{figure}
    \centering
    \input{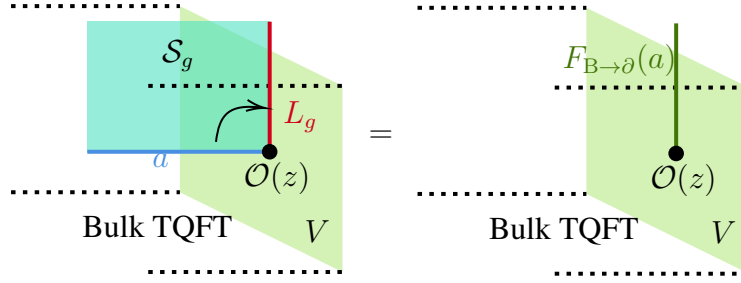}
    \caption{A (non-genuine) local operator $\mathcal{O}(z)\in V_a$ can be thought of either as living at the junction of $a$ and $L_g$ as in the left, or at the end of the boundary line $F_{\mathrm{B}\to\partial}(a)$ as on the right. }\label{fig:Vareinterp}
\end{figure}

Thus, we can rephrase the previous subsection as follows. Given a group $G$ of automorphisms of a chiral algebra $V$, thought of as a gapless chiral boundary of a bulk 3D TQFT, we can produce a category of \emph{boundary} topological line operators $\mathcal{B}_G$ which enriches the usual MTC $\mathcal{B}$ associated to $V$ in the sense that $\mathcal{B}_G$ forms a $G$-crossed braided extension of $\mathcal{B}$. 

It is then natural to ask when the reverse holds. That is, suppose one hands you some fusion category $\mathcal{C}$ of topological line operators supported on a rational chiral algebra boundary condition $V$ of a bulk TQFT. When can one think of $\mathcal{C}$ as a symmetry-enrichment of the Moore-Seiberg data of $V$? Our answer is \emph{always}, provided one is alright with replacing groups with hypergroups in general.

\begin{mainresult}
 Suppose that $V$ is a rational chiral algebra boundary condition of a semi-simple bulk 3D TQFT, and $\mathcal{C}$ is a fusion category of boundary topological line operators which contains $\mathcal{B}$ as a subcategory.\footnote{That is, we assume that $\mathcal{C}$ contains every line obtained by pushing a bulk anyon $a\in\mathcal{B}$ to the boundary.} Then $\mathcal{C}$ is a $K$-graded extension of $\mathcal{B}$,
\begin{align}\label{eqn:directsumdecomp}
    \mathcal{C} \cong \bigoplus_{r_i\in K}\mathcal{B}_{r_i}, \ \ \ \ \ \ \mathcal{B}_{r_0}\cong \mathcal{B},
\end{align}
where $K:=K_{\mathcal{C}}\sslash K_{\mathcal{B}}=\{r_0,\dots,r_{n-1}\}$. Here, $K_{\mathcal{C}}$ (resp.\ $K_{\mathcal{B}}$) is the hypergroup induced by the fusion ring of $\mathcal{C}$ (resp.\ $\mathcal{B}$), and $K_{\mathcal{C}}\sslash K_{\mathcal{B}}$ denotes the finite hypergroup of double cosets of $K_{\mathcal{B}}$ in $K_{\mathcal{C}}$ (see Appendix \ref{app:hypergroups} for background on hypergroups). Furthermore, $K$ acts on the chiral algebra $V$ via  bulk topological surfaces $\mathcal{S}_{r_i}$ terminating on boundary topological line junctions $L_{r_i}$, analogously to Figure \ref{fig:autV}.

\end{mainresult}

In short, this result says that any fusion category of boundary topological line operators defines a hypergroup symmetry-enrichment of the Moore-Seiberg data of $V$. Let us unpack the various ingredients which we have not yet defined. 

\paragraph{The definition of a hypergroup} First of all, we must say what a hypergroup is. A hypergroup is essentially a generalization of a fusion ring where the structure constants are allowed to be non-integers. In more detail, a finite hypergroup $K=\{r_0,\dots,r_{n-1}\}$ is a finte set equipped with an associative multiplication,
\begin{align}
    r_i\star r_j = \sum_{k}P_{ij}^k \cdot r_k, \ \ \ \ P_{ij}^k \in \mathbb{R}_{\geq 0},
\end{align}
which is further assumed to be ``stochastically normalized'',
\begin{align}
    \sum_k P_{ij}^k=1.
\end{align}
It is required to have a unit $r_0$, and every element $r_i$ must have a dual $r_{i^\ast}$. (See Appendix \ref{app:hypergroups} for the more precise definition.) When every structure constant $P_{ij}^k$ is either $0$ or $1$, the hypergroup describes an ordinary group-like symmetry. Otherwise, we say that the hypergroup is \emph{noninvertible}.

The prototypical example is the double cosets $G\sslash H$ of one finite group inside another, in which case the multiplication is defined to be 
\begin{align}\label{eqn:doublecosethypgpintro}
    [g]\star [g'] = \frac{1}{|H|}\sum_{h\in H}[ghg'].
\end{align}
When $H$ is normal in $G$, this construction reproduces the usual quotient group $G/H$. When $H$ is not normal, the quotient group doesn't exist, but the quotient hypergroup does. In a similar way, while the quotient of a fusion ring by a subring does not always produce another fusion ring, it does always produce a hypergroup. (Again, see Appendix \ref{app:hypergroups} for the detailed construction.) The hypergroup $K_{\mathcal{C}}\sslash K_{\mathcal{B}}$ appearing in the main result is constructed in this way from the fusion rings of $\mathcal{C}$ and $\mathcal{B}$.

\paragraph{The hypergroup-grading} Next, we must say what it means for a fusion category to be \emph{graded} by a hypergroup. In addition to requiring that $\mathcal{C}$ split into a direct sum over graded components $\mathcal{B}_{r_i}$ for each $r_i\in K$, as in Equation \eqref{eqn:directsumdecomp}, it is further required that, if $X\in\mathcal{B}_{r_i}$ and $Y\in \mathcal{B}_{r_j}$, then 
\begin{align}\label{eqn:hypgraded}
    \frac{\dim(\pi_k(X\otimes Y))}{\dim (X\otimes Y)} = P_{ij}^k,
\end{align}
where $\dim(Z)$ is the quantum dimension of $Z\in\mathcal{C}$ and $\pi_k:\mathcal{C}\to \mathcal{B}_{r_k}$ is the natural projection from $\mathcal{C}$ onto its $r_k$-graded component. In words, Equation \eqref{eqn:hypgraded} says that the fraction of the quantum dimension of $X\otimes Y$ which resides in the $r_k$-graded component is controlled by the structure constant $P_{ij}^k$ of the hypergroup. 

\paragraph{The hypergroup action} The condition that $\mathcal{C}$ is hypergroup graded in particular implies that each $\mathcal{B}_{r_i}$ transforms as a module category over $\mathcal{B}_{r_0}\cong \mathcal{B}$ (see e.g.\ \cite{Choi:2023xjw} for a description of module categories aimed at physicists). This is just a fancy way of saying that the fusion of a bulk anyon in $\mathcal{B}$ onto a boundary line in $\mathcal{B}_{r_i}$ produces another boundary line in $\mathcal{B}_{r_i}$. It is known that $\mathcal{B}$-module categories are in one-to-one correspondence with topological surfaces in the bulk TQFT. Hence, for each $r_i\in K$, we obtain an abstractly defined surface operator $\mathcal{S}_{r_i}$. 

The physical interpretation of this surface operator is as follows. When the boundary line $X\in \mathcal{B}_{r_i}\subset \mathcal{C}$ is dragged into the bulk, it is trailed by precisely the surface operator $\mathcal{S}_{r_i}$ which anchors to $V$ on some topological line junction $L_{r_i}$, similarly to the left of Figure \ref{fig:FBtopartial}. (We will define the line junction $L_{r_i}$ using a SymTFT construction in Section \ref{subsec:surfaces}, see Figure \ref{fig:pinch}.) In particular, $X$ can equivalently be thought of as a twisted sector bulk anyon $a$ whose worldline lives on the boundary of the surface $\mathcal{S}_{r_i}$.

These surfaces $\mathcal{S}_{r_i}$, as well as the corresponding boundary line junctions $L_{r_i}$, can be used to define the action of a hypergroup on $V$. Indeed,  we obtain for each $r_i\in K$ a corresponding map $\widehat{r_i}:V\to V$ from genuine boundary local operators to genuine boundary local operators. The operator $\widehat{r_i}(\mathcal{O})$ is obtained by collapsing onto $\mathcal{O}$ a ``dome'' built out of the surface $\mathcal{S}_{r_i}$, anchored on the boundary using $L_{r_i}$, as in the left of Figure \ref{fig:domehypergroup}. It turns out that the $\widehat{r_i}$ compose according to the multiplication law of the hypergroup $K$, up to overall normalization constants, as we explain in Section \ref{subsec:surfaces}. This explains why the hypergroup $K$ acts on the chiral algebra $V$.

\bigskip 

Our main result above illustrates the rich structure which follows once one has identified a fusion category of boundary topological line operators. A natural question is then what methods are available for finding such categories in practice. Our next main result, which builds on a number of previous papers \cite{Bischoff:2016jmy,Burbano:2021loy,Rie22,Rayhaun:2023pgc,Moller:2024xtt,Dong:2025ttr}, says that such fusion categories are in one-to-one correspondence with rational conformal subalgebras.

\begin{mainresult}
    Let $V$ be a rational chiral algebra, thought of as a gapless boundary condition of a 3D bulk topological order $\mathrm{TQFT}_V$. Fusion categories $\mathcal{C}\supset \mathcal{B}$ of boundary topological line operators are in one-to-one correspondence with rational conformal subalgebras $W\subset V$.\footnote{Of course, there will be many fusion categories $\mathcal{C}\supset \mathcal{B}$ which do not act on $V$ (i.e.\ which do not form categories of boundary topological line operators) and so do not correspond to conformal subalgebras of $V$.} The correspondence assigns 
    \begin{align}
    \begin{split}
        \mathcal{C}&\mapsto V^{\mathcal{C}}, \\
        W&\mapsto \Ver(V/W),
    \end{split}
    \end{align}
    where $V^{\mathcal{C}}$ is the conformal subalgebra of $V$ consisting of boundary local operators which are neutral under all the lines in $\mathcal{C}$ (see Figure \ref{fig:VCdef}), and $\Ver(V/W)$ is the subcategory of $\mathcal{C}$ consisting of all lines with respect to which the operators in $W$ are neutral. 
    
    \hspace{.3in}Given a conformal embedding $W\subset V$, one can more explicitly construct $\Ver(V/W)$ as the category of topological line operators supported on a particular half-space gauging interface between $\mathrm{TQFT}_W$ and $\mathrm{TQFT}_V$, where $\mathrm{TQFT}_W$ is the 3D topological order which supports $W$ on its boundary, see Figure \ref{fig:SymTFTconformalembedding}.
\end{mainresult}

The basic observation which underpins this result is that a conformal subalgebra $W\subset V$ can always be thought of as furnishing a ``physical'' boundary condition for a SymTFT-style construction, Figure \ref{fig:SymTFTconformalembedding}. Its utility is that it allows for a reinterpretation of a vast literature on conformal embeddings in the language of generalized global symmetries. Indeed, this principle will be behind most of the examples we consider in Section \ref{sec:examples}.

\begin{figure}
    
\begin{center}
\input{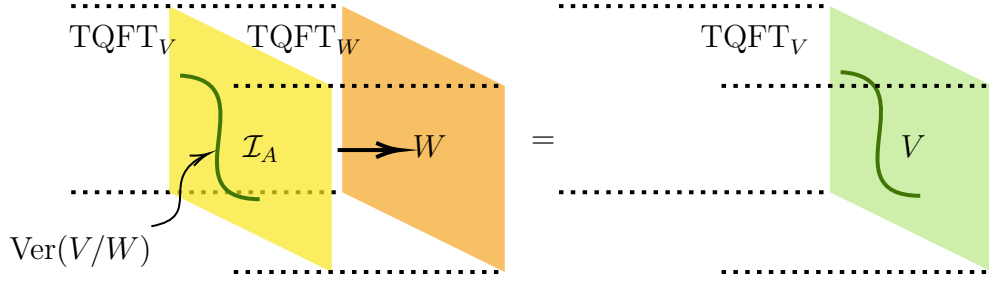}
\caption{The SymTFT picture of a conformal embedding $W\subset V$.}\label{fig:SymTFTconformalembedding}
\end{center}
\end{figure}

\subsection{Implications for absolute CFTs}

So far, the discussion has been about generalized symmetries of relative QFTs, especially chiral algebras. Do our results have any utility for someone who is mainly interested in absolute QFTs? Indeed, we will now sketch how symmetries of chiral algebras distill invaluable information about symmetries, and even boundary conditions, of absolute CFTs which are built on top of those chiral algebras.

To understand the idea, recall that one of the basic philosophies which governs the study of rational conformal field theory is that any exercise can be divided into two parts: a ``holomorphic'' part in which one works at the level of the chiral algebra, and a ``topological'' part in which one works at the level of the bulk 3D TQFT. 

For instance, suppose one is interested in constructing a rational CFT. In the ``holomorphic'' step, one selects a rational chiral algebra $V$ to serve as a subsector of the holomorphic local operators in the theory. In the ``topological'' step, one selects a topological surface operator $\mathcal{S}$ in the bulk 3D TQFT which supports $V$ on its boundary. This data $(V,\mathcal{S})$ is then fed into the Fuchs-Runkel-Schweigert/Kapustin-Saulina construction \cite{Fuchs:2002cm,Kapustin:2010if} to produce an absolute theory we call $\mathrm{CFT}_{V,\mathcal{S}}$. Specifically, one compactifies the 3D TQFT on an interval with $V$ and $\overline{V}$ imposed at the two boundaries,  and with the topological surface $\mathcal{S}$ placed at the midpoint of the interval, see Figure \ref{fig:KS}. This construction produces every rational CFT whose maximal chiral algebra contains $V$ as a conformal subalgebra. The different ways of gluing left-movers to right-movers are parametrized by the choice of surface operator $\mathcal{S}$.

\begin{figure}
    \begin{center}
        \input{Figures/KS}
        \caption{The definition of the 2D theory $\mathrm{CFT}_{V,\mathcal{S}}$.}\label{fig:KS}
    \end{center}
\end{figure}

\paragraph{Symmetries of full CFTs}
Suppose now that one wants to construct a category of topological line operators of the theory $\mathrm{CFT}_{V,\mathcal{S}}$. Our claim is that this can again be achieved in two steps. In the ``holomorphic'' step, one first selects a category $\mathcal{C}_L$ of topological line operators supported on the left-moving  chiral algebra. (For example, this category $\mathcal{C}_L$ could come from a conformal subalgebra $W\subset V$, as in Figure \ref{fig:SymTFTconformalembedding}.) Similarly, one selects a category $\mathcal{C}_R$ of topological line operators supported on the right-moving chiral algebra. In the ``topological'' step, one glues together these two categories to produce a symmetry of the full theory. 

Our next main result describes how precisely to achieve this gluing. To this end, let $\mathbf{L}(\mathcal{S})$ denote the category of topological line operators supported on the surface $\mathcal{S}$. The category $\mathbf{L}(\mathcal{S})$ consists precisely of the symmetries of $\mathrm{CFT}_{V,\mathcal{S}}$ which preserve the left- and right-moving chiral algebras $V\otimes \overline{V}$. It admits the structure of a $(\mathcal{B},\mathcal{B})$-bimodule category due to the fact that one can fuse bulk anyons from both the left and the right onto any line in $\mathbf{L}(\mathcal{S})$ (cf.\ Figure \ref{fig:CDbimodule}). For a similar reason, $\mathcal{C}_L$ and $\mathcal{C}_R$ are $\mathcal{B}$-module categories.

\begin{mainresult}{}{}
Suppose $\mathcal{C}_L$ and $\mathcal{C}_R$ are categories of topological line operators supported on the left- and right-moving chiral algebras of $\mathrm{CFT}_{V,\mathcal{S}}$, respectively.\footnote{We assume as always that $\mathcal{C}_L$ contains the bulk anyons $\mathcal{B}$ as a subcategory, and similarly for $\mathcal{C}_R$.} Then 
\begin{align}\label{eqn:gluingformula}
    \mathcal{C}_L\boxtimes_{\mathcal{B}}\mathbf{L}(\mathcal{S})\boxtimes_{\mathcal{B}}\mathcal{C}_R^{\vee}
\end{align}
furnishes a category of topological line operators of the full theory $\mathrm{CFT}_{V,\mathcal{S}}$, where $\boxtimes_{\mathcal{B}}$ denotes the relative Deligne product over $\mathcal{B}$, and $\mathcal{C}_R^{\vee}$ is the dual category to $\mathcal{C}_R$. In particular, if $\mathcal{C}_L$ (resp.\ $\mathcal{C}_R$) is the \emph{complete} category of topological line operators supported on the left-moving (resp.\ right-moving) chiral algebra, then Equation \eqref{eqn:gluingformula} captures the complete category of topological line operators of the full theory $\mathrm{CFT}_{V,\mathcal{S}}$.
\end{mainresult}

We give an extensive mathematical and physical discussion of the relative Deligne product in Section \ref{sec:gluing}. Readers unfamiliar with the precise definition can blackbox Equation \eqref{eqn:gluingformula} as providing a formula for gluing together symmetries of the left- and right-moving chiral algebras to produce symmetries of the full rational CFT. We emphasize that the \emph{relative} Deligne product is a more subtle construction than the standard Deligne product: in particular, the symmetries of a CFT are \emph{not} simply a ``direct product'' of the symmetries supported on its left- and right-moving chiral algebras.

As a simple example, assume that $\mathcal{S}$ is the trivial topological surface, in which case $\mathrm{CFT}_V\equiv \mathrm{CFT}_{V,\mathcal{S}=\mathds{1}}$ is the canonical diagonal rational CFT built on $V$, and $\mathbf{L}(\mathcal{S}=\mathds{1})\cong \mathcal{B}$ is simply the category of anyons in the bulk. Further assume that $\mathcal{C}_L\cong\mathcal{C}_R\cong  \mathcal{B}$, i.e.\ we take $\mathcal{C}_L$ and $\mathcal{C}_R$ to consist just of the boundary topological lines obtained by pushing bulk anyons onto the left- and right-moving chiral algebras, respectively. Then, because $\mathcal{B}$ behaves as the identity with respect to $\boxtimes_{\mathcal{B}}$, the category in Equation \eqref{eqn:gluingformula} reduces to 
\begin{align}
    \mathcal{B}\boxtimes_{\mathcal{B}}\mathcal{B}\boxtimes_{\mathcal{B}}\mathcal{B}^{\vee}\cong\mathcal{B}\cong\Rep(V),
\end{align}
which recovers the standard category of Verlinde lines \cite{Verlinde:1988sn} of $\mathrm{CFT}_V$. We emphasize that even the diagonal theory $\mathrm{CFT}_V$ has symmetries beyond the Verlinde lines $\mathcal{B}\cong \Rep(V)$, simply because the chiral algebra $V$ generically possesses topological line operators beyond those coming from the bulk anyons. (For example, this happens whenever $V$ possesses a non-trivial conformal subalgebra.) We will see this explicitly in Section \ref{subsec:SU(3)1} in the context of the $SU(3)_1$ Wess-Zumino-Witten model.

\paragraph{Boundaries of full CFTs}

Topological line operators of a chiral algebra $V$ are also closely related to boundary conditions of $\mathrm{CFT}_{V,\mathcal{S}}$. Let us explain how this works when $\mathcal{S}=\mathds{1}$, i.e.\ when we are working with the canonical diagonal theory $\mathrm{CFT}_V$ built on $V$. The more general case will be treated in Section \ref{subsec:boundaries}.

\begin{mainresult}{}{}
Unitary topological line operators $X \in \Ver(V/W)$ (i.e.\ boundary topological line operators supported on the chiral algebra $V$ which preserve a conformal subalgebra $W$) are in one-to-one correspondence with unitary boundary conditions of the diagonal theory $\mathrm{CFT}_V$ which preserve the subalgebra $W$. In particular, taking $W$ to be the Virasoro algebra, we learn that every unitary boundary condition of $\mathrm{CFT}_V$ corresponds to some unitary topological line operator supported on the chiral algebra $V$, and vice versa.
\end{mainresult}

While this result can be proved by abstract nonsense, the physical intuition comes from Figure \ref{fig:squashboundary}. The basic idea is that one can place the 3D TQFT on a solid ball $B^{(3)}$ with the chiral algebra boundary condition $V$ imposed on $\partial B^{(2)}=S^2$. If one wraps a boundary topological line operator $X$ along the equator of $\partial B^{(2)}=S^2$, then, after squashing the ball into a disk, one engineers a boundary condition $\partial X$ of the diagonal theory $\mathrm{CFT}_V$. For example, in the special case that $X=\mathds{1}$ is taken to be the identity topological line, $\partial X$ engineers the  Cardy boundary condition corresponding to the identity primary \cite{Cardy:1989ir}. The squashing map $X\mapsto \partial X$ turns out to be a bijection. While Figure \ref{fig:squashboundary} already appeared in \cite{Kapustin:2010if}, op.\ cit.\ only considered topological lines and boundary conditions which preserve the chiral algebra $V$. The novelty of our result is that it extends this correspondence to cover \emph{all} topological lines of $V$, and \emph{all} boundary conditions of $\mathrm{CFT}_V$, including those which break the chiral algebra down to just Virasoro.

\begin{figure}
    \centering
    \input{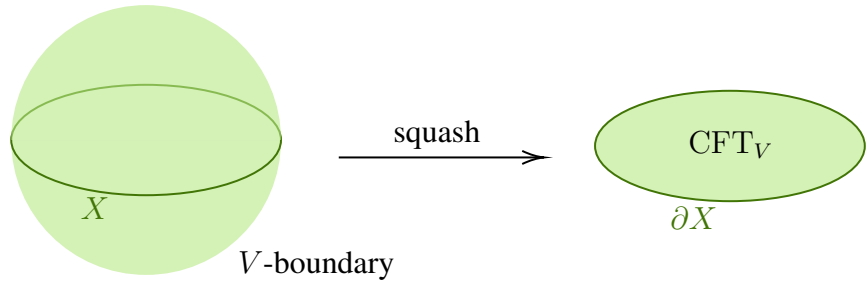}
    \caption{Left: the bulk 3D TQFT on a solid ball $B^{(3)}$ with a chiral algebra $V$ imposed as a boundary condition on $\partial B^{(2)}=S^2$ and a boundary topological line operator $X$ wrapped on the equator. Right: after squashing, this produces a boundary condition $\partial X$ of the diagonal theory $\mathrm{CFT}_V$ built on $V$.}\label{fig:squashboundary}
\end{figure}

This result has several interesting structural corollaries. For example, it naturally leads to the conclusion that, given any boundary condition $b$ of $\mathrm{CFT}_V$, any other boundary condition $b'$ occurs as a summand of $\mathcal{L}\otimes b$ for some topological line defect $\mathcal{L}$ of $\mathrm{CFT}_V$, 
\begin{align}\label{eqn:Lbbp}
   \mathcal{L}\otimes b = b'\oplus \cdots ,
\end{align}
where $\mathcal{L}\otimes b$ denotes the boundary condition obtained by fusing $\mathcal{L}$ onto $b$. One might say that the topological lines of $\mathrm{CFT}_V$ act transitively on its boundary conditions. In particular, this shows that the ``phantom symmetry'' mechanism described in \cite{Antinucci:2025uvj} can always be used to explain the existence of an interface conformal manifold, at least in diagonal rational CFTs.

It is not too hard to convince oneself that, when one takes $b=\partial\mathds{1}$ to be the identity Cardy boundary condition, then one can choose $\mathcal{L}$ in Equation \eqref{eqn:Lbbp} so that the $\cdots$ disappear, i.e.\ so that $b'=\mathcal{L}\otimes \partial \mathds{1}$. Since the quantum dimension of a unitary topological line is bounded from below by $1$, this implies a bound on the $g$-function of any boundary condition in a diagonal rational CFT, 
\begin{align}\label{eqn:gfunctionbound}
    g_{b'} = \dim(\mathcal{L}) g_{\partial\mathds{1}} \geq \sqrt{S_{11}},
\end{align}
where $S_{ab}$ is the modular S-matrix of the chiral algebra $V$. In particular, the identity Cardy boundary condition has the lowest $g$-function amongst all boundary conditions of $\mathrm{CFT}_V$. This bound, Equation \eqref{eqn:gfunctionbound}, was observed numerically for stable\footnote{A boundary condition of a 2D CFT is stable if it does not support any boundary local operators with dimension $\Delta<1$, aside from the identity.} boundary conditions in several diagonal rational CFTs in \cite{Collier:2021ngi}. Our results, in addition to showing that the word ``stable'' can be dropped, provides a structural explanation for its existence.

\subsection{Outline}

The remainder of this article is organized as follows. 

In Section \ref{sec:relativesymmetries}, we develop our approach to generalized symmetries of chiral algebras. In Section \ref{subsec:gaplessboundaries}, we summarize some basic facts about how one should think about rational chiral algebras as gapless boundary conditions of 3D TQFTs. In Section \ref{subsec:boundarylines}, we describe a SymTFT approach to studying categories of topological line operators supported on these gapless boundaries. In particular, we explain symmetry/subalgebra duality, a correspondence between the symmetries of a chiral algebra and its conformal subalgebras. In Section \ref{subsec:twisteddome}, we describe the twisted representations of a chiral algebra enriched with a noninvertible symmetry, and observe that they are acted on by a generalization of the tube algebra known as the dome algebra. In Section \ref{subsec:hypergroup}, we isolate the effective hypergroup from the dome algebra, and in Section \ref{subsec:surfaces}, we explain how it is related to topological surfaces in the bulk. In Section \ref{subsec:topman}, we survey the various topological manipulations available in the relative setting, and derive how symmetries mutate upon performing them. In Section \ref{subsec:symmetryresolvedpf}, we sketch how twisted genus-1 partition functions are defined. In Section \ref{subsec:Galois}, we explain how noninvertible symmetries round out an analog of Galois theory for chiral algebras. Finally, in Section \ref{subsec:computationalaids}, we provide various practical tools for calculating noninvertible symmetries of chiral algebras.

Section \ref{sec:gluing} is dedicated to spelling out the implications of the results of Section \ref{sec:relativesymmetries} for absolute CFTs. In particular, in Sections \ref{subsec:settingup}--\ref{subsec:relativedeligne}, we describe how to glue together symmetries of chiral algebras to produce symmetries of full CFTs with both left- and right-movers. In Section \ref{subsec:boundaries}, we explain the relationship between the boundary conditions of a full CFT and the symmetries of its chiral algebra.

Section \ref{sec:examples} illustrates our results in examples. We orient ourselves in Section \ref{subsec:cleft} by treating invertible symmetries of $c=1$ chiral boson theories. In Section \ref{subsec:nonchiralexample}, we show that our formalism applies equally well to non-chiral relative theories by treating the example of the $\mathbb{Z}_2$-even sector of the Ising CFT. In Section \ref{subsec:haagerup}, we provide a family of noninvertible hypergroup actions which are closely related to the Haagerup fusion categories. In Section \ref{subsec:G21}, we derive the boundary conditions of the $(G_2)_1$ WZW model which preserve its $\widehat{\mathfrak{su}}(2)_{28}$ chiral subalgebra. In Section \ref{subsec:SU(3)1}, we glue together symmetries of the $\widehat{\mathfrak{su}}(3)_1$ chiral algebra to produce symmetries of the full $SU(3)_1$ WZW model. Finally, in Section \ref{subsec:Heisenberg}, we calculate the full, infinite category of topological line operators of the $\widehat{\mathfrak{u}}(1)$ Kac--Moody algebra, foreshadowing an extension of our machinery to the irrational/infinite setting.

We conclude in Section \ref{sec:future} with suggestions for future research. Appendix \ref{app:hypergroups} contains a review of hypergroups, and Appendix \ref{app:E8} contains some data related to the example in Section \ref{subsec:G21}.

\paragraph{Note added:} While this manuscript was in its final stages of preparation, the two papers \cite{Bottini:2026edf,Eck:2026aiz} appeared on the arXiv, which have some overlap with our work.

\section{Symmetries of Relative CFTs and Chiral Algebras}\label{sec:relativesymmetries}

In this section, we describe a number of ideas related to chiral algebras and, more generally, relative 2D CFTs. Specifically, we focus on developing an effective theory for working with their noninvertible symmetries. 

For ease of exposition, our presentation is geared mainly towards chiral algebras --- known to mathematicians as vertex operator algebras --- which can be thought of as relative CFTs with vanishing right-moving central charge. However, most of what we say extends straightforwardly to relative CFTs with both left- and right-movers, as we will demonstrate by way of an example in Section \ref{subsec:nonchiralexample}.\footnote{We expect that relative CFTs with both left- and right-movers can be incorporated into the theory in a mathematically rigorous way by invoking full field algebras or full vertex algebras, see e.g.\ \cite{Huang:2005gz,Moriwaki:2020cxf,Moriwaki:2020dlj,Adamo:2024etu}.}

Throughout this section, we assume that the symmetries we study are finite and, relatedly, that the chiral algebras we work with are rational. This is mainly so that we can invoke powerful theorems about fusion categories and semisimple modular tensor categories. However, we expect that many of the statements we make will generalize when these assumptions are dropped. 
As a teaser, we present the computation of the (infinitely many) topological lines of the (irrational) $\widehat{\mathfrak{u}}(1)$ Kac-Moody algebra (a.k.a.\ the Heisenberg VOA) in Section \ref{subsec:Heisenberg}.

\subsection{Chiral algebras as gapless boundaries}\label{subsec:gaplessboundaries}

We begin by recalling a physicist's definition of a chiral algebra \cite{Zamolodchikov:1985wn,Moore:1988qv} and briefly giving familiar examples.

\begin{definition}
    A chiral algebra $V$ is the structure formed by the space of holomorphic local operators in a 2D conformal field theory $\mathcal{H}$,
\begin{align}
    V=\{\mathcal{O}(z,\bar z)\in\mathcal{H} \mid \partial_{\bar z}\mathcal{O} = 0\}.
\end{align} 
\begin{example}{}{}
    The chiral algebra of the Ising CFT is the (simple quotient of the) Virasoro vertex operator algebra at central charge $c=\sfrac12$. 
\end{example}
\begin{example}{}{latticeVOA}
    The chiral algebra of the $c=1$ compact free boson theory with radius $R^2=2p/q$ (in conventions where $R^2=2$ is the self T-dual radius) is $V_{2pq}$, where $V_{2m}$ is generated by a chiral boson $\partial_z \phi$ and vertex operators $e^{i\lambda \phi(z)}$ with momenta in the charge lattice $\lambda \in \sqrt{2m}\mathbb{Z}.$ 
\end{example}
\end{definition}
For a physicist, a rational conformal field theory is one with finitely many primary operators with respect to its maximal chiral algebra $V$. When a conformal field theory is rational, its chiral algebra is expected to be described mathematically by a \emph{strongly rational vertex operator algebra} (see e.g.\ \cite{Creutzig:2016fms} for a mathematical review). We use the terms ``rational chiral algebra'' and ``strongly rational vertex operator algebra'' interchangeably.

Strongly rational vertex operator algebras have a very tame representation theory. For example, any representation of a strongly rational vertex operator algebra can be decomposed into a direct sum of irreducible representations. In fact, more is true: 
\begin{theorem}[Proved in \cite{Moore:1988qv,Huang:2005gs}]
    The representation category $\Rep(V)$ of a (not necessarily unitary) strongly rational vertex operator algebra $V$ admits the structure of a modular tensor category.
\end{theorem} 
\noindent When $V$ is further unitary, its representation category is a \emph{unitary} modular tensor category \cite{Gui:2017hrk,Gui:2017xnu}, though the converse is less clear. Namely, it is not known whether $\Rep(V)$ being a unitary modular tensor category implies that $V$ itself is unitary. 

All vertex operator algebras are assumed to be unitary in this paper unless otherwise stated (see e.g.\ \cite{dong2014unitary,Carpi:2015fga,Gui:2017hrk,Gui:2017xnu} for the formal definition). Much of what follows holds if we merely assume that $\Rep(V)$ is pseudounitary (i.e.\ that all the quantum dimensions are positive). Pseudounitarity of $\Rep(V)$ is guaranteed provided the conformal weights $h_M$ of all non-vacuum $V$-modules are positive. However, since it is very difficult to find any examples where $\Rep(V)$ is pseudounitary without $V$ being unitary, we will assume unitarity for simplicity.

Rational chiral algebras are closely related to gapless chiral boundary conditions of 3D topological field theories. For example, the current algebra $\mathfrak{g}_k$ with $\mathfrak{g}$ a simple complex Lie algebra and  $k$ a positive integer (i.e.\ the chiral algebra of a  Wess--Zumino--Witten model \cite{Wess:1971yu,Witten:1983tw,Witten:1983ar}), can be thought of as living on the boundary of Chern--Simons theory with gauge group $G$  and level $k$ \cite{Witten:1988hf,Elitzur:1989nr}, where $G$ is the simply-connected Lie group whose complexified Lie algebra is $\mathfrak{g}$. As such, we can think of chiral algebras as giving tractable examples of relative quantum field theories in two dimensions \cite{Freed:2012bs}. Much of what we say about chiral algebras in this paper can be extended to that more general setting with a few modifications. In fact, we will study an example of a non-chiral relative conformal field theory (the $\mathbb{Z}_2$-even sector of the Ising CFT) in detail in Section \ref{subsec:nonchiralexample}.

To describe the relationship between chiral algebras and boundaries more precisely, recall that the defining data of a 3D topological field theory is a tuple $(\mathcal{B},c)$ consisting of a unitary modular tensor category $\mathcal{B}$ which describes the algebraic properties of its anyons/topological line operators, and a rational number $c$ which captures its chiral central charge. (Sometimes we abbreviate $(\mathcal{B},c)$ to just $\mathcal{B}$ to avoid clutter.) 

Recall also that any automorphism $g\in\mathrm{Aut}(V)$ permutes around the representations of $V$ (see e.g.\ the discussion around \cite[Equation (2.20)]{Dong:1994wn}) and in fact induces a ribbon auto-equivalence 
\begin{align}\label{eqn:inducedautoequiv}
    g_\ast :\Rep(V)\to\Rep(V)
\end{align} 
of the representation category of $V$ \cite{Mcrae:2019pol}. However, importantly, not every ribbon auto-equivalence of $\Rep(V)$ arises in this way. A simple example is the $\widehat{\mathfrak{su}}(2)_k$ current algebra with $k=2~\mathrm{mod}~4$ and $k\geq 6$. Since $\mathrm{Aut}(\widehat{\mathfrak{su}}(2)_k)=\textsl{PSL}_2(\mathbb{C})$, every automorphism is continuously connected to the identity and hence induces the trivial ribbon auto-equivalence of $\Rep(\widehat{\mathfrak{su}}(2)_k)$.\footnote{The subgroup of automorphisms which preserve the unitary structure of $\widehat{\mathfrak{su}}(2)_k$ is $\mathrm{Aut}^\dagger (\widehat{\mathfrak{su}}(2)_k)\cong SO(3)$.} On the other hand, when $k=2~\mathrm{mod}~4$ and $k\geq 6$, it is known that $\Rep(\widehat{\mathfrak{su}}(2)_k)$ possesses a $\mathbb{Z}_2$ worth of ribbon autoequivalences \cite{Kirillov:2001ti,Edie-Michell:2022abq}. For a physicist, this follows from the ADE classification of CFTs built on the $\widehat{\mathfrak{su}}(2)_k$ chiral algebras \cite{Cappelli:1987xt}: precisely when $k=2~\mathrm{mod}~4$ and $k\geq 6$, there is a D-series CFT whose Hilbert space is constructed by gluing left-movers to right-movers using this auto-equivalence. 

Using similar logic, any isomorphism $f:V\to V'$ induces an equivalence $f_\ast:\Rep(V')\to \Rep(V)$ of representation categories.

We have the following expectation regarding boundaries of $(\mathcal{B},c)$. (We refer readers to \cite{Kong:2017etd,Kong:2019byq,Kong:2019cuu} for a somewhat different but closely related perspective.) 

\begin{figure}
\begin{center}
\input{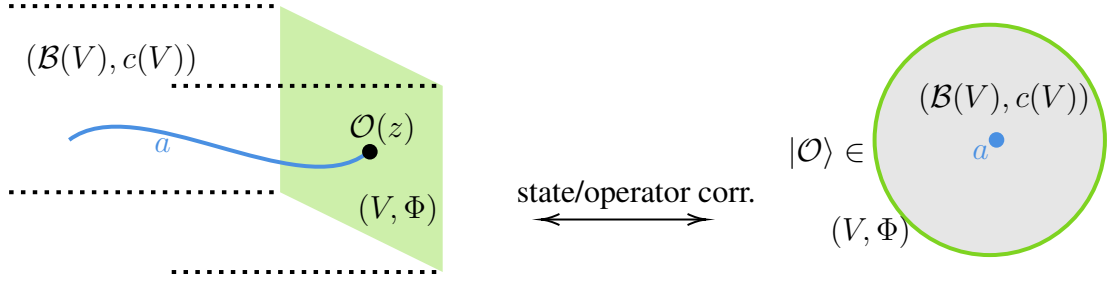}
\caption{A vertex operator algebra $V$ and a ribbon equivalence $\Phi:\mathcal{B}(V)\to \Rep(V)$ together define a gapless chiral boundary condition of the 3D topological field theory $(\mathcal{B}(V),c(V))$. The local operators $\mathcal{O}(z)$ living on the boundary and attached to a bulk line $a\in\mathcal{B}(V)$ belong to the $V$-module $V_a:=\Phi(a)$. By the state/operator correspondence, $V_a$ can equivalently be described as the Hilbert space that the 3D TQFT $(\mathcal{B}(V),c(V))$ associates to a punctured disk.}\label{fig:gaplesschiralboundary}
\end{center}
\end{figure}

\begin{claim}{}{claim:(B,c)boundaries}
Let $V$ be a unitary strongly rational vertex operator algebra of central charge $c$, and $\Phi:\mathcal{B}\to \Rep(V)$ a ribbon equivalence. Then every pair $(V,\Phi)$ gives rise to a gapless chiral boundary condition of the 3D TQFT $(\mathcal{B},c)$. Two boundaries $(V,\Phi)$ and $(V',\Phi')$ are isomorphic if there is an isomorphism $f:V\to V'$ satisfying $f_\ast\circ \Phi'= \Phi$.
\end{claim}

The physical interpretation of the boundary $(V,\Phi)$ is as follows. The chiral algebra $V$ is the Hilbert space of genuine boundary local operators. The map $\Phi$ associates to a topological line $a\in\mathcal{B}$ the Hilbert space $V_a:=\Phi(a)$ of boundary local operators $\mathcal{O}(z)$ on which $a$ can terminate. Alternatively, by the state/operator correspondence, $V_a:=\Phi(a)$ is the Hilbert space obtained by quantizing the 3D topological field theory on a two-dimensional disk with the line operator $a$ piercing the origin. See Figure \ref{fig:gaplesschiralboundary}. The space $V_{\mathds{1}}=\Phi(\mathds{1})$ corresponding to the identity line is simply the vertex operator algebra itself, and more generally $V_a$ is always a $V$-module, which is irreducible if and only if the input line $a$ is simple.

We note that there are other gapless boundary conditions of $(\mathcal{B},c)$ which are not of this kind, but we will not have any use for them.

To better appreciate the role of the tensor equivalence $\Phi$, we offer the following observation on how chiral boundary conditions interact with zero-form symmetries. Recall that ribbon auto-equivalences $F:\mathcal{B}\to\mathcal{B}$ can be thought of as invertible topological surface operators in the bulk 3D topological field theory described by $(\mathcal{B},c)$. It is straightforward to convince oneself that the action (by parallel fusion) of the surface described by $F$ on the boundary condition described by $(V,\Phi)$ is simply 
\begin{align}
    F \otimes (V,\Phi) \cong (V,\Phi\circ F^{-1}).
\end{align}
In particular, the $(V,\Phi)$ with $V$ fixed are transitively permuted amongst themselves via the action of invertible topological surfaces in the bulk. See Figure \ref{fig:braidedauts}.

\begin{figure}
    \begin{center}
\input{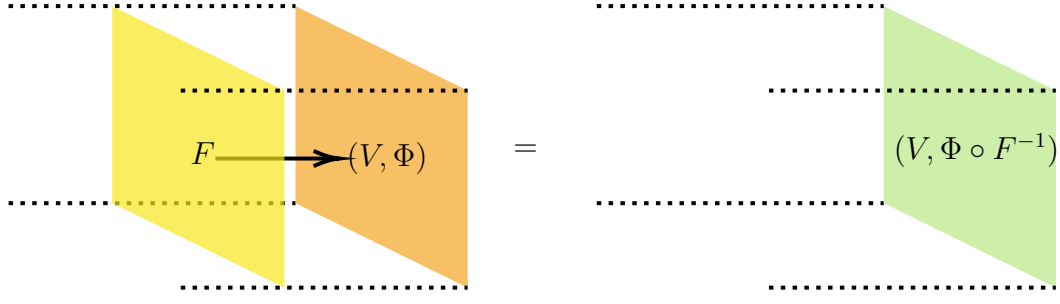}
\caption{The action of a bulk topological surface defined by ribbon auto-equivalence $F:\mathcal{B}\to\mathcal{B}$ on a gapless chiral boundary condition $(V,\Phi)$.}\label{fig:braidedauts}
    \end{center}
\end{figure}

To explain why the boundaries $(V,\Phi)$ and $(V',\Phi')$ are isomorphic when $f_\ast \circ \Phi'=\Phi$, note that an isomorphism $f:V\to V'$ is implemented physically by a boundary topological line operator on which the bulk surface 
\begin{align}
    \mathcal{S}_f:=(\Phi')^{-1}\circ f_\ast^{-1}\circ \Phi
\end{align} terminates, as in Figure \ref{fig:identifyboundaries}. By bending the surface $\mathcal{S}_f$ so that it fuses onto half of the boundary, we obtain an invertible topological line interface $I$ between the boundaries $(V,f_\ast\circ \Phi')=(V,\Phi)$ and $(V',\Phi')$. The existence of an invertible interface between two defects/boundaries in a quantum field theory is tantamount to saying that they are isomorphic.

\begin{figure}
    \begin{center}
        \input{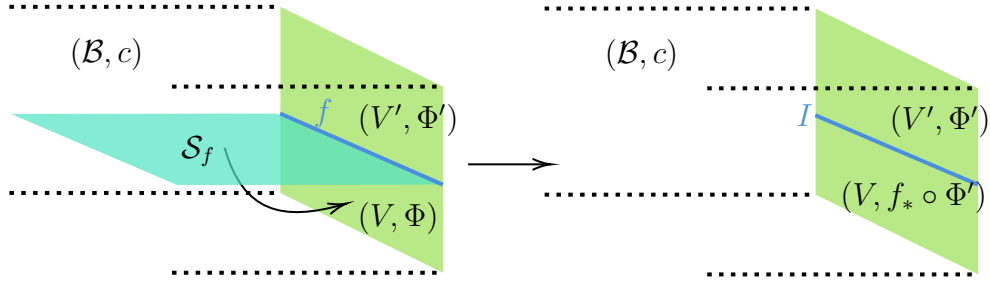}
        \caption{An isomorphism $f:V\to V'$ of chiral algebras is interpreted as a bulk surface $\mathcal{S}_f$ terminating on a boundary topological line interface $f$, as on the left. Bending the bulk surface so that it fuses onto half of the boundary produces an invertible topological line interface $I$ between $(V,f_\ast\circ \Phi')$ and $(V',\Phi')$.}\label{fig:identifyboundaries}
    \end{center}
\end{figure}

Most of the discussion which follows is insensitive to the particular choice of $\Phi$ we make, and so we will fix it at the outset. Accordingly, we will often abbreviate the boundary condition $(V,\Phi)$ to just $V$, or refer to it as a ``$V$-boundary'' when the precise choice of $\Phi$ does not matter. When we would like to emphasize that $\Phi$ is a tensor equivalence associated with the vertex operator algebra $V$ (as opposed to some other vertex operator algebra $W$), we will add $V$ as a decoration and write $\Phi_V$. Similarly, we will write $\mathcal{B}(V)$ when we want to emphasize that the modular tensor category is equivalent to the representation category $\Rep(V)$ of $V$. 

\subsection{Boundary topological line operators}\label{subsec:boundarylines}

We are interested in probing the category of topological line operators supported on a gapless chiral boundary condition $(V,\Phi_V)$. Unlike  gapped boundary conditions, which are home to only finitely many simple topological line operators, we will see that gapless chiral boundaries can (and generically do) support infinitely many simple lines. 

Throughout this paper, we will take for granted the basic expectation from physics that any finite collection of elementary topological line operators which are confined to a two-dimensional locus of spacetime and which close under parallel fusion generate a spherical fusion category.

\begin{figure}
\begin{center}
\input{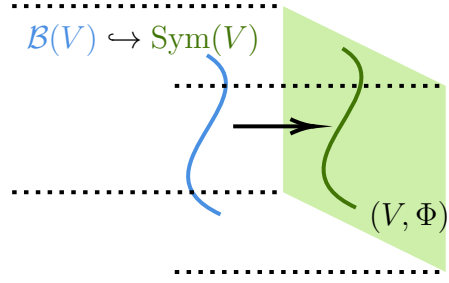}
\caption{The category $\mathcal{B}(V)$ of topological line operators in the bulk is a fusion subcategory of the category $\mathrm{Sym}(V)$ of topological line operators on the boundary defined by $(V,\Phi)$.}\label{fig:B(V)Sym(V)}
\end{center}
\end{figure}

At a minimum, we always have topological line operators which come from pushing bulk lines in $\mathcal{B}(V)$ onto the boundary. Unlike the case of gapped boundary conditions, where simple bulk lines can split into a direct sum of several simple lines when they are brought to the boundary, the lines in $\mathcal{B}(V)$ remain simple when they are pushed onto $(V,\Phi_V)$. In particular, the bulk-to-boundary map for chiral algebra boundaries of the kind we are discussing is fully faithful.\footnote{A functor $F:\mathcal{C}\to\mathcal{D}$ is fully faithful if each of the induced maps on hom-spaces $\mathrm{Hom}_{\mathcal{C}}(X,Y)\to \mathrm{Hom}_{\mathcal{D}}(F(X),F(Y))$ is an isomorphism.} Thus, we have the following.
\begin{claim}{}{}
The complete (generally infinite) category $\mathrm{Sym}(V)$ of topological line operators supported on the boundary   $(V,\Phi_V)$ contains $\mathcal{B}(V)$ as a fusion subcategory.
\end{claim}
\noindent See Figure \ref{fig:B(V)Sym(V)}.

We may consider how a line $b$ in $\mathcal{B}(V)$ acts on boundary local operators attached to a line $a$ in $\mathcal{B}(V)$, i.e.\ how $b$ acts on local operators in the Hilbert space $V_a=\Phi_V(a)$. 
We obtain an operator $\widehat{b}:V_a\to V_a$ by encircling the topological line $b$ around boundary local operators (possibly attached to a bulk line $a$). See Figure \ref{fig:encirclingaction}. Since $b$ can be dragged back into the bulk, we can straightforwardly calculate the action of $\widehat{b}$ to be 
\begin{align}\label{eqn:actionpullingbulk}
    \widehat{b}\cdot \mathcal{O} = \frac{S_{ab}^V}{S_{1a}^V} \mathcal{O}, \hspace{.3in}\text{ for all }\mathcal{O}\in V_a,
\end{align}
where $S^V_{ab}$ is the modular S-matrix of $\mathcal{B}(V)$. Here, we have used the standard fact that open Hopf links can be evaluated in terms of the S-matrix, see e.g.\ \cite[II.1.5.a]{turaev2016quantum} or \cite[Equation (40)]{Barkeshli:2014cna}.
As one might have anticipated, this action takes the same form as the action of Verlinde lines on local operators in a diagonal rational conformal field theory \cite{Verlinde:1988sn}.

\begin{figure}
\begin{center}
\input{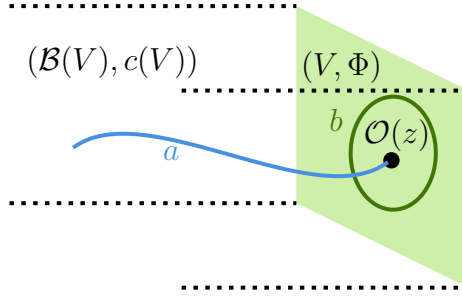}
\caption{The action of a boundary topological line operator $b$ on an operator $\mathcal{O}(z)$ in the Hilbert space $V_a$ is obtained by shrinking $b$ onto $\mathcal{O}(z)$.}\label{fig:encirclingaction}
\end{center}
\end{figure}

Our central tool for discovering boundary topological line operators beyond those coming from $\mathcal{B}(V)$ is a concept which was called symmetry/subalgebra duality in \cite{Rayhaun:2023pgc,Moller:2024xtt} (though some antecedents appear in \cite{Bischoff:2016jmy,Burbano:2021loy,Rie22}). Recall that a conformal subalgebra $W\subset V$ is a subalgebra whose stress tensor coincides with that of $V$. Conversely, we say that $V$ is a conformal extension of $W$.

\begin{claim}{Symmetry/subalgebra duality
}{claim:sym/subduality}
Suppose $V$ is a unitary rational chiral algebra. Unitary rational conformal subalgebras $W\subset V$ are in one-to-one correspondence with unitary fusion categories $\mathcal{C}$ satisfying $\mathcal{B}(V)\subset \mathcal{C}\subset\mathrm{Sym}^\dagger(V)$, where $\mathrm{Sym}^\dagger(V)$ is the complete category of unitary topological line operators on $V$.
\end{claim}

\noindent We expect that symmetry/subalgebra duality continues to hold when $W$ is not rational, and even to some extent when $V$ is not rational. For example, we will treat the case that $V$ is the $\widehat{\mathfrak{u}}(1)$ Kac--Moody algebra (also known as the Heisenberg VOA) in Section \ref{subsec:Heisenberg}.

The bijection of Claim \ref{cl:claim:sym/subduality} is easy to state at a physics level of rigor. Given a symmetry category $\mathcal{C}\subset\mathrm{Sym}^\dagger(V)$ of topological line operators supported on the boundary $(V,\Phi_V)$, one may obtain a unitary conformal subalgebra $V^{\mathcal{C}}$ by passing to the operators $\mathcal{O}(z)$ in $V$ on which lines $X$ in $\mathcal{C}$ act as multiplication by the quantum dimension,
\begin{align}
    V^{\mathcal{C}}=\{\mathcal{O}(z)\in V\mid \widehat{X}\cdot \mathcal{O} = \mathsf{d}_{X}\mathcal{O}, \ \forall X\in \mathcal{C}\}.
\end{align}
This definition of $V^{\mathcal{C}}$ works when every line in $\mathcal{C}$ has finite quantum dimension (as is the case when $\mathcal{C}$ is a fusion category), but in the more general situation that some lines have infinite quantum dimension, we must define $V^{\mathcal{C}}$ to be the subspace of local operators $\mathcal{O}(z)$ which are transparent to every topological line in $\mathcal{C}$, see Figure \ref{fig:VCdef}. We sometimes refer to $V^{\mathcal{C}}$ as a fixed-point subalgebra, or a transparent subalgebra.

\begin{figure}
    \begin{center}
\input{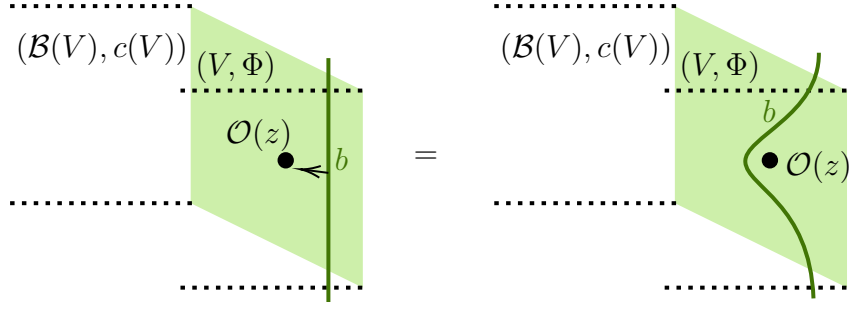}
    \end{center}
    \caption{Operators $\mathcal{O}(z)$ in $V^{\mathcal{C}}$ are by definition transparent to every topological line $b$ in $\mathcal{C}$.}\label{fig:VCdef}
\end{figure}

In the other direction, if one has a unitary conformal subalgebra $W\subset V$, one can consider the category $\Ver (V/W)$ of topological lines in $\mathrm{Sym}^\dagger(V)$ which commute with the local operators in $W$, that is
\begin{align}
    \Ver (V/W)=\{X\in\mathrm{Sym}^\dagger(V)\mid X \text{ commutes w/ }\mathcal{O}, \ \forall\mathcal{O}\in W\}. 
\end{align}

While symmetry/subalgebra duality abstractly establishes that conformal subalgebras are interchangeable with symmetries, it is not terribly useful in practice. For example, the definition of $\Ver (V/W)$ seems to require that we know the entire unitary symmetry category $\mathrm{Sym}^\dagger(V)$ of $V$ at the outset, which is often difficult to achieve. It turns out we can define $\Ver (V/W)$ in a more tractable way, at least when $V$ and $W$ are both rational. We turn to this next.

A useful slogan which will serve as our starting point is the following. \\

\noindent\textbf{Slogan:} \emph{Conformal extensions of a unitary rational chiral algebra $W$ are dual to patterns of noninvertible one-form symmetry gauging --- sometimes called anyon condensation by condensed matter physicists --- in the bulk TQFT $\mathcal{B}(W)$.} \\

\noindent  Indeed, note that any $V$ can always be thought of as a representation over itself, and therefore as a (generally reducible) representation of $W$ by restriction. In particular, $V$ is an object of $\Rep(W)$, or equivalently, 
\begin{align}
    A:=\Phi_W^{-1}(V)
\end{align} defines a non-simple topological line operator in $\mathcal{B}(W)$. We have chosen to label this line operator $A$ because, it turns out, $A$ defines a condensable algebra object in $\mathcal{B}(W)$; equivalently, it prescribes a way of gauging noninvertible one-form symmetries in the bulk topological field theory supporting $W$ on its boundary. See \cite{Kirillov:2001ti,Kong:2013aya,Huang:2014ixa,Creutzig:2017anl} for more on condensable algebras and their relations to VOA extensions.

The topological field theory obtained by gauging $A$ is none other than the bulk theory $\mathcal{B}(V)$ which supports $V$ on its boundary. That is,
\begin{align}\label{eqn:gaugingA}
    (\mathcal{B}(W)_A^{\mathrm{loc}},c(W))=(\mathcal{B}(V),c(V))
\end{align}
where we recall that $\mathcal{B}(W)_A^{\mathrm{loc}}$ is the category of \emph{local} $A$-modules inside of $\mathcal{B}(W)$ (see e.g.\ \cite{Kong:2022cpy} for a description of these concepts). We may therefore obtain a topological interface $\mathcal{I}_A$ between $\mathcal{B}(W)$ and $\mathcal{B}(V)$ by starting with the theory $\mathcal{B}(W)$ and gauging the one-form symmetry corresponding to the algebra object $A$ in half of spacetime.\footnote{We note that, because $W$ is a \emph{conformal} subalgebra of $V$, we necessarily have that $c(W)=c(V)$.} By fusing this topological interface onto the gapless chiral boundary condition $(W,\Phi_W)$ of $\mathcal{B}(W)$, one obtains a boundary condition of the form $(V,\Phi_V)$ of $\mathcal{B}(V)$, i.e.\ 
\begin{align}
    \mathcal{I}_A\otimes W\text{-boundary} \cong V\text{-boundary}.
\end{align}
In this way, one ends up with a picture which is reminiscent of the SymTFT construction.

\begin{claim}{}{}
   Figure \ref{fig:SymTFTconformalembedding} is the SymTFT representation of a conformal embedding $W\subset V$ of rational chiral algebras.
\end{claim}

As in the SymTFT, we want to interpret this construction as ``exposing'' certain topological line operators of $V$.  

\begin{claim}{}{claim:verdef}The category of (not necessarily local) $A$-modules inside of $\mathcal{B}(W)$,
\begin{align}\label{eqn:Ver(V/W)}
    \Ver (V/W)\equiv\mathcal{B}(W)_A,
\end{align}
defines a fusion category of topological line operators supported on the boundary $(V,\Phi_V)$. In the SymTFT picture of Figure \ref{fig:SymTFTconformalembedding}, it arises as the category of topological line operators supported on the interface $\mathcal{I}_A$.
\end{claim}

Of course, one expects to also be able to run this construction in reverse. That is, if one identifies a fusion category $\mathcal{C}$ of topological line operators supported on the boundary condition $(V,\Phi_V)$ which contains the category $\mathcal{B}(V)$ of line operators in the bulk, then there should exist a strongly rational conformal subalgebra $W$ which serves to define the ``physical boundary condition'' of a SymTFT construction as in Figure \ref{fig:SymTFTconformalembedding}. This $W$ is none other than the subalgebra of operators in $V$ which commute with the lines in $\mathcal{C}$, i.e.\ $W=V^{\mathcal{C}}.$ It is then clear from looking at Figure \ref{fig:SymTFTconformalembedding} that $Z(\mathcal{C})\cong \mathcal{B}(W)\boxtimes \overline{\mathcal{B}(V)}$, and we can conclude the following.
\begin{claim}{}{}
    Let $\mathcal{C}$ be a fusion category of topological line operators supported on a rational chiral algebra boundary $V$, and assume that $\mathcal{C}\supset \mathcal{B}(V)$. The corresponding SymTFT is
\begin{align}
    \mathrm{SymTFT} = Z^{\overline{\mathcal{B}(V)}}(\mathcal{C})\cong \Rep(V^{\mathcal{C}}),
\end{align}
where $Z^{\overline{\mathcal{B}(V)}}(\mathcal{C})$ is the Müger centralizer of $\overline{\mathcal{B}(V)}$ in $Z(\mathcal{C})$. It supports $V^{\mathcal{C}}$ as a ``physical'' boundary condition.
\end{claim}

The reason we have used the notation $\Ver(V/W)$ to denote the category $\mathcal{B}(W)_A$ is because, just as the Verlinde lines \cite{Verlinde:1988sn} of a rational conformal field theory are precisely the topological line operators of the theory which commute with the left- and right-moving chiral algebras, $\Ver (V/W)$ is the category of topological line operators of $V$ which commute with the conformal subalgebra $W$. The notation $\Ver (V/W)$ is also meant to evoke Galois theory. In that context, one studies the Galois group $\mathrm{Gal}(E/F)$ of a field extension $E/F$, which is defined to be the subgroup of the automorphism group of $E$ which preserves the subfield $F$. As we will see in Section \ref{subsec:Galois}, the analogy between $\Ver (V/W)$ and $\mathrm{Gal}(E/F)$ can be made sharp and extended. 

We note that, as one should expect, the category $\Ver (V/W)$ contains $\mathcal{B}(V)$ as a fusion subcategory
\begin{align}
    \mathcal{B}(V)\cong \mathcal{B}(W)_A^{\mathrm{loc}}\subset\mathcal{B}(W)_A = \Ver(V/W).
\end{align}
Indeed, by definition \eqref{eqn:Ver(V/W)}, $\Ver (V/W)$  is $\mathcal{B}(W)_A$, the category of $A$-modules in $\mathcal{B}(W)$. On the other hand, by Equation \eqref{eqn:gaugingA}, $\mathcal{B}(V)$ is $\mathcal{B}(W)_A^{\mathrm{loc}}$, the subcategory of $\mathcal{B}(W)_A$ consisting just of the \emph{local} $A$-modules. In fact, as we will see in Section \ref{subsec:hypergroup}, $\Ver (V/W)$ is a $K$-graded extension of $\mathcal{B}(V)$, with $K$ a certain hypergroup.

In a moment, we will treat a class of examples where the conformal subalgebra $W$ is obtained from $V$ by passing to the invariant states under the action of a finite group $G$ of automorphisms, $W=V^G$. To prepare for this, let us recall a widely-believed conjecture in the theory of vertex operator algebras.

\begin{conjecture}\label{conj:Grationality}
    If $V$ is a strongly rational vertex operator algebra and $G$ is a finite group of automorphisms of $V$, then the conformal subalgebra $V^G$ of $G$-invariant states is also strongly rational.
\end{conjecture}

This conjecture is known to be true when $G$ is a solvable group \cite{Carnahan:2016guf}. Also, complete rationality of $G$-fixed points has been established in general for conformal nets \cite{Xu:2000wt}, which are expected to be equivalent to (unitary) VOAs \cite{Carpi:2015fga,Carpi:2023onx,henriques2025every}. For ease of exposition, we will assume that Conjecture \ref{conj:Grationality} is true for the rest of the paper.

\begin{example}{}{cleftpt1}
    Recall that pointed modular tensor categories $\Vec_D^{\sigma,\omega}$ (i.e.\ Abelian TQFTs, up to the chiral central charge) are determined up to equivalence by an Abelian group $D$ (describing the anyons and their fusion), and an Abelian 3-cocycle $(\sigma,\omega)\in H^3_{\mathrm{ab}}(D,U(1))$, where $\omega:D\times D\times D\to U(1)$ encodes the F-symbols/monoidal structure, and $\sigma:D\times D \to U(1)$ encodes the braiding.  The pair $(\sigma,\omega)$ is subject to some coherence conditions, see e.g.\ \cite[Example 2.2]{Galindo:2024qzg} for further details on pointed modular tensor categories. 
    
    \hspace{.3in}Suppose that $V$ appears at the boundary of an Abelian TQFT, i.e.\ that  
    \begin{align}
        \Rep(V)\cong \mathrm{Vec}_D^{\sigma,\omega}
    \end{align}
    for some triple $(D,\sigma,\omega)$.\footnote{Even though $\mathcal{B}(V)$ is an Abelian Chern-Simons modular tensor category, it does not mean that $V$ is lattice vertex operator algebra/free boson chiral algebra. For example, the monster VOA $V^\natural$ appears on the boundary of the TQFT $(\mathcal{B},c)=(\mathrm{Vec},24)$ which is trivially an Abelian TQFT, even though $V^\natural$ is not a lattice VOA. There are even Abelian TQFTs $(\mathcal{B},c)$ which do not support any lattice vertex operator algebras at a given central charge $c$, but which do support some non-lattice $V$-boundaries, see e.g.\ \cite[Example 3.7]{Moller:2024plb}.} Consider a conformal subalgebra  $W=V^G$ consisting of the $G$-invariant states of $V$, where $G$ is a finite group of automorphisms. 
    
    \hspace{.3in}Recall from the discussion preceding Claim \ref{cl:claim:(B,c)boundaries} that every $g\in G$  induces a ribbon auto-equivalence $g_\ast:\Rep(V)\to \Rep(V)$. In simple cases, $g_\ast$ is the identity ribbon auto-equivalence for every $g\in G$, in which case one says that the automorphism group $G$ is \underline{cleft}. In this situation, it is known \cite[Theorem 4.1]{Gannon:2024tcl} that the Verlinde category is group-like, i.e.\ 
    \begin{align}
        \Ver (V/V^G) \cong \mathrm{Vec}_\Gamma^{\tilde\omega},
    \end{align}
    where $\Gamma$ is a central extension of $G$ by $D$, and $\tilde\omega\vert_D=\omega$ in $H^3(D,U(1))$. Here, $\Vec_\Gamma^{\tilde\omega}$ is the pointed fusion category with objects and fusion rules determined by the group $\Gamma$, and with F-symbols/associativity data given by $\tilde\omega \in H^3(\Gamma,U(1))$. We will see a cleft example worked out in detail in Section \ref{subsec:cleft}.
\end{example}

Before moving on, let us generalize this discussion slightly by incorporating topological surface operators of the bulk 3D TQFT $\mathcal{B}(V)$. More specifically, we would like to study the ways that topological surfaces in the bulk can end on boundary lines.

To begin, recall that topological surfaces of $\mathcal{B}(V)$ are in one-to-one correspondence with $\mathcal{B}(V)$-module categories (see e.g.\ \cite{Choi:2023xjw} for a physicist-friendly description of module categories). Equivalently, they correspond to (Morita equivalence classes\footnote{Two algebras $B$ and $B'$ are said to be Morita equivalent if their category of modules, ${_B}\mathcal{B}(V)$ and ${_{B'}}\mathcal{B}(V)$, are equivalent as $\mathcal{B}(V)$-module categories.} of) symmetric special Frobenius algebra objects of $\mathcal{B}(V)$ \cite{Fuchs:2002cm,Kapustin:2010if}, which we refer to simply as ``gaugeable'' algebras. Indeed, as explained in \cite{Roumpedakis:2022aik}, every topological surface is obtained by ``1-gauging''  an algebra $B$ of lines in $\mathcal{B}(V)$, i.e.\ inserting a mesh of the line $B$ just along a codimension-1 locus of 3D spacetime, as in Figure \ref{fig:condensationsurface}. The construction depends only on $B$ through ${_B}\mathcal{B}(V)$, the category of $B$-modules in $\mathcal{B}(V)$, which defines a $\mathcal{B}(V)$-module category.

\begin{figure}
    \begin{center}
        \input{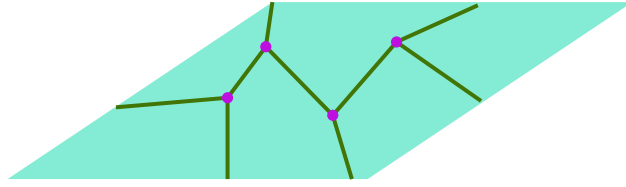}
        \caption{A surface in a 3D TQFT $\mathcal{B}$ is defined by 1-gauging a gaugeable algebra $B$ of lines. The dark green lines are described by $B$, and the purple junctions are the multiplication and comultiplication morphisms $m:B\otimes B\to B$ and ${^\circ}m:B\to B\otimes B$ of the algebra.}\label{fig:condensationsurface}
    \end{center}
\end{figure}

We pause to emphasize the distinction between the notion of a condensable algebra encountered previously, and the notion of a gaugeable algebra described presently. Physically, the former is the structure needed to perform a standard gauging of a one-form symmetry in a 3D TQFT; the latter is the structure needed to perform gauging of a zero-form symmetry of a 2D QFT, or alternatively, a ``1-gauging'' of a one-form symmetry in a 3D TQFT, as in Figure \ref{fig:condensationsurface}. Every condensable algebra defines a gaugeable algebra, but not the other way around. For the more mathematically inclined 2D CFT readers, we note that both kinds of algebras may be used as input to the Fuchs-Runkel-Schweigert construction \cite{Fuchs:2002cm} of rational conformal field theories: gaugeable algebras lead to the most general kinds of gluings between left and right-movers, whereas condensable algebras lead to pure extension type gluings (i.e.\ gluings in which one extends the chiral algebra and then glues diagonally).

With these preliminaries in place, we can state the following claim.

\begin{figure}
    \begin{center}
        \input{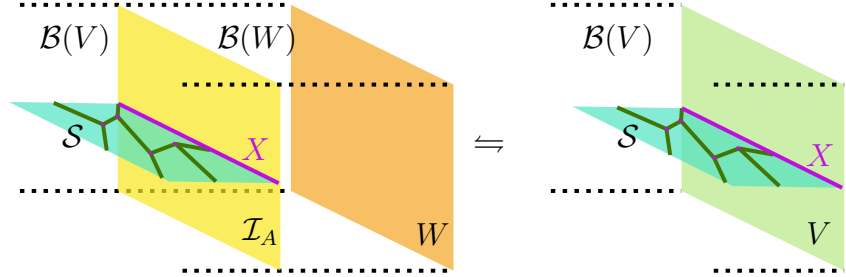}
        \caption{Every object $X$ of the category $\mathrm{Ver}_{\mathcal{S}}(V/W)$ is a $B$-module $(M,\mu)$ in $\Ver(V/W)$, i.e.\ a topological line $M$ on the $V$ boundary on which the bulk surface $\mathcal{S}$ can consistently end via a junction $\mu:M\otimes B\to M$. The SymTFT shows that the line junctions $X$ commute with local operators in $W$.}\label{fig:VerS(VW)}
    \end{center}
\end{figure}

\begin{claim}{}{}
    Suppose $\mathcal{C}\supset \mathcal{B}(V)$ is a fusion category of topological line operators supported on a $V$-boundary, and $\mathcal{S}$ is a topological surface in the bulk TQFT $\mathcal{B}(V)$ described by a gaugeable algebra $B$. 
    Then the category $\Ver_{\mathcal{S}}(V/W)$ of boundary topological line operators on which $\mathcal{S}$ can end and which commute with the local operators in $W=V^{\mathcal{C}}$ is
    \begin{align}\label{eqn:VerS(VW)}
        \Ver_{\mathcal{S}}(V/W) \cong  \Ver(V/W)_B\cong \mathcal{C}_B,
    \end{align}
    i.e.\ the category of $B$-modules in $\mathcal{C}\cong \Ver(V/W)$.
    See Figure \ref{fig:VerS(VW)} for a depiction.
\end{claim}

The basic intuition behind Equation \eqref{eqn:VerS(VW)} is that objects of $\Ver_{\mathcal{S}}(V/W)$ are  lines $M$ on the interface $\mathcal{I}_A$ on which the mesh of algebra object $B$ can consistently end. As part of the data, one must choose a junction $\mu$ between the mesh $B$ and the line $M$, i.e.\ a morphism $M\otimes B \to M$. Thus, objects $X$ in $\Ver_{\mathcal{S}}(V/W)$ are really $B$-modules $X=(M,\mu)$ in $\Ver(V/W)\cong \mathcal{C}$. Similar principles determine how boundary conditions of a 2D QFT transform under gauging zero-form symmetries, c.f.\ \cite[Equation (5.8)]{Huang:2021zvu}.

Before moving on, we note that there is another formula for $\Ver_{\mathcal{S}}(V/W)$ which makes manifest the fact that the surface $\mathcal{S}$ depends only on the gaugeable algebra $B$ through its category of modules, ${_B}\mathcal{B}(V)$. Specifically, one can show that 
\begin{align}
    \Ver_{\mathcal{S}}(V/W) \cong \mathrm{Fun}_{\mathcal{B}(V)}({_B}\mathcal{B}(V),\Ver(V/W)),
\end{align}
i.e.\ $\Ver_{\mathcal{S}}(V/W)$ is the category of right $\mathcal{B}(V)$-module functors from ${_B}\mathcal{B}(V)$ to $\Ver(V/W)$.

\subsection{Twisted sectors and the dome algebra}\label{subsec:twisteddome}

\begin{figure}
    \begin{center}
\input{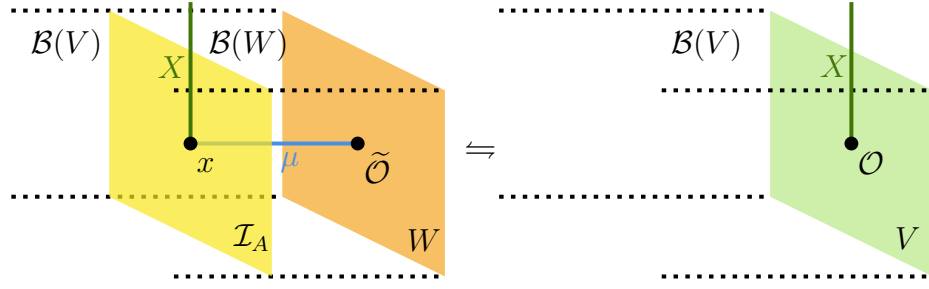}
\caption{The SymTFT interpretation of the decompositions in Equation \eqref{eqn:twisteddecomposition} and \eqref{eqn:twistedtriple}.}\label{fig:SymTFTtwistedops}
    \end{center}
\end{figure}

Suppose that $\mathcal{C}$ is a fusion category of topological line operators supported on a $V$-boundary, and further that $\mathcal{C}$ contains the lines $\mathcal{B}(V)$ coming from the bulk as a subcategory, $\mathcal{B}(V)\subset\mathcal{C}$. As we have described, $\mathcal{C}$ will always be of the form $\Ver(V/W)$ for some strongly rational conformal subalgebra $W$ (specifically, for $W=V^{\mathcal{C}}$ the subalgebra of $\mathcal{C}$-neutral operators). 

For each $X\in\mathcal{C}$, we may consider the Hilbert space $V_X$ of local operators which live at the endpoint of the topological line $X$. A physicist might refer to $V_X$ as a ``twisted sector'' whenever $X$ is different from the identity line, however a mathematician typically reserves the expression ``twisted module'' for when $X \notin \mathcal{B}(V)$, i.e.\ when $X$ cannot be dragged into the bulk. We will adopt the mathematical terminology in what follows, referring to $V_a$ as an ``ordinary $V$-module'' or just ``$V$-module'' whenever $a\in\mathcal{B}(V)$ (see Figure \ref{fig:gaplesschiralboundary} for a depiction of ordinary $V$-modules), and to $V_X$ as a ``twisted $V$-module'' or ``$W$-local twisted $V$-module'' when $X\in\Ver(V/W)$ but $X\notin \mathcal{B}(V)$.

Let us start by describing  ordinary $V$-modules in more detail. If $V$ admits a rational conformal subalgebra $W$ then, by restriction, every $V$-module can be decomposed into a direct sum of irreducible $W$-modules as 
\begin{align}\label{eqn:Vmoddecomp}
    V_a\cong \bigoplus_{\mu\in\mathrm{Irr}(\mathcal{B}(W))} J_a^\mu\otimes  W_\mu, \ \ \ \ \ \ (a\in\mathcal{B}(V)),
\end{align}
where $J_a^\mu$ is a finite-dimensional vector space whose dimension $B_{a\mu}\equiv \dim(J_a^\mu)$ encodes the multiplicity with which the $W$-module $W_\mu$ appears in the decomposition of $V_a$. 

This decomposition is manifested in the SymTFT. Indeed, one can work in a basis such that every operator $\mathcal{O} \in V_a$ decomposes into a triple of the  form, 
\begin{align}\label{eqn:Vmodtriple}
    \mathcal{O}\leftrightharpoons (x,\mu,\widetilde{\mathcal{O}}), \ \ \ \ \ \ \  x\in J_a^\mu, \ \ \mu\in\mathcal{B}(W), \ \ \widetilde{\mathcal{O}}\in W_\mu,
\end{align}
where $\mu$ is a topological line in the $\mathcal{B}(W)$ TQFT, $\widetilde{\mathcal{O}}$ is an operator in $W_\mu$ (i.e.\ an operator on the $(W,\Phi_W)$ boundary attached to the line $\mu$  in the bulk), and $x$ is a topological point junction on the interface $\mathcal{I}_A$ between the lines $\mu\in\mathcal{B}(W)$ and $a\in\mathcal{B}(V)$. The picture corresponding to Equation \eqref{eqn:Vmodtriple} is recovered by taking $X=a$ in Figure \ref{fig:SymTFTtwistedops}.

\begin{figure}
    \begin{center}
        \input{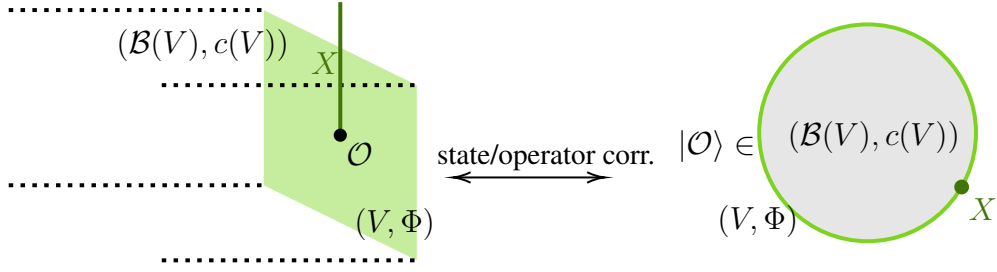}
        \caption{By the state/operator correspondence, the Hilbert space $V_X$ of local operators at the endpoint of a boundary line $X$ is the same as the Hilbert space of states on the disk background on the right. When $X$ is in $\mathcal{B}(V)$, this picture can be alternatively drawn as in Figure \ref{fig:gaplesschiralboundary} by dragging $X$ into the bulk.}\label{fig:stateopcorrtwisted}
    \end{center}
\end{figure}

More generally, for any topological line operator $X\in\mathcal{C}$, we may study the corresponding Hilbert space $V_X$ consisting of the local operators on the boundary $(V,\Phi_V)$ on which $X$ can terminate. By the state/operator correspondence, $V_X$ can also be thought of as the Hilbert space of states on the disk background on the right of Figure \ref{fig:stateopcorrtwisted}. 

The space $V_X$ can no longer be thought of as an ordinary $V$-module because boundary local operators on $(V,\Phi_V)$ incur a monodromy when they pass through the topological line $X$. However, local operators in the subalgebra $W=V^{\mathcal{C}}$ are neutral with respect to the symmetry category $\mathcal{C}$, and so we \emph{can} think of $V_X$ as a (generally reducible) $W$-module. In particular, we obtain a decomposition of the form 
\begin{align}\label{eqn:twisteddecomposition}
    V_X \cong \bigoplus_{\mu\in\mathcal{B}(W)} J_{X}^\mu \otimes W_\mu, \ \ \ \ \ \ (X\in\mathcal{C}, \  B_{X\mu}\equiv \dim(J_X^\mu)),
\end{align}
which generalizes Equation \eqref{eqn:Vmoddecomp}, and also  \cite[Equation (1.6)]{Lin:2022dhv}. In other words, $V_X$ decomposes into ordinary irreducible representations of the transparent subalgebra $W$. Correspondingly, we have the following.
\begin{claim}{}{}
 Every twisted-sector operator $\mathcal{O}(z)\in V_X$ decomposes into a triple of the form
\begin{align}\label{eqn:twistedtriple}
    \mathcal{O}\leftrightharpoons (x,\mu,\widetilde{\mathcal{O}}), \ \ \ \ \ \ x\in J_X^\mu, \ \ \mu\in\mathcal{B}(W), \ \ \widetilde{\mathcal{O}}\in W_\mu,
\end{align}
as depicted in Figure \ref{fig:SymTFTtwistedops}.
\end{claim}

In \cite{grmath}, we will give a sharper definition of the notion of a $W$-local twisted $V$-module. For now, we note that there is a tensor equivalence 
\begin{align}
\begin{split}
    \Phi:\mathcal{C}&\xrightarrow{\sim} \mathrm{TwRep}_{W}(V) \\
    X&\mapsto \Phi(X)=:V_X,
\end{split}
\end{align} 
between the category of boundary lines in $\mathcal{C}$ and the category $\mathrm{TwRep}_{W}(V)$ of $W$-local twisted $V$-modules. This extends the equivalence $\Phi:\mathcal{B}(V)\xrightarrow{\sim}\Rep(V)$ described in Figure \ref{fig:gaplesschiralboundary}.

Let us give further details on the multiplicities $B_{X\mu}$ from Equation \eqref{eqn:twisteddecomposition}. Note that $X$, by definition, is an object of $\Ver(V/W)\cong \mathcal{B}(W)_A$, i.e.\ an $A$-module inside of $\mathcal{B}(W)$. Thus in particular, there is a ``restriction'' functor $R:\mathcal{B}(W)_A\to\mathcal{B}(W)$ which is obtained by forgetting the structure of $X\in\mathcal{B}(W)_A$ as an $A$-module, thinking of it instead just as an object of $\mathcal{B}(W)$. The physical interpretation of the restriction functor is given in Figure \ref{fig:physicalrestriction}.
From this, it follows that the number of topological point junctions $B_{X\mu}$ is captured by the number of times that $\mu$ appears inside of $R(X)$, i.e.
\begin{align}\label{eqn:Bdef}
    B_{X\mu}=\dim\mathrm{Hom}_{\mathcal{B}(W)}(R(X),\mu).
\end{align}
In Section \ref{subsec:computationalaids}, we will explain how induction provides a concrete computational tool for computing the $B_{X\mu}$ in examples.

\begin{figure}
\begin{center}
    \input{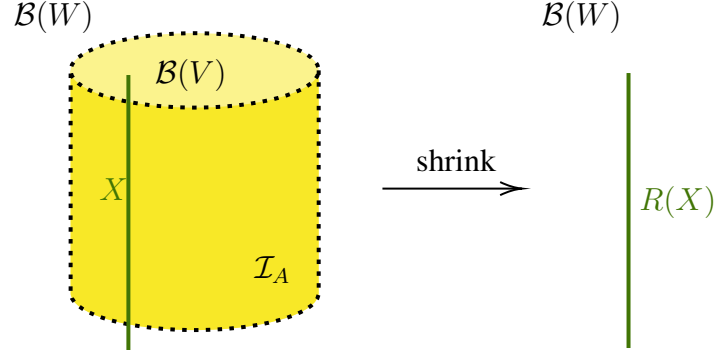}
    \caption{The physical interpretation of the restriction functor $R:\mathcal{B}(W)_A\to\mathcal{B}(W)$. One considers a topological line $X\in\mathcal{B}(W)_A$ on the interface $\mathcal{I}_A$, wrapped into the shape of a cylinder, and considers the line $R(X)\in\mathcal{B}(W)$ obtained when the radius of the cylinder is taken to $0$.}\label{fig:physicalrestriction}
\end{center}
\end{figure}

Let us define an extended Hilbert space $\mathbb{V}_{\mathcal{C}}$ as the direct sum over all irreducible $W$-local twisted $V$-modules,
\begin{align}\label{eqn:extendedHilbertspace}
    \mathbb{V}_{\mathcal{C}}\equiv \bigoplus_{X\in\mathrm{Irr}(\mathcal{C})}V_X. 
\end{align}
The decomposition of this space into $W$-modules, 
\begin{align}\label{eqn:SchurWeyl}
    \mathbb{V}_{\mathcal{C}}\cong \bigoplus_{\mu\in \mathrm{Irr}(\mathcal{B}(W))} J^\mu\otimes W_\mu, \ \ \ \ \ J^\mu\equiv \bigoplus_{X\in \mathrm{Irr}(\mathcal{C})}J^\mu_X,
\end{align}
admits an interesting representation-theoretic interpretation.

To explain this, note that the tube algebra $\mathrm{Tube}(\mathcal{C})$ (first defined in \cite{ocneanu1994chirality}) acts on $\mathbb{V}_{\mathcal{C}}$ via lassos \cite{Lin:2022dhv},
\begin{align}\label{eqn:lassoaction}
\input{Figures/lassoaction}
\end{align}
Here, $X,Y,Z$ are lines in $\mathcal{C}$ and $v$ is a suitable topological point junction between them.
However, a subtlety we must take into account is that the tube algebra does not act faithfully in the setting of relative QFTs, i.e.\ when the bulk $\mathcal{B}(V)$ has anyons. Indeed, if one lassos an operator $\mathcal{O}$ with a line $a \in\mathcal{B}(V)$, then $a$ can be pulled into the bulk and collapsed to a point on the other side of $\mathcal{O}$, as in Figure \ref{fig:domeequivrel}. Thus, the algebra which acts faithfully is a certain quotient of $\mathrm{Tube}(\mathcal{C})$ which was called the dome algebra in \cite{Green:2023ork}.
\begin{definition}
    The dome algebra
\begin{align}
    \mathrm{Dome}_{\mathcal{B}(V)}(\mathcal{C})\equiv \mathrm{Tube}(\mathcal{C})/\sim,
\end{align}
is the quotient of the tube algebra by the equivalence relation in Figure \ref{fig:domeequivrel}. 
\end{definition}

\begin{figure}
    \begin{center}
 \input{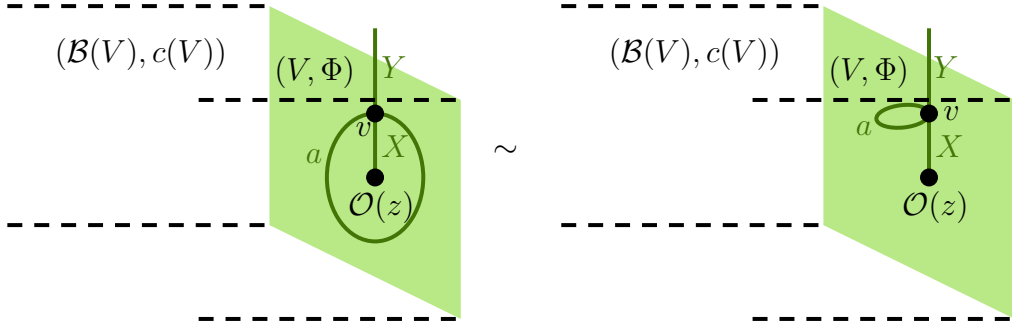}
\caption{In a relative theory, lasso actions of $\Ver (V/W)$ generate a quotient of the tube algebra, rather than the tube algebra itself, due to the fact that certain lines $a\in\mathcal{B}(V)$ can be dragged into the bulk. }\label{fig:domeequivrel}
    \end{center}
\end{figure}

On the other hand, the SymTFT picture makes it clear that lassos only act on a twisted sector operator $\mathcal{O}\leftrightharpoons (x,\mu, \widetilde{\mathcal{O}})$ through its associated topological point junction $x \in J^\mu$; this is simply because the lassos appear in the SymTFT  supported on the topological interface $\mathcal{I}_A$ in Figure \ref{fig:SymTFTtwistedops}, and thus by locality cannot alter $\mu$ or $\widetilde{\mathcal{O}}$. 
In particular, the action of the dome algebra on $\mathbb{V}_{\mathcal{C}}$ descends to an action on each of the spaces $J^\mu$ individually. Schematically,
\begin{align}
    L_{X;Z}^{Y,v} \cdot (x,\mu,\widetilde{\mathcal{O})} = ( L_{X;Z}^{Y,v}\cdot x,\mu, \widetilde{\mathcal{O}}).
\end{align}
It turns out that each of the $J^\mu$ define an \emph{irreducible}  representation of $\mathrm{Dome}_{\mathcal{B}(V)}(\mathcal{C})$, and in fact, every irreducible representation of the dome algebra arises in this way \cite{Green:2023ork}. In particular, we learn the following.
\begin{claim}{}{}
There is a one-to-one correspondence between representations of the dome algebra and anyons in the SymTFT (i.e.\ representations of the transparent subalgebra $V^{\mathcal{C}}$),
    \begin{align}\label{eqn:domereps}
    \mathrm{Rep}(\mathrm{Dome}_{\mathcal{B}(V)}(\mathcal{C}))\cong \mathcal{B}(V^{\mathcal{C}})\cong Z^{\overline{\mathcal{B}(V)}}(\mathcal{C}).
\end{align}
\end{claim}
By way of comparison, we remark that, by folding along $\mathcal{I}_A$, one deduces that the representation category of the tube algebra in this context would be given by 
\begin{align}\label{eqn:Tubereps}
    \Rep(\mathrm{Tube}(\mathcal{C}))\cong Z(\mathcal{C})\cong  \mathcal{B}(V^{\mathcal{C}})\boxtimes\overline{\mathcal{B}(V)},
\end{align}
where we have used the general fact that $\Rep(\mathrm{Tube}(\mathcal{C}))\cong Z(\mathcal{C})$.
Thus, the quotient by the equivalence relation $\sim$ has the effect of killing the representations associated with the factor $\overline{\mathcal{B}(V)}$ in Equation \eqref{eqn:Tubereps}. Indeed, representations of the tube algebra labeled by $(\mu,b)\in\mathcal{B}(V^{\mathcal{C}})\boxtimes\overline{\mathcal{B}(V)}$ with $b\neq 1$ are precisely the ones in which the equivalence relation $\sim$ is violated, as one can easily convince oneself by appealing to the SymTFT.

Now, the action of the dome algebra on $\mathbb{V}_{\mathcal{C}}$ clearly commutes with the action of $V^{\mathcal{C}}$ on $\mathbb{V}_{\mathcal{C}}$ via the OPE; indeed, the SymTFT makes this clear because it exhibits the dome algebra as acting on topological junctions $x$ on the interface $\mathcal{I}_A$, while $V^{\mathcal{C}}$ acts on operators $\widetilde{\mathcal{O}}$ on the (spatially separated) $W$ boundary. Thus, we have the following interpretation.

\begin{claim}{}{}
    Equation \eqref{eqn:SchurWeyl} is a Schur-Weyl decomposition of the extended Hilbert space $\mathbb{V}_{\mathcal{C}}$ into $(\mathrm{Dome}_{\mathcal{B}(V)}(\mathcal{C}),V^{\mathcal{C}})$-bimodules.
\end{claim}
This generalizes observations made in \cite{Lin:2022dhv,Choi:2024tri} to the setting of 2D \emph{relative} quantum field theories. It would be interesting to understand when $\mathrm{Dome}_{\mathcal{B}(V)}(\mathcal{C})$ admits some kind of Hopf structure. 

The dome algebra is a new entrant in an expanding list of algebras which act on observables of 2D QFTs, including tube algebras \cite{Lin:2022dhv}, boundary tube algebras/strip algebras \cite{Cordova:2024vsq,Cordova:2024iti,Copetti:2024dcz,Choi:2024tri}, generalized tube algebras/annular algebras \cite{Choi:2024tri}, and double triangle algebras \cite{Petkova:2001ag}. Connections between and generalizations of these algebras are clearly worthy of further study.

\begin{example}{}{cleftpt2}
Consider a chiral algebra $V$ and define $W=V^G$ for some finite group $G$ of automorphisms. In this case, the category $\mathrm{TwRep}_{V^G}(V)$ of $V^G$-local twisted $V$-modules just coincides with the usual category of $G$-twisted $V$-modules (see e.g.\ \cite{Dong:1997ea} for a formal definition of $G$-twisted $V$-modules in the context of vertex operator algebras). 

\hspace{.3in}Continuing Example \ref{ex:cleftpt1}, assume further that the chiral algebra $V$ has a pointed representation category, $\Rep(V)\cong \mathrm{Vec}_D^{\sigma,\omega}$, and that $G$ is a cleft group of automorphisms so that $\Ver (V/V^G)\cong \mathrm{Vec}_\Gamma^{\tilde\omega}$. In this situation, the dome algebra $\mathrm{Dome}_{\mathcal{B}(V)}(\Ver (V/V^G))$ will be a certain quotient of 
\begin{align}
    \mathrm{Tube}(\Ver (V/V^G))\cong \mathscr{D}_{\tilde\omega}(\Gamma),
\end{align}
where $\mathscr{D}_{\tilde\omega}(\Gamma)$ is the ($\tilde\omega$-twisted) quantum double of $\Gamma$. We conjecture that this quotient coincides with a quasi-Hopf algebra $\mathscr{D}_{\tilde\omega}(\Gamma,D)$ defined by Mason-Ng\footnote{It was shown in \cite{Gannon:2024tcl} that this quasi-Hopf algebra is also equivalent to one defined by Naidu \cite{naidu2010crossed}.} \cite{Mason:2014kea} (see also \cite[Definition A.5]{Gannon:2024tcl}), i.e. 
\begin{align}
    \mathrm{Dome}_{\mathcal{B}(V)}(\Ver (V/V^G))\cong \mathscr{D}_{\tilde\omega}(\Gamma,D).
\end{align}
One consistency check on this conjecture is that the representation category of the Mason-Ng algebra $\mathscr{D}_{\tilde\omega}(\Gamma,B)$ is known to be the $G$-equivariantization of $\Ver (V/V^G)\cong\mathrm{Vec}_\Gamma^{\tilde\omega}$, which in turn recovers the representation category of $W=V^G$. In other words, it is known that 
\begin{align}
\Rep(\mathscr{D}_{\tilde{\omega}}(\Gamma,D))\cong \mathcal{B}(V^G) , 
\end{align}
in agreement with the general expectation Equation \eqref{eqn:domereps} set by the SymTFT.
We will give further evidence for this conjecture in Example \ref{ex:cleftpt3}.
\end{example}

\subsection{The effective hypergroup}\label{subsec:hypergroup}

In an absolute 2D QFT with a fusion category symmetry $\mathcal{C}$, the Grothendieck ring/fusion algebra of $\mathcal{C}$ appears as a subalgebra of the tube algebra $\mathrm{Tube}(\mathcal{C})$. Indeed, it corresponds to lasso diagrams where the incoming and outgoing lines are both taken to be the identity, i.e.\ to lassos $\mathsf{L}_{X;Z}^{Y,v}$ in Equation \eqref{eqn:lassoaction} with $X=Y=\mathds{1}$ and $v$ the identity junction. Intuitively, the fusion algebra of $\mathcal{C}$ is the part of $\mathrm{Tube}(\mathcal{C})$ which maps genuine local operators to genuine local operators.

In relative QFTs, we saw that $\mathrm{Tube}(\mathcal{C})$ does not act faithfully on the extended Hilbert space, and, it turns out, neither does the subalgebra corresponding to the fusion algebra of $\mathcal{C}$. 
Instead, one must pass to 
\begin{align}\label{eqn:effectivehypergroup}
   K:=K_{\mathcal{C}}\sslash K_{\mathcal{B}(V)},
\end{align} 
where $K_{\mathcal{C}}$ (resp.\ $K_{\mathcal{B}(V)}$) is the hypergroup induced by the fusion ring of $\mathcal{C}$ (resp.\ $\mathcal{B}(V)$), and $K_{\mathcal{C}}\sslash K_{\mathcal{B}(V)}$ is the double coset hypergroup  \cite{Bischoff:2016jmy,Rie22} (see Appendix \ref{app:hypergroups} for a review). This is the correct algebraic structure which acts faithfully on the genuine local operators in $V$, 
\begin{align}
    K \curvearrowright V.
\end{align}
If we consider a symmetry obtained from a subalgebra $W$, i.e.\ take $\mathcal{C}=\Ver (V/W)$, then the subspace of states which are invariant with respect to the action of $K=K_{\Ver(V/W)}\sslash K_{\mathcal{B}(V)}$ is precisely $W$ \cite[Theorem 7]{Rie22}, i.e.\ 
\begin{align}
    W=V^K := \{\mathcal{O}(z)\in V\mid r_i\cdot \mathcal{O}=\mathcal{O}, \forall r_i \in K\}.
\end{align}
This is a clear corollary of what we have said so far because $W$ is the subalgebra of $V$ which is transparent to the lines in $\mathcal{C}$, and $K$ is just a kind of quotient of $\mathcal{C}$ by the subcategory which acts trivially. 

Another important property of the hypergroup \eqref{eqn:effectivehypergroup} is the following.

\begin{claim}{}{cla:hypergroupgrading}
Any category $\mathcal{C}\supset\mathcal{B}(V)$ of topological line operators supported on a $V$-boundary is a $K$-graded extension of the category $\mathcal{B}(V)$ of anyons in the bulk,
\begin{align}
    \mathcal{C}\cong \bigoplus_{r_i\in K}\mathcal{B}(V)_{r_i}, \ \ \ \ \ \ \mathcal{B}(V)_1 \cong \mathcal{B}(V).
\end{align}
\end{claim}
\noindent We will give a physical argument for this in the next subsection. 

For now, let us give the definition of a $K$-grading. Let $P_{ij}^k$ be the (usually not integer) structure constants of the hypergroup $K$, i.e.
\begin{align}
    r_i\star r_j = \sum_{k} P_{ij}^k\cdot r_k, \ \ \ \ \ \  \sum_{k}P_{ij}^k = 1,
\end{align}
and let $\pi_i:\mathcal{C}\to \mathcal{B}(V)_{r_i}$ be the natural projections. To be $K$-graded means that, if $X$ is a simple object of $\mathcal{B}(V)_{r_i}$, and $Y$ a simple object of $\mathcal{B}(V)_{r_j}$, then 
\begin{align}\label{eqn:hypergroupgrading}
    \frac{\dim(\pi_k(X\otimes Y))}{\dim(X\otimes Y)} = P_{ij}^k, \ \ \ \ \ X\in\mathcal{B}(V)_{r_i}, \ Y\in\mathcal{B}(V)_{r_j}.
\end{align}
In other words, the fraction of the quantum dimension of $X\otimes Y$ which comes from the $r_k$-graded component is given by the hypergroup structure constant $P_{ij}^k$. In particular, if $P_{ij}^k=0$, then the tensor product $X\otimes Y$ contains no  summands $Z\subset X\otimes Y$ which belong to the component $\mathcal{B}(V)_{r_k}$.

\begin{example}{}{Ghypergroup}The simplest example is when $W=V^G$, for some finite subgroup $G$ of $\mathrm{Aut}(V)$. Assume that $W=V^G$ is strongly rational, which is a theorem in the case that $G$ is solvable \cite{Carnahan:2016guf} (cf.\ Conjecture \ref{conj:Grationality}). In this case, the hypergroup is simply $G$ itself, and one learns that $\Ver (V/W)$ is a $G$-graded extension of $\mathcal{B}(V)$, 
\begin{align}
    \Ver(V/W)\cong \bigoplus_{g\in G} \mathcal{B}(V)_g, \ \ \ \ \mathcal{B}(V)_{1}=\mathcal{B}(V).
\end{align}The way one usually sees this stated is that the category of $G$-twisted representations $\mathrm{TwRep}_G(V)$, which in this case coincides with $\mathrm{TwRep}_W(V)\cong \Ver (V/W)$, is a $G$-graded extension of $\Rep(V)$.  In fact, it is known that $\Ver (V/W)$ is a \underline{$G$-crossed braided} extension of $\mathcal{B}(V)$ \cite{kirillov2002modular,Kirillov:2001uz,kirillov2004g,muger2004galois,muger2005conformal,Mcrae:2019pol}. It is expected there exists a generalization of $G$-crossed braided categories to the setting of hypergroups, see e.g.\ \cite{Rie22} for some discussion. 
\end{example}

\begin{example}{}{}
Another simple example is when $V$ is a chiral CFT, a.k.a.\ a holomorphic vertex operator algebra, defined as a chiral algebra $V$ with a trivial representation category, $\mathcal{B}(V) = \mathrm{Vec}$. If $V$ admits a fusion category symmetry $\mathcal{C}$, then the hypergroup is simply $K_{\mathcal{C}}$---where  $K_{\mathcal{C}}$ is the hypergroup induced by the fusion ring of $\mathcal{C}$, see Appendix \ref{app:hypergroups}---since we do not need to take any quotient in Equation \eqref{eqn:effectivehypergroup}. The category $\mathcal{C}$ is trivially a $K_{\mathcal{C}}$-graded extension of $\mathcal{B}(V)=\mathrm{Vec}$. See e.g.\ \cite{Lin:2019hks,Burbano:2021loy,Rayhaun:2023pgc,Volpato:2024goy,Fosbinder-Elkins:2024hff,Moller:2024xtt} for examples of fusion category symmetries of chiral CFTs.
\end{example}

The following is a far-reaching situation. Suppose that one has a sequence of conformal embeddings $W\subset V\subset U$, and that the conformal embeddings $W\subset U$ and $V\subset U$ induce associated hypergroups $K$ and $K'$, respectively. In this case, the hypergroup induced by the embedding $W\subset V$ is the double coset hypergroup $K\sslash K'$:
\begin{equation}
\label{eqn:consecutiveembeddings}
    \begin{tikzcd}
& U & \\
W \arrow[rr,"K\sslash K'"] \arrow[ru,"K"] & & V\,. \arrow[lu,"K'",swap]  
\end{tikzcd}
\end{equation}
The action of $K\sslash K'$ on $V=U^{K'}$ is given by 
\begin{align}
    [r_i]\cdot \mathcal{O}(z) \equiv  \omega_{K'}\cdot r_i \cdot \omega_{K'}\cdot\mathcal{O}(z),
\end{align}
where $\omega_{K'}$ is the Haar element of $K'$, Equation \eqref{eqn:Haarelement}. This action is of course independent of the choice of representative $r_i\in K$ of $[r_i]$. 

In this situation, not only is $\Ver (V/W)$ a $K\sslash K'$-graded extension of $\mathcal{B}(V)=\Ver (V/V)$, but more generally, $\Ver (U/W)$ is a $K\sslash K'$-graded extension of $\Ver (U/V)$.

\begin{example}{Cosine hypergroup symmetry}{cosine}
Take $U=V_{2m}$ and $V=V_{2m}^+$ and $W=\widehat{\mathfrak{u}}(1)^+$. Here, $V_{2m}$ is the chiral algebra of the compact boson theory of radius $R^2=2m$ (cf.\ Example \ref{ex:latticeVOA}) and $\widehat{\mathfrak{u}}(1)$ is the $U(1)$ Kac-Moody algebra. Also, $V_{2m}^+$ (resp.\ $\widehat{\mathfrak{u}}(1)^+$) is the subalgebra of operators in $V_{2m}$ (resp.\ $\widehat{\mathfrak{u}}(1)$) which are neutral under charge conjugation symmetry. 

\hspace{.25in}By Example \ref{ex:Ghypergroup}, since $V_{2m}^{O(2)}=\widehat{\mathfrak{u}}(1)^+$, the hypergroup induced by the embedding $\widehat{\mathfrak{u}}(1)^+\subset V_{2m}$ is $K\equiv O(2)\cong \mathrm{Aut}^\dagger(V_{2m})$. Similarly, the hypergroup induced by the embedding $V_{2m}^+\subset V_{2m}$ is $K'\equiv \mathbb{Z}_2^{\mathrm{C}}$, the group generated by the charge conjugation symmetry. Thus, we have an action of 
\begin{align}
    K\sslash K' \cong O(2)\sslash \mathbb{Z}_2^{\mathrm{C}}
\end{align}
on the chiral algebra $V_{2m}^+$. To describe this hypergroup, call $R_\theta$ with $\theta \in [0,2\pi)$ the rotations in $O(2)$ and call $C$ the reflection.  The double cosets are $K\sslash K'=\{R_\theta\mid \theta\in[0,\pi]\}$, where
\begin{align}
    [R_\theta] = \begin{cases}
        \{R_\theta,R_{-\theta},R_\theta C, R_{-\theta}C\},&\theta\notin \{ 0,\pi\} \\
        \{R_\theta ,R_\theta C\}, &\theta\in \{0,\pi\}.
    \end{cases}
\end{align}
The hypergroup multiplication follows from Equation \eqref{eqn:doublecosethypgpintro}, which unpacks as
\begin{align}
    [R_\theta]\star [R_{\theta'}]=\frac12([R_{\theta+\theta'}]+[R_{\theta-\theta'}]).
\end{align}
One can readily compute that the action of $[R_\theta]$ on e.g.\ the operators $\mathcal{O}_\ell(z)\equiv \cos(\sqrt{2m}\ell\phi(z))\in V_{2m}^+$ is given by 
\begin{align}
    [R_\theta]\cdot \mathcal{O}_\ell(z)=\cos(\theta\ell) \mathcal{O}_{\ell}(z).
\end{align}
This could be called a cosine hypergroup symmetry of the chiral algebra of the orbifold branch.
\end{example}

\subsection{Relation to bulk surfaces}\label{subsec:surfaces}

It will be useful in what follows to develop an alternative perspective on categories of boundary topological line operators, one which incorporates topological surfaces in the bulk. In particular, this will allow us to see why they are hypergroup-graded extensions.

First, recall that topological surfaces $\mathcal{S}$ in the TQFT described by $\mathcal{B}(V)$ are in one-to-one correspondence with $\mathcal{B}(V)$-module categories $\mathcal{M}$. Under this correspondence, objects $m\in\mathcal{M}$ correspond to topological lines which can appear at the boundary of the topological surface $\mathcal{S}$, and the module category map $\mathcal{B}(V)\times \mathcal{M}\to\mathcal{M}$ encodes the fusion of lines in the bulk onto lines on the boundary of $\mathcal{S}$. (See e.g.\ \cite{Choi:2023xjw} for a review of module categories oriented towards physicists.)

The following will be a useful lemma.
\begin{claim}{}{} 
Let $\mathcal{C}\supset \mathcal{B}(V)$ be a category of topological line operators on a $V$-boundary, let $W=V^{\mathcal{C}}$ be the subalgebra of $\mathcal{C}$-neutral operators, and call $\mathcal{I}\equiv \mathcal{I}_A$ the topological interface between $\mathcal{B}(W)$ and $\mathcal{B}(V)$ (cf.\ Figure \ref{fig:SymTFTconformalembedding}). Then $\mathcal{C}$ is equivalent to the category of topological lines which can appear at the boundary of the (generally decomposable) surface $\mathcal{S}\equiv \mathcal{I}^\ast \otimes \mathcal{I}$ in $\mathcal{B}(V)$.
\end{claim}

\noindent In particular, this means that the $\mathcal{B}(V)$-module category corresponding to $\mathcal{S}=\mathcal{I}^\ast\otimes\mathcal{I}$ is simply $\mathcal{C}$ itself.\footnote{Note that $\CC$ is a $\mathcal{B}(V)$-module category simply because $\mathcal{B}(V)$ is a subcategory of $\CC$ and therefore acts on $\CC$ by the tensor product.} 

The equivalence of the claim above sends an object $X$ of $\CC$, which can be thought of as a line supported on the interface $\mathcal{I}$, to a line $\underline{X}$ at the boundary of $\mathcal{S}$ by  ``pinching''  the interface $\mathcal{I}$ at the location of the line, as in Figure \ref{fig:pinch}. In the process, we see that $\mathcal{S}$ admits a privileged topological line junction $L$ on the interface $\mathcal{I}$. We use the following schematic equation to encode the content of Figure \ref{fig:pinch}: 
\begin{align}\label{eqn:pinchequation}
    X \rightleftharpoons (\underline{X},\mathcal{S},L).
\end{align}
We think of $X$ as being ``inflated'' into a strip of the surface $\mathcal{S}$ which is bounded by $\underline{X}$ and $L$. Alternatively, we imagine that we are able to drag $X$ into the bulk at the expense of leaving behind a trailing surface $\mathcal{S}$ which is anchored to the topological interface $\mathcal{I}$ via the line junction $L$. A similar ``pinching'' manipulation was considered in \cite{KNBalasubramanian:2025vum}.

Let us see how the hypergroup grading arises from this perspective.  
The surface $\mathcal{S}$ generally decomposes into a direct sum of elementary surfaces, 
\begin{align}\label{eqn:Sdecomp}
    \mathcal{S}\cong\bigoplus_{r_i\in K}\mathcal{S}_{r_i}.
\end{align} 
Each $\mathcal{S}_{r_i}$, by virtue of being a topological surface operator, corresponds to a $\mathcal{B}(V)$-module category which we call $\mathcal{B}(V)_{r_i}$. The decomposition in Equation \eqref{eqn:Sdecomp} translates at the level of $\mathcal{B}(V)$-module categories to a decomposition of $\mathcal{C}$, i.e.\ 
\begin{align}
    \mathcal{C}\cong \bigoplus_{r_i\in K}\mathcal{B}(V)_{r_i}.
\end{align}
Then, instead of using the surface $\mathcal{S}$ in Figure \ref{fig:pinch} (or Equation \eqref{eqn:pinchequation}), we can pass to components and write  
\begin{align}\label{eqn:stripdecompositionSk}
    X\leftrightharpoons (\underline{X},\mathcal{S}_{r_i},L_{r_i}),
\end{align}
for any $X\in\mathcal{B}(V)_{r_i}\subset\mathcal{C}$. So far this is just a decomposition of $\mathcal{C}$, but we will see shortly why it is in fact a hypergroup grading, in the sense of Equation \eqref{eqn:hypergroupgrading}.

\begin{figure}
    \begin{center}
\input{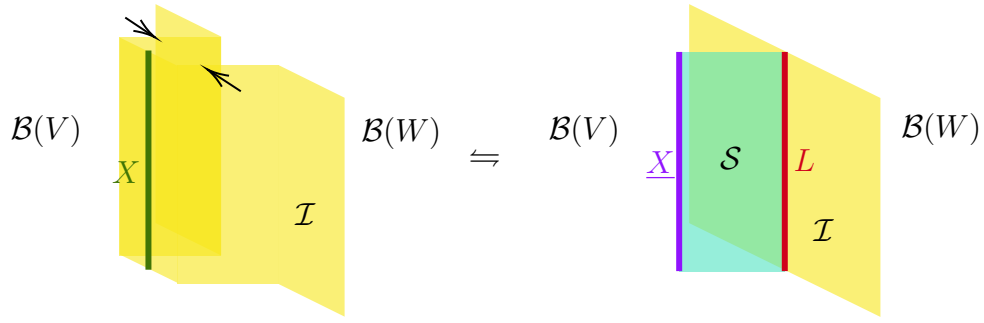}
\caption{Pinching the interface $\mathcal{I}$ near the location of the line $X$ produces a line $\underline{X}$ at the boundary of the surface $\mathcal{S}=\mathcal{I}^\ast\otimes\mathcal{I}$ in $\mathcal{B}(V)$.}\label{fig:pinch}
    \end{center}
\end{figure}

\begin{example}{}{}
   Lines $X$ in $\mathcal{B}(V)\cong \mathcal{B}(V)_1\subset\CC$ correspond to triples $(\underline{X},\mathcal{S}_1,L_1)$ with $\mathcal{S}_1$ the trivial surface (corresponding to the fact that $\mathcal{B}(V)$ is the regular module category over itself), $L_1$ the identity line on $\mathcal{I}$, and $\underline{X}$ a bulk topological line operator. Of course, in this situation, inflating $X$ into a triple $(\underline{X},\mathcal{S}_1,L_1)$ corresponds to simply pulling the line $X$ off the interface $\mathcal{I}$ and into the bulk $\mathcal{B}(V)$. 
\end{example} 

\begin{example}{}{}Suppose that $V$ admits a faithful action of a finite group $K=G$ by automorphisms. In this situation, $\Ver (V/V^G)$ is a $G$-crossed braided extension of $\mathcal{B}(V)$, and one can ask how to describe the surface operators $\mathcal{S}_g$ which appear in the ``strip'' picture of line operators in $\mathcal{B}(V)_g$. As mentioned in Example \ref{ex:cleftpt1}, each $g\in G$ induces a permutation $g_\ast$ on simple modules of $V$, which further extends to a ribbon auto-equivalence of $\mathrm{Rep}(V)\cong \mathcal{B}(V)$. Ribbon auto-equivalences of $\mathcal{B}(V)$ can be thought of as \emph{invertible} topological surface operators of the corresponding topological field theory, and the claim is that $\mathcal{S}_g$ is to be identified with this surface operator $g_\ast$ (using $\Phi:\mathcal{B}(V)\to \Rep(V)$ to relate autoequivalences of $\Rep(V)$ and autoequivalences of $\mathcal{B}(V)$). 
\end{example}

\begin{figure}
    \begin{center}
\input{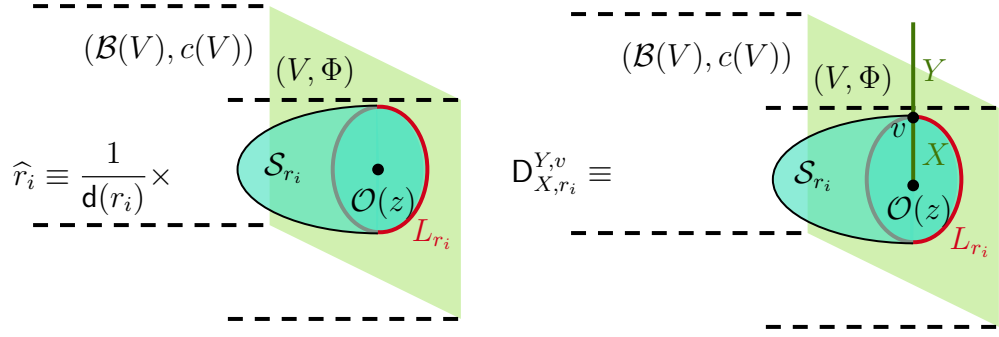}
\caption{Left: The action of an element $r_i$ of the hypergroup $K$ on a local operator $\mathcal{O}(z)$ in $V$ can be represented via a dome, up to a $r_i$-dependent proportionality constant $1/\mathsf{d}(r_i)$. Right: More generally, the action of the dome algebra on twisted local operators can also be represented via domes.}\label{fig:domehypergroup}
    \end{center}
\end{figure}

The ``strip'' perspective on topological lines on a $V$-boundary affords a nice picture of the  hypergroup action on $V$, and the action of the dome algebra on the extended Hilbert space $\mathbb{V}_{\mathcal{C}}$ from Equation \eqref{eqn:extendedHilbertspace}, to which we turn next.

To make contact with the effective hypergroup, we work in a slightly unusual normalization in which the action of topological lines $X\in\mathcal{C}$ on genuine local operators $\mathcal{O}(z)\in V$ is rescaled from the usual one (cf.\ Figure \ref{fig:encirclingaction}) by a factor of the inverse quantum dimension:
\begin{center}
    \input{Figures/rescaledaction}
\end{center}
This defines an action on $V$ of the (stochastically normalized) hypergroup $K_{\mathcal{C}}$ induced by the fusion ring of $\mathcal{C}$.

It is clear that if $X$ lives in the subcategory $\mathcal{B}(V)\subset\mathcal{C}$ of lines which can be dragged into the bulk, then $X\cdot \mathcal{O}(z)=\mathcal{O}(z)$. In particular, the subhypergroup $K_{\mathcal{B}(V)}$ is in the kernel of the action $K_{\mathcal{C}}\curvearrowright V$, and thus we obtain an action of the double coset hypergroup, $K_{\mathcal{C}}\sslash K_{\mathcal{B}(V)}\curvearrowright V$.

We can represent this action in another way. We imagine inflating $X$ into a strip, and then contracting $\underline{X}$ to a point, wherein we incur an $X$-dependent proportionality factor which we call $\gamma(X)$:
\begin{center}
\input{Figures/line2dome}
\end{center}
Although it naively appears as though the proportionality constant $\gamma(X)/\dim(X)$ depends on $X$, it only does so through the hypergroup element $r_i$ that specifies the graded component $\mathcal{B}(V)_{r_i}\subset \mathcal{C}$ to which $X$ belongs. Indeed, by calculating the action of $X$ on the identity local operator in two ways --- one in the standard way using Equation \eqref{eqn:hypact1}, and one by using Equation \eqref{eqn:hypact2} --- one can straightforwardly derive that 
\begin{align}\label{eqn:gammadimd}
    \frac{\gamma(X)}{\dim(X)}=\frac{1}{\mathsf{d}(r_i)},
\end{align}
where $\mathsf{d}(r_i)$ is defined to be the multiple of the identity operator obtained by the following manipulation:
\begin{center}
    \input{Figures/dri}
\end{center}
In total, one finds that one can represent the action of the effective hypergroup $K_{\mathcal{C}}\sslash K_{\mathcal{B}(V)}=\{r_i\}$ on $V$ via domes built out of the surfaces $\mathcal{S}_{r_i}$, as in the left of Figure \ref{fig:domehypergroup}. Similarly, one can represent the action of the dome algebra on the extended Hilbert space $\mathbb{V}_{\mathcal{C}}$ as in the right of Figure \ref{fig:domehypergroup}.

It will also be convenient in what follows to define the following rescaled operators on $V$:
\begin{align}\label{eqn:rescaledhypergroup}
   \mathbf{r}_i \equiv \mathsf{d}(r_i)\widehat{r}_i.
\end{align}
That is, the $\mathbf{r}_i$ describe the action of domes with the factor $\mathsf{d}(r_i)^{-1}$ stripped off.

\begin{example}{}{cleftpt3}
    Consider again the situation of a cleft group $G$ of automorphisms acting on a chiral algebra $V$ appearing at the boundary of an Abelian TQFT, $\mathcal{B}(V)\cong \mathrm{Vec}_D^{\sigma,\omega}$. Take $W=V^G$ so that $\Ver (V/V^G)\cong \mathrm{Vec}_\Gamma^{\tilde\omega}$. (Cf.\ Example \ref{ex:cleftpt1}.) We conjectured in Example \ref{ex:cleftpt2} that the dome algebra in this case is given by a Mason-Ng algebra $\mathscr{D}_{\tilde\omega}(\Gamma,D)$ \cite[Definition A.5]{Gannon:2024tcl}. We can now make a conjectural identification between the basis elements.
    
    \hspace{.3in}The underlying vector space of $\mathscr{D}_{\tilde\omega}(\Gamma,D)$ is $\mathbb{C}[\Gamma]\otimes \mathbb{C}[G]$, with basis elements which we label as $\Delta_xg$. Suppressing the choice of point junctions $v$ (which are unique up to multiplication by a complex number), we assert that one should identify the $\Delta_{x}g$  with dome diagrams of Figure \ref{fig:domehypergroup} according to
    \begin{align}
        \Delta_xg \equiv \mathsf{D}_{g^{-1}xg,g}^{x}, \ \ \ \ \ \ \  x\in \Gamma, g\in G.
    \end{align}
    The composition of two domes $\Delta_xg$ and $\Delta_yh$ should clearly be proportional to the dome defined by $\delta_{g^{-1}xg,y}\Delta_xgh$, i.e.\ 
    \begin{align}
        \Delta_xg\otimes \Delta_yh \propto \delta_{g^{-1}xg,y}\Delta_xgh.
    \end{align}The identification with the Mason-Ng algebra gives the constant of proportionality, which we do not write explicitly here.
    % %
\end{example}

\begin{figure}
\begin{center}
    \input{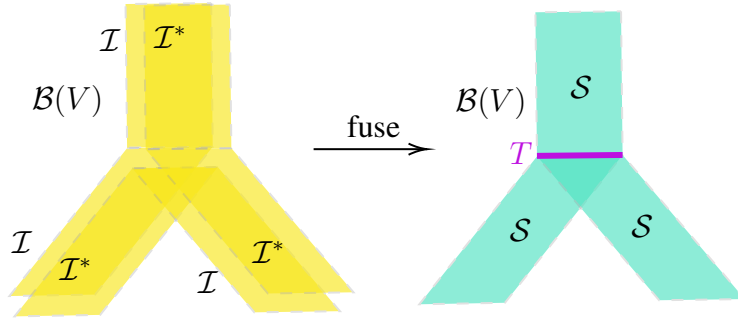}
    \caption{The definition of the $T$-junction line operator.  }\label{fig:Tjunction}
\end{center}
\end{figure}

Next, let us understand the composition of hypergroup elements using the dome picture. To this end, it is helpful to define a trivalent topological line junction, called the $T$-junction in e.g.\ \cite{Carqueville:2018sld,Mulevicius:2020tgg,Mulevicius:2020bat,Mulevicius:2022gce,Heinrich:2025wkx}, between three copies of the surface operator $\mathcal{S}$. This $T$-junction plays a role in our approach which is similar to that of symmetry fractionalization in the setting of $G$ symmetry enrichment \cite{Barkeshli:2014cna,Delmastro:2022pfo,Brennan:2022tyl}. In our setting, it can be defined implicitly by assembling a network of interfaces $\mathcal{I}$ and $\mathcal{I}^\ast$ as in the left of Figure \ref{fig:Tjunction}, and then fusing each interface $\mathcal{I}$ with its orientation reversal. We sometimes use the notation ${_k}T_{ij}$ to denote the components of $T$ which sit at the trivalent junction of $\mathcal{S}_{r_i}$, $\mathcal{S}_{r_j}$, $\mathcal{S}_{r_k}$, and similarly $T_{ij}$ for the components of $T$ which sit at the trivalent junction of $\mathcal{S}_{r_i}$, $\mathcal{S}_{r_j}$, $\mathcal{S}$. 

\begin{figure}
    \begin{center}
    \input{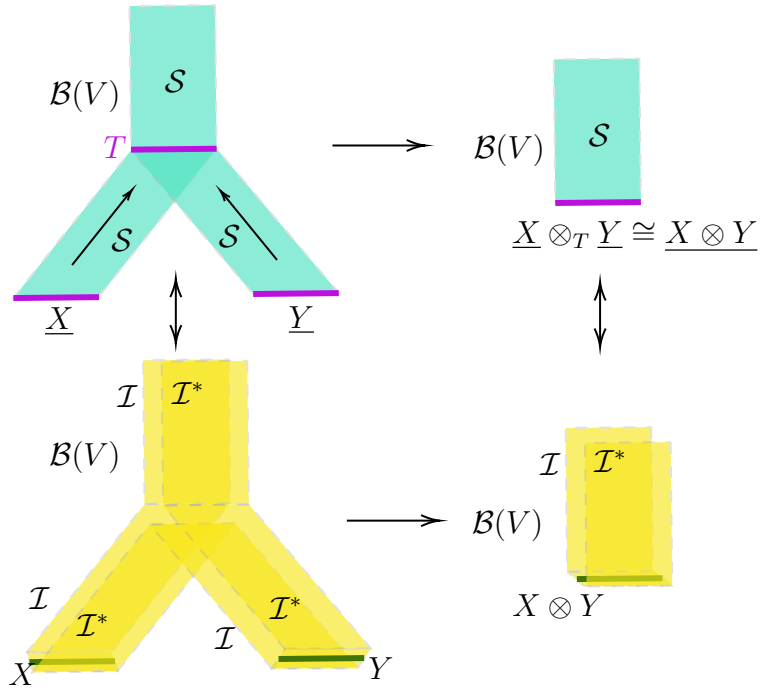}
    \caption{The $T$-junction defines a tensor product $\otimes_T$ on the category of lines which appear on the boundary of $\mathcal{S}$. This tensor product agrees with the fusion of lines on $\mathcal{I}$. }\label{fig:Tfusion}
    \end{center}
\end{figure}

The $T$-junction is useful for a variety of purposes. For example, it can be used to define a tensor product $\otimes_T$ on the category of line operators which live on the boundary of the surface $\mathcal{S}$, as illustrated on the top row of Figure \ref{fig:Tfusion}. Recall that lines on $\mathcal{I}$ can be identified with lines on the boundary of $\mathcal{S}$ via the pinching functor $X\mapsto \underline{X}$ defined in Figure \ref{fig:pinch}. The tensor product $\otimes_T$ is compatible with the pinching functor in the sense that 
\begin{align}\label{eqn:Tfusion}
    \underline{X\otimes Y}\cong \underline{X}\otimes_T \underline{Y}.
\end{align}
See Figure \ref{fig:Tfusion} for a visual illustration of Equation \eqref{eqn:Tfusion}. 

We note that, if $X\in \mathcal{B}(V)_{r_i}$ and $Y\in \mathcal{B}(V)_{r_j}$, then we can recover $\pi_k(X\otimes Y)$ (i.e.\ the part of $X\otimes Y$ which resides in the graded component $\mathcal{B}(V)_{r_k}$, cf.\ Equation \eqref{eqn:hypergroupgrading}), by replacing the three surfaces in Figure \ref{fig:Tfusion} with $\mathcal{S}_{r_i}$, $\mathcal{S}_{r_j}$, and $\mathcal{S}_{r_k}$, as in Figure \ref{fig:pikXY}.

\begin{figure}
    \centering
    \input{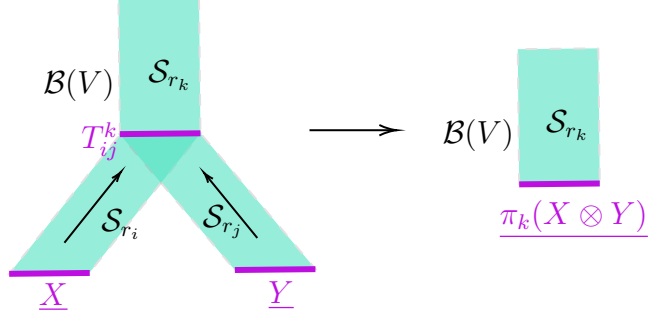}
    \caption{Replacing the outgoing surface $\mathcal{S}=\bigoplus_{r_k}\mathcal{S}_{r_k}$ with one of its components $\mathcal{S}_{r_k}$ has the effect of projecting $X\otimes Y$ onto its part $\pi_k(X\otimes Y)$ which resides in the graded component $\mathcal{B}(V)_{r_k}$.}
    \label{fig:pikXY}
\end{figure}

Importantly, the $T$-junction controls how two $\mathcal{S}$-surfaces recombine in the presence of the interface $\mathcal{I}$. Indeed, one can resolve each $\mathcal{S}$ into $\mathcal{I}\otimes\mathcal{I}^\ast$, deform the resulting configuration of interfaces, and then refuse everything together, as is sketched in Figure \ref{fig:Tmove}.

We can use the recombination rule of Figure \ref{fig:Tmove} to determine how the composition of hypergroup elements is realized in terms of domes. In order to make it easier to draw pictures, let us represent the action of the hypergroup at the level of states instead of operators, by using the state/operator correspondence of Figure \ref{fig:gaplesschiralboundary}. Then, two consecutive hypergroup elements can be represented as acting on the Hilbert space that the TQFT $\mathcal{B}(V)$ associates to the disk with boundary $V$ imposed, see the left of Figure \ref{fig:domefusion}. We can partially fuse $\mathcal{S}_{r_i}$ and $\mathcal{S}_{r_j} $ together using the recombination rule of Figure \ref{fig:Tmove}, and then squash the spherical shell which is left over to produce a point operator $\mathbf{P}_{ij}$ on the surface $\mathcal{S}$, as in the right of Figure \ref{fig:domefusion}.

\begin{figure}
    \begin{center}
        \input{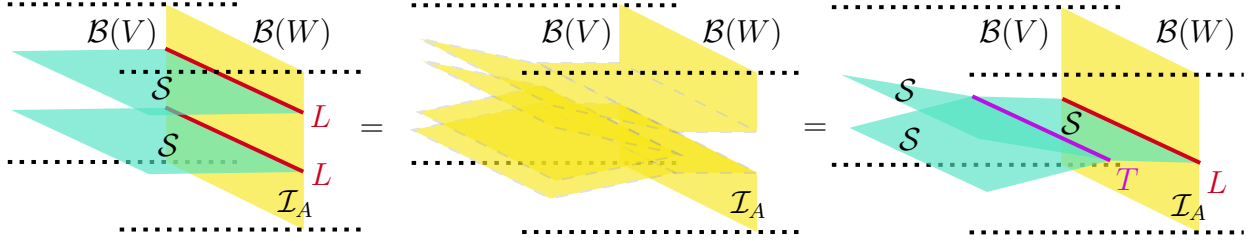}
        \caption{Two $\mathcal{S}$ surfaces terminating on the interface $\mathcal{I}_A$ can be fused together using the $T$-junction.}\label{fig:Tmove}
    \end{center}
\end{figure}

This point operator encodes the (rescaled) structure constants of the hypergroup.\footnote{The suggestion that the point operators $\mathbf{P}_{ij}$ should be related to the effective hypergroup was given to us by Ingo Runkel.} Indeed, each surface $\mathcal{S}_{r_i}$ supports a unique topological point operator (the identity $\mathds{1}_i$), so a point operator on $\mathcal{S}\cong \bigoplus_{r_i} \mathcal{S}_{r_i}$ can be written as 
\begin{align}
   \mathbf{P}_{ij} = \sum_{k} \mathbf{P}_{ij}^k \mathds{1}_k,
\end{align} 
for some collection of numbers $\mathbf{P}_{ij}^k$. The operators acting on $V$ from Equation \eqref{eqn:rescaledhypergroup} fuse according to these numbers, 
\begin{align}
    \mathbf{r}_i\cdot \mathbf{r}_j=\sum_{k}\mathbf{P}_{ij}^k \mathbf{r}_k.
\end{align}
The $\mathbf{P}_{ij}^k$ are related to the stochastically normalized structure constants of the hypergroup by incorporating the $\mathsf{d}(r_i)^{-1}$ factors, 
\begin{align}\label{eqn:rescaledhypergrouprelation}
    \widehat{r}_i\cdot \widehat{r}_j = \sum_k P_{ij}^k \widehat{r}_k, \ \ \ \ \ \ \ P_{ij}^k=\frac{\mathsf{d}(r_k)}{\mathsf{d}(r_i)\mathsf{d}(r_j)}\mathbf{P}_{ij}^k, \ \ \ \ \ \ \ \sum_k P_{ij}^k=1.
\end{align}Because they arise from point operators, which can be rescaled by arbitrary complex numbers, there is no a priori reason why the $P_{ij}^k$ or the $\mathbf{P}_{ij}^k$ need to be integers. (In fact, as we will see by way of example in Section \ref{subsec:G21}, it is sometimes the case that there is no rescaling of the $r_i$ at all which makes the structure constants integers.)

Let us combine the various ingredients so far to demonstrate that $\mathcal{C}$ is a hypergroup graded extension of $\mathcal{B}(V)$ in the sense of Equation \eqref{eqn:hypergroupgrading}. Choose lines $X\in\mathcal{B}(V)_{r_i}$ and $Y\in \mathcal{B}(V)_{r_j}$. First, by combining the pinching trick (Figure \ref{fig:pinch}) with the fact that the projection of $X\otimes Y$ onto $\mathcal{B}(V)_{r_k}\subset\mathcal{C}$ can be implemented as in Figure \ref{fig:pikXY}, we obtain the following characterization of $\dim(\pi_k(X\otimes Y))$:
\begin{center}
    \input{Figures/dimkXY1}
\end{center}
This same quantity can be computed another way, by contracting the $X$ and $Y$ lines to points (incurring factors of $\gamma(X)$ and $\gamma(Y)$ as in Equation \eqref{eqn:hypact2}), and then using the hypergroup composition rule from Figure \ref{fig:domefusion} to obtain
\begin{center}
\input{Figures/dimkXY2}
\end{center}
Equating these two expressions and using Equations  \eqref{eqn:gammadimd} and \eqref{eqn:rescaledhypergrouprelation} we find that 
\begin{align}
    \dim(\pi_k(X\otimes Y))= \dim(X)\dim(Y) P_{ij}^k,
\end{align}
which is exactly the defining equation of a hypergroup grading, Equation \eqref{eqn:hypergroupgrading}.

The upshot is that one has a choice when presenting a noninvertible symmetry of $V$. One can work with a fusion ring (namely, the Grothendieck ring of $\CC$) with integer structure constants, at the price of the action on $V$ possessing a kernel. This is the structure one naturally encounters when one works with topological line operators on the $V$-boundary. Alternatively, one can quotient out the kernel and work with the effective hypergroup, wherein one encounters possibly non-integerizable structure constants. This is the structure which arises when one trades lines for domes using the pinching functor of Figure \ref{fig:pinch}.

\begin{figure}
    \begin{center}
        \input{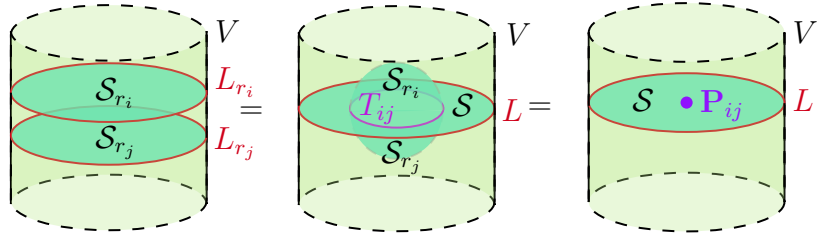}
        \caption{The composition law of the effective hypergroup.}\label{fig:domefusion}
    \end{center}
\end{figure}

\subsection{Behavior under topological manipulations}\label{subsec:topman}

In full, modular invariant conformal field theories, topological line operators behave in a universal manner under topological manipulations, like orbifolding \cite{Dixon:1985jw,Dixon:1986jc,Fuchs:2002cm,Bhardwaj:2017xup}. The simplest example goes back to \cite{Vafa:1989ih}, where it was observed that if one orbifolds a 2D theory $\mathcal{T}$ with a $\mathbb{Z}_n$ symmetry, then the orbifolded theory $\mathcal{T}/\mathbb{Z}_n$ is guaranteed by the structure of the gauging to have a ``dual'' $\widehat{\mathbb{Z}}_n$ symmetry, where $\widehat{G}$ denotes the Pontryagin dual of $G$. More generally, if $\mathcal{T}$ has a fusion category $\mathcal{C}$ of topological line operators, one may gauge a gaugeable algebra object $A$ inside $\mathcal{C}$, in which case the theory $\mathcal{T}/A$ is guaranteed to have a dual symmetry category ${_A}\mathcal{C}_A$ given by the category of $A$-$A$ bimodules in $\mathcal{C}$ \cite{Fuchs:2002cm,Bhardwaj:2017xup}.

The basic procedure of performing an orbifold consists of two steps. (In the math literature on VOAs, the word ``orbifold'' often refers just to the first step.) In the first step, one shrinks the operator algebra by passing to the $\mathcal{C}$-invariant operators of $\mathcal{T}$, and in the second step, one re-enlarges the operator algebra by adding back in suitable twisted sectors. The need to perform both steps together is dictated by modular invariance, but in the more permissive setting of chiral algebras and relative quantum field theories, where modular invariance is spoiled anyways by the existence of a non-trivial bulk, we may contemplate performing the first and second steps separately. Each step may be interpreted as a kind of gauging of the 3D bulk TQFT in the presence of a boundary condition described by a chiral algebra or relative QFT. One may then ask whether a category of boundary topological line operators behaves universally under such topological manipulations.

\paragraph{Behavior under conformal extension/one-form symmetry gauging}

\begin{figure}
\begin{center}
    \input{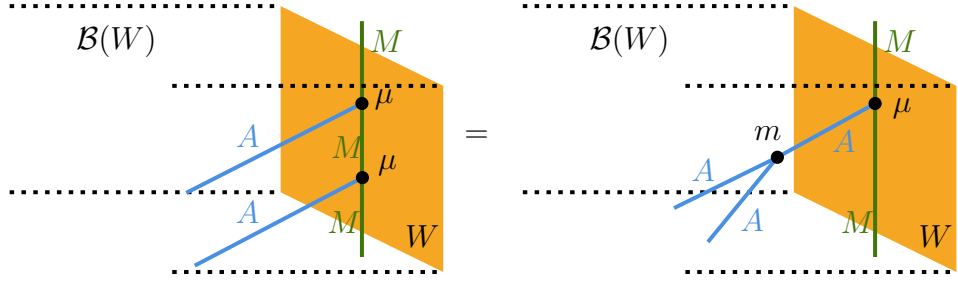}
    \caption{The boundary topological line operators after gauging a condensable algebra $(A,m)$ in the bulk/boundary system are described by $A$-modules $(M,\mu)$. Here, $m:A\otimes A\to A$ is the multiplication map on the algebra, and $\mu:M\otimes A\to M$ is the $A$-module map for $M$. The associativity constraint is visualized.}\label{fig:Amodule}
\end{center}
\end{figure}

The easier of the two steps to analyze is the second one, which in the language of chiral algebras corresponds to conformal extension/gauging noninvertible one-form symmetries of the bulk/boundary system.  

In more detail, suppose one is given a rational chiral algebra $W$ with a category $\mathcal{C}$ of boundary topological line operators containing the bulk lines $\mathcal{B}(W)$. Assume further that there is a condensable algebra $A$ in $\mathcal{B}(W)\subset \mathcal{C}$, which we recall defines a way to gauge noninvertible one-form symmetries in the bulk. We would like to determine the category of topological line operators of the theory after gauging $A$, which takes $W$ to a larger chiral algebra (more precisely, a conformal extension) $V\supset W$ with $\mathcal{B}(V)\cong \mathcal{B}(W)_A^{\mathrm{loc}}$. Our main claim is the following.
\begin{claim}{}{oneformgauging}
    Suppose $\mathcal{C}$ is a category of boundary topological line operators supported on a chiral algebra $W$, and $V$ is a chiral algebra obtained from $W$ by gauging a condensable algebra object $A$ of the bulk TQFT $\mathcal{B}(W)$ in the presence of its $W$-boundary. Then the category of boundary topological line operators of $V$ which is guaranteed by the structure of the gauging is $\mathcal{C}_A$, 
    the category of $A$-modules in $\mathcal{C}$, which has quantum dimension 
    \begin{align}
        \dim(\mathcal{C}_A)=\dim(\mathcal{C})/\dim(A), \ \ \ \ \dim(\mathcal{C})=\sum_{X\in\mathrm{Irr}(\mathcal{C})} \dim(X)^2.
    \end{align}In particular, if $\mathcal{C}\cong \Ver (W/T)$ for some conformal subalgebra $T\subset W$, then 
    \begin{align}
        \Ver (V/T)\cong \Ver (W/T)_A.
    \end{align}
\end{claim}
\noindent Note that in the special case that $T=W$, then $\Ver (W/W)\cong \mathcal{B}(W)$ and we go over to the content of Claim \ref{cl:claim:verdef}, which says that $\Ver (V/W)\cong \mathcal{B}(W)_A$.

The mathematical proof is essentially an application of \cite[Section 5.1]{Kong:2013aya}. Physically, since the gauging to go from $W$ to $V$ is implemented by inserting a fine mesh of $A$ into the bulk/boundary system, the lines of the gauged theory are those lines in $\mathcal{C}$ on which the mesh of $A$ can end consistently. From Figure \ref{fig:Amodule}, we see that those lines are described by $A$-modules.

\paragraph{Behavior under passing to fixed points/gauging zero-form symmetries}

Suppose we are given a category $\mathcal{C}$ of topological lines supported on a chiral algebra $V$. (Again, we assume that $\mathcal{C}\supset \mathcal{B}(V)$.) What can we say about the topological line operators of a chiral algebra $W=V^H$ obtained by passing to the states of $V$ which are invariant with respect to a hypergroup $H$? Eventually it should be possible to analyze this situation using the framework developed in \cite{Carqueville:2018sld} (see also \cite{Carqueville:2021edn,Mulevicius:2020tgg,KNBalasubramanian:2025vum,Carqueville:2021dbv,Mulevicius:2022gce}) for performing generalized orbifolds of 3D TQFTs. In lieu of performing this more general analysis, we content ourselves here with considering the simpler situation that $H$ is a finite group. 

\begin{figure}
    \begin{center}
        \input{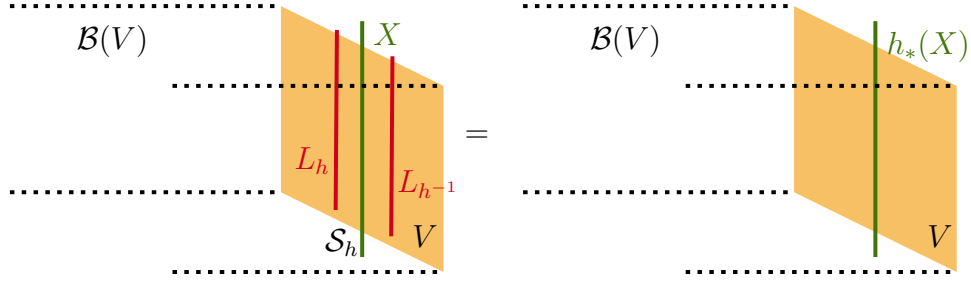}
        \caption{For each $h\in H$, there is a functor $h_\ast:\mathcal{C}\to\mathcal{C}$, with $\mathcal{C}$ the category of boundary line operators.}\label{fig:hactionfunctor}
    \end{center}
\end{figure}

Indeed, call $K=K_{\mathcal{C}}\sslash K_{\mathcal{B}(V)}$ the possibly noninvertible hypergroup induced by the topological lines $\mathcal{C}$, and consider a subhypergroup $H\subset K$ described by a finite group. In this situation, we expect to find an action of $H$ on $\mathcal{C}$ as described in \cite{Muger2010BraidedCrossedG} (see also \cite[Definition 2.3]{Galindo:2024qzg}). Let us sketch why this is so.

\begin{figure}
    \begin{center}
\input{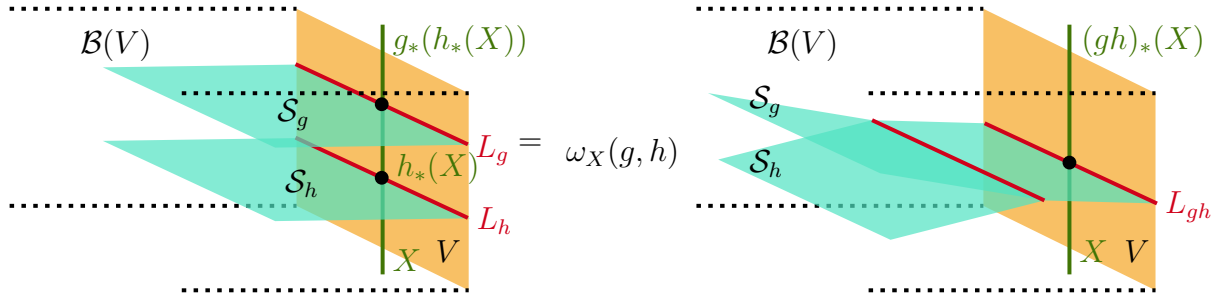}
\caption{The physical interpretation of the associativity natural isomorphisms $T^{g,h}_2(X)$.}\label{fig:T2gh}
    \end{center}
\end{figure}

First, for each $h\in H$, we obtain a functor 
\begin{align}
    h_\ast :\mathcal{C}\to \mathcal{C},
\end{align}
which acts as in Figure \ref{fig:hactionfunctor}. The functor assigns to $X$ in $\mathcal{C}$ the boundary topological line operator $h_\ast(X)$ obtained by wrapping $\mathcal{S}_h$ into a half cylinder surrounding $X$ and shrinking it. 

Additionally, there should exist natural isomorphisms 
\begin{align}
    T_2^{g,h}(X):(gh)_\ast(X)\to g_\ast(h_\ast(X)).
\end{align}
These natural isomorphisms can be encoded in phases $\omega_X(g,h)$ which arise when one fuses two surfaces $\mathcal{S}_g$ and $\mathcal{S}_h$ in the presence of the boundary topological line operator $X$, as in Figure \ref{fig:T2gh}. These natural isomorphisms are subject to coherence conditions \cite[Definition 2.3, Part (ii)]{Galindo:2024qzg} which encode the associativity of fusing three surfaces in the two different orders. 

There should also exist a tensor structure on each of the functors $h_\ast$, i.e.\ 
natural isomorphisms 
\begin{align}
    \tau^h_{X,Y}:h_\ast(X\otimes Y)\to h_\ast(X)\otimes h_\ast(Y),
\end{align}
which are subject to the coherence conditions in \cite[Definition 2.3, Part (iii)]{Galindo:2024qzg} that have straightforward physical interpretations. These natural isomorphisms can be encoded in numbers $[\eta_{X,Y}^{Z;h}]_{xy}$ which express what happens when one drags a surface $\mathcal{S}_h$ terminating on the boundary through a topological point junction from $X\otimes Y$ to $Z$, as in Figure \ref{fig:tauh}.

\begin{figure}
\begin{center}
\input{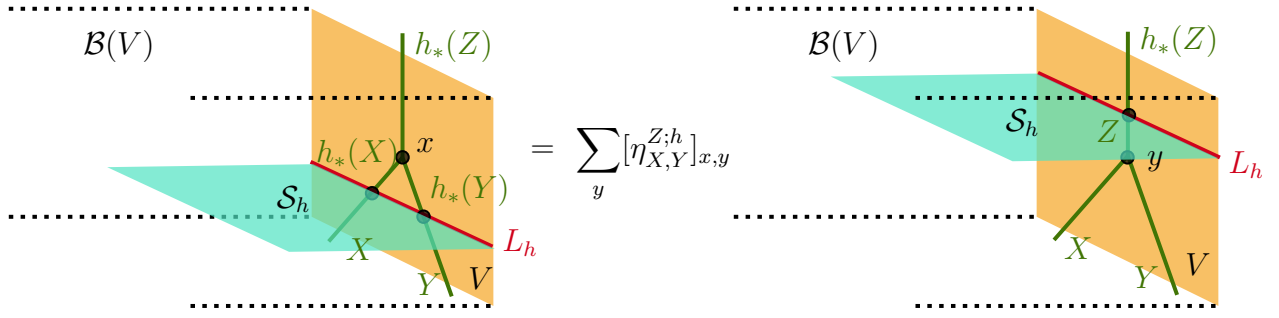}
\caption{The physical interpretation of the tensor structure natural isomorphisms $\tau^h_{X,Y}$.}\label{fig:tauh}
\end{center}
\end{figure}

Whenever a fusion category $\mathcal{C}$ admits an action by a group $H$, there is a procedure for obtaining a new fusion category $\mathcal{C}^H$ known as equivariantization, see e.g.\ \cite{burciu2013fusion} for a helpful reference. 
\begin{definition}[Equivariantization]\label{def:equivariantization}
    The $H$-equivariantization of $\mathcal{C}$, which we denote $\mathcal{C}^H$, by definition consists of pairs $(X,\varphi)$, where $X$ is a (not necessarily simple) object of $\CC$, and $\varphi$ is a family of isomorphisms $\varphi_h:h_\ast(X)\to X$ for each $h\in H$ satisfying 
    \begin{align}\label{eqn:equivariantizationCondition}
        \varphi_{gh} = \varphi_{g}\circ g_\ast(\varphi_h)\circ T_2^{g,h}.
    \end{align} 
    See \cite[Definition 2.5]{Galindo:2024qzg} for the full definition. 
\end{definition}
A useful fact is that 
\begin{align}
    \dim(\CC^H)=|H|\dim(\CC).
\end{align}
We will interpret this procedure physically shortly.

\begin{example}{$\mathbb{Z}_2$-equivariantization}{}
    The simplest example is equivariantization with respect to a $\mathbb{Z}_2$ symmetry. Indeed, call $h$ the generator of this $\mathbb{Z}_2$. If a simple object $X$ is acted upon non-trivially by $h$ (i.e.\ if $h_\ast(X)$ is not the same topological line as $X$), then $X$ and $h_\ast(X)$ merge into a single object in the equivariantization,
    \begin{align}
        (X\oplus h_\ast(X),\varphi),
    \end{align}
     with $\varphi_h=\mathrm{id}_{h_\ast(X)}\oplus T_2^{h,h}(X)$. On the other hand, if $X$ is stabilized by the action of $h$, then $X$ splits into two objects in the equivariantization,
    \begin{align}
        (X,\varphi^{\pm}),
    \end{align} 
     see \cite[Example 2.6]{Galindo:2024qzg} for details.
\end{example}

Now, given a strongly rational vertex operator algebra $V$ with an action of a finite group $H$ by symmetries, and assuming that $V^H$ is strongly rational, it is known that the representation category of the fixed-point vertex operator subalgebra $V^H$ is given by applying equivariantization to the fusion category $\mathrm{TwRep}_H(V)$ of $H$-twisted modules of $V$. That is,
\begin{align}\label{eqn:Rep(VH)}
    \Rep(V^H) = \mathrm{TwRep}_H(V)^H.
\end{align}
Alternatively, one may say that the bulk TQFT $\mathcal{B}(V^H)$ is obtained by gauging the $H$ zero-form symmetry induced on $\mathcal{B}(V)$ by the action of $H$ on $V$. We claim that this statement can be generalized as follows.
\begin{claim}{}{claim:0gauging}
    Suppose that $V$ supports a fusion category $\mathcal{C}$ of topological line operators which contains the lines $\mathcal{B}(V)$ coming from the bulk. Suppose further that $H$ is a finite group arising inside of the (possibly noninvertible) hypergroup $K=K_{\mathcal{C}}\sslash K_{\mathcal{B}(V)}$ induced by $\mathcal{C}$, i.e.\ $H\subset K$. Then $\mathcal{C}$ admits an action of $H$, and the equivariantization $\mathcal{C}^H$ describes a fusion category of topological line operators supported on the chiral algebra $V^H$ obtained by passing to $H$ fixed-points of $V$. In particular, thinking of $\mathcal{C}=\Ver(V/V^K)$ as being obtained from the conformal embedding $V^K\subset V$, one has that 
    \begin{align}\label{eqn:Hequiv}
        \Ver(V^H/V^K)=\Ver(V/V^K)^H.
    \end{align}
    Equation \eqref{eqn:Rep(VH)} corresponds to the special case that $K=H$. 
\end{claim}
We have already argued that $\mathcal{C}$ admits an action by $H$, so all that is left to do is to argue that the equivariantization $\mathcal{C}^H$ describes a category of topological line operators supported on the chiral algebra $V^H$. The basic idea is that the passage from $V$ to $V^H$ is implemented by performing a gauging of the $H$ zero-form symmetry in the bulk $\mathcal{B}(V)$ in the presence of the boundary defined by $V$. Gauging a zero-form symmetry corresponds to inserting a mesh of the surface $\widetilde{\mathcal{S}}=\bigoplus_{h\in H}\mathcal{S}_h$, so determining the boundary lines in the gauged theory amounts to imposing that this mesh be able to consistently end on lines in $\mathcal{C}$. 

This is precisely what equivariantization achieves. The line $(X,\varphi)$ can be thought of as a boundary line $X\in\mathcal{C}$ equipped with a choice of topological point junction $\varphi_h$ describing how the surface $\mathcal{S}_h$ intersects $X$ on the boundary. The condition \eqref{eqn:equivariantizationCondition} in the definition of equivariantization amounts to imposing that surfaces can be freely fused on the boundary with impunity.

Before moving on, we note that one can always ``undo'' topological manipulations \cite{Vafa:1989ih,Gaiotto:2020iye,Diatlyk:2023fwf}, even in the relative setting. For example, suppose we are in the setting of Claim \ref{cl:claim:0gauging} and are interested in ``undoing'' the $H$ zero-form gauging of the bulk/boundary theory $\mathcal{B}(V)/V$. The category $\mathcal{B}(V^H)$ which describes anyons in the bulk TQFT after gauging the $H$ zero-form symmetry possesses objects $\mathcal{L}_\rho\sim (\mathds{1},\rho)$ for each irreducible representation $\rho$ of $H$. (The irreducible representations $\rho$ parametrize the different $H$-equivariant structures on the identity line.) These lines enjoy the fusion rules of the fusion category $\Rep(H)$, and moreover 
\begin{align}
    A_{\Rep(H)}=\bigoplus_\rho \dim(\rho)\cdot  \mathcal{L}_{\rho}
\end{align}
admits the structure of a condensable algebra in $\mathcal{B}(V^H)$. If one gauges the noninvertible one-form symmetry prescribed by $A$ in the bulk/boundary theory $\mathcal{B}(V^H)/V^H$, then one goes back to the original chiral algebra $V$. Moreover, if one keeps track of what happens at the level of boundary topological line operators, one finds that
\begin{align}
\begin{split}
\text{Bulk TQFT: }& \mathcal{B}(V)\xrightarrow{\text{gauge } H^{(0)}} \mathcal{B}(V^H) \xrightarrow{\text{gauge }A_{\mathrm{Rep}(H)}^{(1)}} \mathcal{B}(V) \\
\text{Boundary chiral algebra: }& V\xrightarrow{\text{gauge } H^{(0)}} V^H \xrightarrow{\text{gauge }A_{\mathrm{Rep}(H)}^{(1)}}V  \\ 
 \text{Boundary lines: }&   \mathcal{C} \xrightarrow{\text{gauge } H^{(0)}}\mathcal{C}^H \xrightarrow{\text{gauge }A_{\mathrm{Rep}(H)}^{(1)}}(\mathcal{C}^H)_{A_{\Rep(H)}} \cong \mathcal{C}.
 \end{split}
\end{align}
That is, by applying Claim \ref{cl:claim:0gauging} and Claim \ref{cl:oneformgauging} sequentially, one recovers the category $\mathcal{C}$ of topological line operators on $V$ that one started with. 
\subsection{Symmetry-resolved partition functions}\label{subsec:symmetryresolvedpf}

Let us now briefly sketch a theory of symmetry-resolved partition functions for relative 2D CFTs and chiral algebras, generalizing the treatment in \cite{Lin:2022dhv,Choi:2024tri} for ordinary CFTs. The symmetry-resolved partition functions of a chiral algebra $V$ might also be called its ``twisted characters.''

To set the stage, recall that, given a representation $R$ of a finite group $G$, there are two ways to encode its data. One may either keep track of the non-negative integers $a_i$ arising in the decomposition of $R$ into irreducible representations $R_i$, 
\begin{align}
    R\cong \bigoplus_i a_i  R_i,
\end{align}
or one may keep track of the traces of group elements over  $R$ (i.e.\ the characters), 
\begin{align}
    \chi_R(g) \equiv \mathrm{Tr}_R (g).
\end{align}
Of course, one may go back and forth between the $\{a_i\}$ and the class function $\chi_R$ if one has knowledge of the irreducible characters $\chi_{R_i}(g)$ of $G$, e.g.\ 
\begin{align}\label{eqn:finitegroupcharacterrelation}
    \chi_R(g) = \sum_i a_i \chi_{R_i}(g).
\end{align}
Following the terminology of \cite{Choi:2024tri}, we will refer to the $a_i$ as presenting $R$ in the ``representation basis'' and the function $\chi_R$ as presenting $R$ in the ``symmetry basis''. Orthogonality relations enjoyed by the $\chi_{R_i}$ then allow one to invert Equation \eqref{eqn:finitegroupcharacterrelation}.

We can apply this way of thinking to the extended Hilbert space $\mathbb{V}_\CC$ in Equation \eqref{eqn:extendedHilbertspace}. Recall that this Hilbert space is acted on by the dome algebra $\mathrm{Dome}_{\mathcal{B}(V)}(\CC)$. By virtue of the decomposition in Equation \eqref{eqn:SchurWeyl}, we see that the ordinary modules $W_\mu$ of $W$ precisely give the representation basis presentation of $\mathbb{V}_\CC$. That is, if we grade $W_\mu$ by $L_0$ eigenvalues, 
\begin{align}
    W_\mu = \bigoplus_{h} W_{\mu,h},
\end{align}
then $\dim(W_{\mu,h})$ gives the multiplicity with which the irreducible representation $J^\mu$ of the dome algebra appears in the extended Hilbert space $\mathbb{V}_\CC$ at conformal dimension $h$. Thus, the characters of $W$, 
\begin{align}
    \mathbf{Z}_\mu(\tau) \equiv \mathrm{Tr}_{W_\mu}q^{L_0-c/24}, \ \ \ \ \ \  q=e^{2\pi i\tau}, \ \mu\in \Rep(W),
\end{align}
are the symmetry-resolved partition functions of $V$ in the representation basis. 

Alternatively, we may present the data of the action of the dome algebra on $\mathbb{V}_\CC$ in the symmetry basis,
\begin{align}\label{eqn:symmetrybasispf}
    Z_X^{k,v}(\tau) = \mathrm{Tr}_{V_X}\mathsf{D}_{X,k}^{X,v}q^{L_0-c/24}=~\input{Figures/symmetrybasispf}~, \ \ \ \ \ \ X\in\CC, \ k\in K, \ q=e^{2\pi i \tau}.
\end{align}
The trace with $\mathsf{D}_{X,k}^{X,v}$ inserted can be represented as a partition function (with topological point, line, and surface insertions as in Equation \eqref{eqn:symmetrybasispf}) of the 3D TQFT $(\mathcal{B}(V),c(V))$ on a solid torus with the boundary condition $(V,\Phi_V)$ imposed and with boundary complex structure $\tau$.

As in the case of finite groups, one can interpolate between these two bases using character theory. In particular, extending the arguments of \cite{Choi:2024tri} mutatis mutandis, one finds that 
\begin{align}
    Z_{X}^{k,v}(\tau) = \sum_{\mu\in\Rep(W)} [\chi_\mu]_X^{k,v}\mathbf{Z}_\mu(\tau)
\end{align}
where $[\chi_\mu]_X^{k,v}$ is a character of the dome algebra which may be represented as in Figure \ref{fig:domecharacters}. 

\begin{figure}
    \begin{center}
        \input{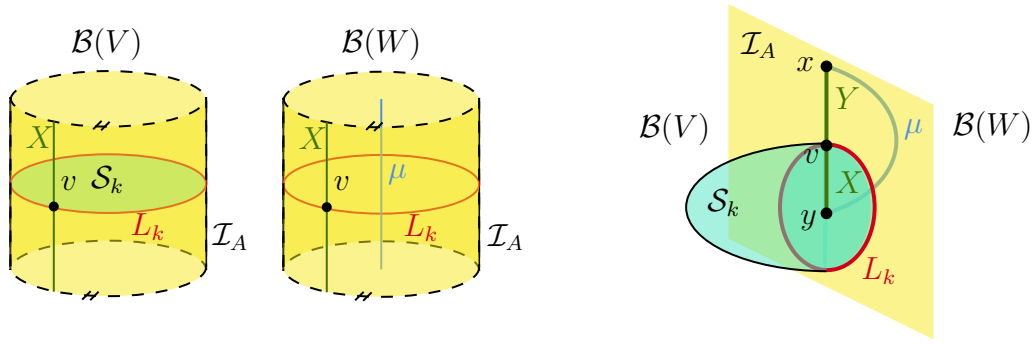}
        \caption{The character $[\chi_\mu]_X^{k,v}$ can be represented as the $S^2\times S^1$ partition function obtained by gluing together the two solid tori on the left along their common boundary (making sure to glue $X$ to $X$ and $L_k$ to $L_k$). Alternatively, $[\chi_\mu]_{X}^{k,v}$ can be reduced to evaluating configurations like those on the right (i.e.\ determining the multiple of the identity operator obtained when one shrinks the configuration down to a point operator).}\label{fig:domecharacters}
    \end{center}
\end{figure}

\begin{example}{}{twisteddim}
In this paper, we will mainly be interested in computing just the graded-dimension of the twisted-sector Hilbert space $V_X$, i.e.\ without any insertions of domes into the trace,
\begin{align}\label{eqn:twistedpf}
    Z_X(\tau) = \mathrm{Tr}_{V_X}q^{L_0-c/24}
 = ~\input{Figures/gradeddimension}~. \end{align}
Equation \eqref{eqn:twisteddecomposition} tells us that this can be expanded into ordinary characters $\mathbf{Z}_\mu(\tau) = \mathrm{Tr}_{W_\mu} q^{L_0-c/24}$ of $W$ as
\begin{align}\label{eqn:twistedpfB}
    Z_X(\tau) = \sum_{\mu\in\Rep(W)}B_{X\mu} \mathbf{Z}_\mu(\tau)
\end{align}
where $B_{X\mu}\equiv \dim J^\mu_X$ is the dimension of the Hilbert space of topological point junctions between the bulk line $\mu \in \mathcal{B}(W)$ and the interface line $X\in\CC$ (cf.\ Figure \ref{fig:SymTFTtwistedops}). We will explain practical methods for computing these integers $B_{X\mu}$ shortly. 
\end{example}

\subsection{Galois theory}\label{subsec:Galois}

As we have been discussing, there is a close relationship between the conformal subalgebras of a chiral algebra and its fusion categories of topological line operators. In this subsection, we will show how this relationship can be strengthened to a Galois correspondence, in a similar spirit as previous work in the conformal net and vertex operator algebra literatures \cite{dong1997quantum,dong2013quantum,Bischoff:2016jmy,Dong:2025ttr} (see also \cite{Moller:2024xtt} for a discussion in the setting of holomorphic vertex operator algebras).

Before we turn to vertex operator algebras, let us briefly review the original example of a Galois correspondence. Given a field extension $E/F$, one is interested in studying the intermediate fields $K$ which sit between $E$ and $F$, i.e.\ which satisfy $F\subset K \subset E$. In favorable circumstances, (namely, when the field extension $E/F$ is finite and Galois), the intermediate fields are controlled by the Galois group $\mathrm{Gal}(E/F)$, which by definition consists of the automorphisms of $E$ which fix $F$ point-wise. In particular, there is a one-to-one correspondence between intermediate fields $K$ and subgroups $H$ of the Galois group,
\begin{align}
\begin{split}
    H\mapsto E^H&=\{e\in E\mid \phi(e)=e, \forall \phi \in H\},  \\
    K \mapsto \mathrm{Aut}(E/K)&=\{\phi\in \mathrm{Aut}(E)\mid \phi(k)=k, \forall k\in K\}.
\end{split}
\end{align}
This correspondence satisfies a number of useful properties. For example, it is inclusion-reversing, meaning that if $H\subset H'$ then $E^{H}\supset E^{H'}$. Also, the degrees of the field extensions entering the correspondence are controlled by group theory, in the sense that 
\begin{align}
    |H|=[E:E^H], \ \ \ \ \  |\mathrm{Gal}(E/F)|/|H|=[E^H:F],
\end{align}
where e.g.\ $[E:F]$ is defined to be the dimension of $E$ thought of as a vector space over $F$. 
Since the inception of this original Galois correspondence, the essential structure has been abstracted and found in many mathematical situations: for example, in the field of algebraic topology, there is a Galois connection between covering spaces of a path-connected topological space and subgroups of its fundamental group.

\begin{table}
\begin{center}
    \begin{tabular}{c|c}
    Galois theory & Vertex operator algebra analog \\\toprule
    Field $E$ & VOA $V$ \\\midrule
    Finite, Galois field extension $E/F$ & Conformal extension of strongly rational VOAs $V/W$ \\\midrule
    \multirow{2}{*}{Galois group $\mathrm{Gal}(E/F)$} & Verlinde category $\Ver (V/W)$ \\
    & or hypergroup $K=K_{\Ver (V/W)}\sslash K_{\mathcal{B}(V)}$\\\midrule
    \multirow{2}{*}{Subgroup $H<\mathrm{Gal}(E/F)$} & Fusion subcategory $\mathcal{C}\subset\Ver (V/W)$ with $\mathcal{C}\supset \mathcal{B}(V)$\\
    & or subhypergroup $L\subset K$ \\\midrule 
     \multirow{2}{*}{$|H|$} & $|\mathcal{C}|=\dim(\mathcal{C})/\dim(\Rep(V))$ \\
    & or $D(L)$ (Equation \eqref{eqn:weighthypergroup})\\\midrule
    \multirow{2}{*}{Fixed-point subfield $E^H$} & Fixed-point subalgebra $V^{\mathcal{C}}$ \\ & or subalgebra $V^L$ \\\midrule 
    Degree $[E:F]$ & Degree $[V:W]=\sqrt{\dim(\Rep(W))/\dim(\Rep(V))}$ 
    \end{tabular}
    \caption{The analogy between Galois theory of fields and Galois theory of vertex operator algebras. }\label{tab:Galoistheory}
\end{center}
\end{table}

By now, the analogy to vertex operator algebras should be clear. We summarize its basic features in Table \ref{tab:Galoistheory}. The key idea is that the role of finite Galois field extensions $E/F$ is played by conformal embeddings of strongly rational vertex operator algebras $V/W$, and the role of the Galois group $\mathrm{Gal}(E/F)$ should be played by the Verlinde fusion category $\Ver (V/W)$. The role of the degree of a field extension is played by 
\begin{align}
    [V:W]=\sqrt{\frac{\dim(\Rep(W))}{\dim(\Rep(V))}}=\dim(A)=\lim_{\tau\to 0}\frac{\mathrm{ch}_V(\tau)}{\mathrm{ch}_W(\tau)},
\end{align}
where $A$ is the condensable algebra in $\Rep(W)$ defined by the conformal extension $W\subset V$, and e.g.\ $\mathrm{ch}_V(\tau) = \mathrm{Tr}_V q^{L_0-c/24}$. Our convention for the global dimension of a fusion category $\CC$ is 
\begin{align}
    \dim(\mathcal{C})=\sum_{X\in\mathrm{Irr}(\CC)}\dim(X)^2.
\end{align} The role of subgroups of the Galois group is played by fusion subcategories $\mathcal{C}\subset\Ver (V/W)$ which contain $\mathcal{B}(V)$, and their ``order'' is defined to be 
\begin{align}
    |\mathcal{C}|=\dim(\mathcal{C})/\dim(\mathcal{B}(V)).
\end{align}
With these definitions in place, it is a matter of following one's nose. 

\begin{claim}{Noninvertible Galois theory for chiral algebras}{}
    The intermediate vertex operator algebras $U$ satisfying $W\subset U\subset V$ are in one-to-one correspondence with fusion subcategories $\mathcal{C}$ of $\Ver (V/W)$ which contain $\mathcal{B}(V)$, with the correspondence assigning 
\begin{align}\label{eqn:VOAgalois}
\begin{split}
   \mathcal{C}&\mapsto V^{\mathcal{C}}=\{\mathcal{O} \in V \mid \widehat{X}\cdot \mathcal{O}=\d_X \mathcal{O},~ \forall X\in\mathcal{C}\} \\
   U&\mapsto \Ver (V/U) = \{X\in\Ver (V/W)\mid \widehat{X}\cdot \mathcal{O}=\d_X\mathcal{O}, ~\forall \mathcal{O}\in U\}.
\end{split}
\end{align}
It is antitone in the sense that $U\subset U'$ if and only if  the corresponding fusion categories satisfy $\mathcal{C}'\subset \mathcal{C}$. Furthermore, 
\begin{align}
    |\mathcal{C}|=[V:V^{\mathcal{C}}], \ \ \ \ |\Ver (V/W)|/|\mathcal{C}|=[V^{\mathcal{C}}:W].
\end{align}
\end{claim}
Note that the proof of the formulae involving degrees of extensions uses Equation \eqref{eqn:dimver} below. The reason we must only consider those $\mathcal{C}$ which contain $\mathcal{B}(V)$ in this correspondence is that the categories obtained in the image of the second map in Equation \eqref{eqn:VOAgalois} always contain $\mathcal{B}(V)$. The proof essentially follows from \cite[Theorem 4.10]{davydov2013witt}, using the fact that, if $V$ is obtained from $W$ via a condensable algebra $A$ in $\mathcal{B}(W)$, then intermediate vertex operator algebras $U$ correspond to condensable subalgebras of $A$. 

It is curious that, while normal subgroups play a role in the Galois theory of field extensions, there does not seem to be a direct analog in the Galois theory of vertex operator algebras.

Note that fusion subcategories $\mathcal{C}\subset\Ver (V/W)$ which contain $\mathcal{B}(V)$ are in one-to-one correspondence with subhypergroups of $K=K_{\Ver (V/W)}\sslash K_{\mathcal{B}(V)}$ via the assignment $\mathcal{C}\mapsto K_{\mathcal{C}}\sslash K_{\mathcal{B}(V)}$. So one could equivalently formulate the Galois theory in terms of hypergroups instead of fusion categories. 

\begin{claim}{}{}
The intermediate vertex operator algebras $U$ satisfying $W\subset U\subset V$  are in one-to-one correspondence with subhypergroups $L$ of $K=K_{\Ver (V/W)}\sslash K_{\mathcal{B}(V)}$, with the correspondence assigning 
\begin{align}
    L\mapsto V^L, \ \ \ \ \ \ U\mapsto K_{\Ver (V/U)}\sslash K_{\mathcal{B}(V)}.
\end{align}
It is antitone in the sense that $U\subset U'$ if and only if the corresponding hypergroups satisfy $L'\subset L$. Furthermore, 
\begin{align}
    D(L)=[V:V^{L}], \ \ \ \ D(K)/D(L)=[V^L:W],
\end{align}
where $D(L)$ is the weight of the hypergroup $L$, defined in Equation \eqref{eqn:weighthypergroup}.
\end{claim}
\noindent This formulation in terms of hypergroups most closely connects with existing results in the conformal net and VOA literature, as the following example demonstrates.

\begin{example}{}{}
    When $G$ is a finite group of automorphisms of $V$ and $W=V^G$, the hypergroup induced by the conformal embedding $V^G\subset V$ is simply $G$ itself. Thus, the formulation of Galois theory in terms of hypergroups says that the intermediate VOAs $V^G\subset U\subset V$ are in one-to-one correspondence with subgroups $H<G$, under the assignment $H\mapsto V^H$. Furthermore, the index of the embedding $|H|=[V:V^H]$ is simply given by the order of $H$. This is the original Galois theory of \cite{dong1997quantum,dong2013quantum}.
\end{example}

\subsection{Computational aids}\label{subsec:computationalaids}

Much of the discussion so far has been dedicated to developing abstract formalism. In this last subsection, we describe a number of constraints on $\Ver (V/W)$ and $K=K_{\Ver (V/W)}\sslash K_{\mathcal{B}(V)}$ which are helpful for carrying out concrete calculations. We will make heavy use of these constraints in Section \ref{sec:examples} and in follow-on work.

First, we note that by \cite[Corollary 3.30]{davydov2013witt}, the Drinfeld center of the Verlinde category $\Ver(V/W)$ is given by
\begin{align}\label{eqn:Drinfeldcenterverlinde}
    Z(\Ver (V/W))\cong \mathcal{B}(W)\boxtimes \overline{\mathcal{B}(V)} \cong \Rep(W)\boxtimes \overline{\Rep(V)}. 
\end{align}
To see this physically, note that by folding Figure \ref{fig:SymTFTconformalembedding} along the interface $\mathcal{I}_A$, we learn that $\Ver (V/W)$ is the category of topological line operators supported on a gapped boundary condition of $\mathcal{B}(W)\boxtimes\overline{\mathcal{B}(V)}$, and it is known that the Drinfeld center of a category of boundary lines recovers the MTC in the bulk. 

Thus, one immediately has access to the Drinfeld center of $\Ver (V/W)$ once one knows the representation theory of $V$ and $W$. In particular, the quantum dimension of $\Ver (V/W)$  follows as a corollary, 
\begin{align}\label{eqn:dimver}
    \dim(\Ver (V/W))^2=\dim(\Rep(W))\dim(\Rep(V)),
\end{align}
where we have used the fact that $\dim(Z(\mathcal{C}))=\dim(\mathcal{C})^2$ for a fusion category $\mathcal{C}$. See \cite[Equation (8)]{davydov2013witt} and references therein.

One also has access to the rank once one knows how $V$-modules decompose into $W$-modules (cf.\ Equation \eqref{eqn:Vmoddecomp}). Indeed, it follows from \cite{Ostrik:2001xnt} (see also Equation (4.3) of \cite{Rie22}) that 
\begin{align}\label{eqn:Brank}
    \mathrm{rank}(\Ver (V/W))=\sum_{a\in\mathcal{B}(V)}\sum_{\mu\in\mathcal{B}(W)}(B_{a\mu})^2.
\end{align}
We emphasize that the sum over $a$ is just over the category $\mathcal{B}(V)$, not over the larger category $\Ver (V/W)$. In fact, more is true: calling $\mathcal{K}_0(\Ver(V/W))$ the Grothendieck ring of $\Ver(V/W)$,  the complexification decomposes into a direct sum of matrix algebras à la Wedderburn-Artin as 
\begin{align}
    \mathcal{K}_0(\Ver(V/W))\otimes \mathbb{C} \cong \bigoplus_{a\in\mathcal{B}(V)} \bigoplus_{\mu\in\mathcal{B}(W)} \mathrm{Mat}_{\mathbb{C}}(B_{a\mu}),
\end{align}
where $\mathrm{Mat}_{\mathbb{C}}(n)$ is the algebra of $n\times n$ complex matrices \cite{Dong:2025ttr}.  In particular, $\Ver(V/W)$ has commutative fusion rules if and only if $B_{a\mu}\leq 1$ for all $a\in\mathcal{B}(V)$ and $\mu\in\mathcal{B}(W)$. 

To help with the computation of the matrix $B$ itself, we note that it must intertwine the modular data of $V$ and $W$, i.e.\ 
\begin{align}\label{eqn:intertwinemodular}
    \sum_{\mu\in\mathcal{B}(W)}B_{a\mu}S^W_{\mu\nu} = \sum_{b\in\mathcal{B}(V)} S^V_{ab}B_{b\nu}, \ \ \ \ \sum_{\mu\in\mathcal{B}(W)}B_{a\mu}T^W_{\mu\nu} = \sum_{b\in\mathcal{B}(V)} T^V_{ab}B_{b\nu},
\end{align}
where $S^V,T^V$ are the modular S- and T-matrices, respectively, of $V$, and likewise for $S^W,T^W$. It also satisfies 
\begin{align}
    \mathrm{Tr}_{V_a}q^{L_0-c/24} = \sum_{\mu\in \mathcal{B}(W)}B_{a\mu} \mathrm{Tr}_{W_\mu}q^{L_0-c/24}, \ \ \ \ \ \ a\in\mathcal{B}(V),
\end{align}
and so it is strongly constrained if one knows the characters of $V$ and $W$.

To gain more detailed information about $\Ver (V/W)$, we note that there is an induction functor, defined as
\begin{align}\label{eqn:induction}
\begin{split}
I:\mathcal{B}(W)&\to \mathcal{B}(W)_A=\Ver (V/W)\\
\mu &\mapsto (\mu\otimes A,  \mathrm{id}_\mu \otimes m),
\end{split}
\end{align}
where $m:A\otimes A\to A$ is the multiplication morphism on $A$.
That is, $I(\mu)$ is the $A$-module which, as an object, is isomorphic to $A\otimes \mu$, and whose module map $I(\mu)\otimes A \to I(\mu)$ is given by $\mathrm{id}_\mu \otimes m$.

Physically, we can think of $I$ as a  ``bulk-to-wall'' map which describes the effect of taking a line in the bulk $\mathcal{B}(W)$ and pushing it onto the interface $\mathcal{I}_A$, as in Figure \ref{fig:induction}. We note that, unlike the bulk-to-wall map $\mathcal{B}(V)\to\Ver (V/W)$, simple lines in $\mathcal{B}(W)$ can and often do become non-simple when brought to the interface $\mathcal{I}_A$.

\begin{figure}
    \begin{center}
\input{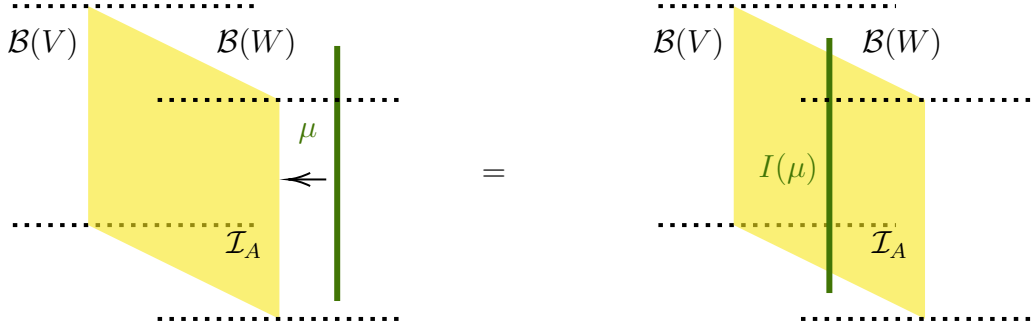}
\caption{The induction functor $I:\mathcal{B}(W)\to \Ver (V/W)$ is a bulk-to-wall map which describes the result of bringing a bulk topological line in $\mathcal{B}(W)$ to the interface $\mathcal{I}_A$.}\label{fig:induction}
    \end{center}
\end{figure}

One useful feature of the induction functor $I$ is that every simple object of $\Ver (V/W)$ occurs as a subobject of $I(\mu)$ for some $\mu\in\mathcal{B}(W)$. Thus, we gain information about all the simples of $\Ver (V/W)$ by computing $I$. Moreover, it is a tensor functor, which means in particular that 
\begin{align}\label{eqn:inductiontensorfunctor}
    I(\mu)\otimes I(\nu) \cong I(\mu\otimes \nu), \ \ \ \ \ I(\mu^\ast)\cong I(\mu)^\ast.
\end{align}
Physically, fusing in the bulk and then pushing onto the interface is the same as pushing onto the interface and then fusing.
Thus, if we know the fusion rules of the chiral algebra $W$ (which are the same as the fusion rules of $\mathcal{B}(W)$), then we can access partial information about the fusion rules of $\Ver (V/W)$.

Induction also preserves quantum dimensions, 
\begin{align}
    \dim_{\mathcal{B}(W)}(\mu) = \dim_{\mathcal{B}(W)_A}(I(\mu)).
\end{align}
Indeed, the quantum dimension of a line can be thought of as its expectation value when it is wound into a circle, and it does not matter whether we take this expectation value before or after pushing onto the interface.

Another useful property is that induction enjoys Frobenius reciprocity with restriction. That is, if one decomposes $I(\mu)$ into a direct sum of simples in $\Ver (V/W)$, 
\begin{align}\label{eqn:Frobrepconcrete}
    I(\mu) \cong \bigoplus_{X\in\Ver (V/W)}B_{X\mu} X,
\end{align}
then the numbers $B_{X\mu}$ which arise are precisely equal to the multiplicities in the decomposition of $V_X$ into $W$-modules, Equation \eqref{eqn:twisteddecomposition}. A more abstract way to encode this same information is to write 
\begin{align}\label{eqn:Frobeniusreciprocity}
    \mathrm{Hom}_{\Ver (V/W)}(I(\mu),X)\cong \mathrm{Hom}_{\mathcal{B}(W)}(\mu,R(X)),
\end{align}
where $R:\Ver (V/W)\to \mathcal{B}(W)$ is the restriction functor, which assigns to $X\in\Ver (V/W)=\mathcal{B}(W)_A$ the corresponding underlying object in $\mathcal{B}(W)$ (i.e.\ forgetting its structure as an $A$-module). That is, 
\begin{align}\label{eqn:restriction}
\begin{split}
    R:\mathcal{B}(W)_A&\to \mathcal{B}(W) \\
    (M,\alpha)&\mapsto M,
\end{split}
\end{align}
where $\alpha:M\otimes A\to M$ is the module map on $M$.
Alternatively, restriction corresponds to taking a twisted module $V_X \in \mathrm{TwRep}_W(V)$ and thinking of it simply as an ordinary $W$-module. (Cf. Figure \ref{fig:physicalrestriction} for a physical picture of the restriction functor.) 

The physical derivation of Equation \eqref{eqn:Frobeniusreciprocity} is to observe that both the left- and right-hand sides can be thought of as the Hilbert space of point junctions between a bulk line $\mu\in\mathcal{B}(W)$ and a line $X$ on the interface $\mathcal{I}_A$. Indeed, the left-hand side of Equation \eqref{eqn:Frobeniusreciprocity} corresponds to viewing Figure \ref{fig:physinterpfrobrep} from right-to-left, while the right-hand side of Equation \eqref{eqn:Frobeniusreciprocity} corresponds to viewing Figure \ref{fig:physinterpfrobrep} from left-to-right.

Another interpretation of Frobenius reciprocity is that it asserts that the number of topological point junctions between $\mu$ and $X$  is equal to the number of times that the $W$-module $W_\mu$ appears in the twisted sector $V_X$, a conclusion we already came to by thinking about the physics of the SymTFT in Section \ref{subsec:twisteddome}. Often times induction is very computable, and so gives a window into the numbers $B_{X\mu}$ which are useful for a variety of purposes.

\begin{figure}
    \centering
    \input{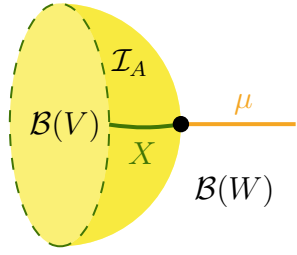}
    \caption{The physical derivation of Frobenius reciprocity, Equation \eqref{eqn:Frobeniusreciprocity}.}\label{fig:physinterpfrobrep}
\end{figure}

A useful fact we note in passing about the relationship between restriction and induction is that 
\begin{align}\label{eqn:ResInd}
    R(I(\mu))\cong  \mu\otimes A.
\end{align}

In the case of consecutive embeddings $U\subset W\subset V$, there is a natural generalization of the induction functor,
\begin{align}
   {_U}I_W^V:\Ver (W/U)\to \Ver (V/U) .
\end{align}
Suppose that the extension $U\subset W$ is mediated by a condensable algebra $B$ in $\mathcal{B}(U)$, and similarly the extension $U\subset V$ corresponds to a condensable algebra $A$ in $\mathcal{B}(U)$. Then,  $\Ver (W/U)=\mathcal{B}(U)_B$ and $\Ver (V/U)= \mathcal{B}(U)_A=(\mathcal{B}(U)_B)_A$, where we use the fact that $A$ can naturally thought of as an  algebra in $\mathcal{B}(U)_B$ (see e.g.\ the discussion below Example 5.2 of \cite{Kong:2013aya}). It is then clear that one can simply replace $\mathcal{B}(W)$ with $\mathcal{B}(U)_B$ in Equation \eqref{eqn:induction} to define an induction functor in an identical way. 

The physical interpretation of ${_U}I_W^V$ is given in Figure \ref{fig:generalizedinduction}. One immediate consequence of this physical picture is that induction functors make the following diagram commute, 
\begin{center}
    \begin{tikzcd}
& \Ver (T/U) & \\
\Ver (W/U) \arrow[rr,"{_U}I_W^V"] \arrow[ru,"{_U}I_W^T"] & & \Ver (V/U)\,,\arrow[lu,"{_U}I_V^T",swap]  
\end{tikzcd}
\end{center}
where we have assumed a sequence $U\subset W\subset V\subset T$ of conformal embeddings. 

\begin{example}{}{}
Suppose that $W$ is obtained from $V$ as the fixed points of a finite group $G$ of automorphisms of $V$, i.e.\ $W=V^G$. Then by Claim \ref{cl:claim:0gauging}, we have that $\mathcal{B}(V^G)$ and $\Ver(V/V^G)$ are related by $G$-equivariantization,
\begin{align}
    \Ver(V/V^G)^G \cong \mathcal{B}(V^G),
\end{align}
so one can think of lines in $\mathcal{B}(V^G)$ as pairs $(X,\varphi)$, where $X$ is a $G$-invariant line in $\Ver(V/V^G)$ and $\varphi$ is a choice of equivariant structure (cf.\ Definition \ref{def:equivariantization}). In this situation, the induction functor $I:\mathcal{B}(V^G)\to \Ver(V/V^G)$ is simply given by
\begin{align}\label{eqn:inductionGequiv}
   I(X,\varphi)\cong X.
\end{align}
\end{example}

\begin{figure}
 \begin{center}
\input{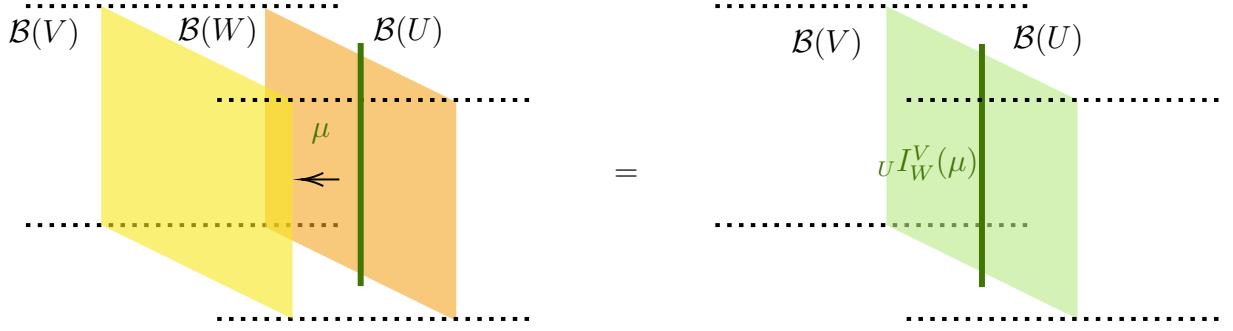}
\caption{The generalized induction functor ${_U}I_W^V$ describes the result of bringing a line $\mu$ on the condensation interface between $\mathcal{B}(U)$ and $\mathcal{B}(W)$ to the condensation interface between $\mathcal{B}(W)$ and $\mathcal{B}(V)$ to produce a topological line ${_U}I_W^V(\mu)$ on the condensation interface between $\mathcal{B}(U)$ and $\mathcal{B}(V)$. }\label{fig:generalizedinduction}
 \end{center}
\end{figure}

The action of $\Ver (V/W)$ on the extended Hilbert space $\mathbb{V}_{\Ver(V/W)}$ in Equation \eqref{eqn:extendedHilbertspace} can also be constrained by induction. Indeed, for lines which can be pulled into the bulk $\mathcal{B}(W)$, one can straightforwardly see that 
\begin{align}\label{eqn:inducedaction}
    \widehat{I(\mu)}\cdot \mathcal{O} = \frac{S_{\mu\nu}^W}{S_{1\nu}^W}\mathcal{O}, \ \ \ \ \ (\mathcal{O}\in B_{a\nu}W_\nu\subset V_a)
\end{align}
where $S^W_{\mu\nu}$ is the modular S-matrix of $\mathcal{B}(W)$.

We also note that $\mathcal{B}(V)$ must be central in $\Ver (V/W)$, that is, 
\begin{align}
    a\otimes X \cong X \otimes a, \ \ \ \text{ for all }a\in\mathcal{B}(V), \ X\in\Ver (V/W).
\end{align}
This can be easily seen because the lines $a\in\mathcal{B}(V)$ can be pulled off the interface $\mathcal{I}_A$ and into the bulk, and then pushed back onto the interface on the other side of $X$. More abstractly, this is a consequence of the fact that the bulk-to-wall functor $\mathcal{B}(V)\to \Ver(V/W)$ is a central functor, and hence can be factored through the Drinfeld center, $\mathcal{B}(V)\to Z(\Ver(V/W))\to \Ver(V/W)$.

Finally, we note that all the standard constraints on fusion categories apply to $\Ver (V/W)$ as well, including e.g.\ associativity of the fusion rules, Frobenius reciprocity,
\begin{align}\label{eqn:Frobrep2}
    N_{ab}^c = N_{c\bar b}^a = N_{\bar a c}^b = N_{b\bar c}^{\bar a} = N_{\bar c a}^{\bar b},
\end{align}
quantum dimensions furnishing an eigenvector of the fusion rules,
\begin{align}
    \sum_{Z\in\Ver (V/W)}N_{XY}^Z \mathsf{d}_Z = \mathsf{d}_X\mathsf{d}_Y,
\end{align}
and so on and so forth. See e.g.\ \cite{etingof2015tensor} for a standard reference.

\section{Symmetries and Boundaries of Absolute CFTs}\label{sec:gluing}

In the previous section, we focused on symmetries of chiral algebras (or more generally, relative QFTs). However, one of the primary applications of chiral algebras is their use as building blocks of full rational conformal field theories which contain both left-movers and right-movers. In this section, we explain how the chiral considerations of Section \ref{sec:relativesymmetries} are relevant for computing symmetries and boundary conditions of such full rational conformal field theories.

\subsection{Setting up the problem}\label{subsec:settingup}

\begin{figure}
    \begin{center}
\input{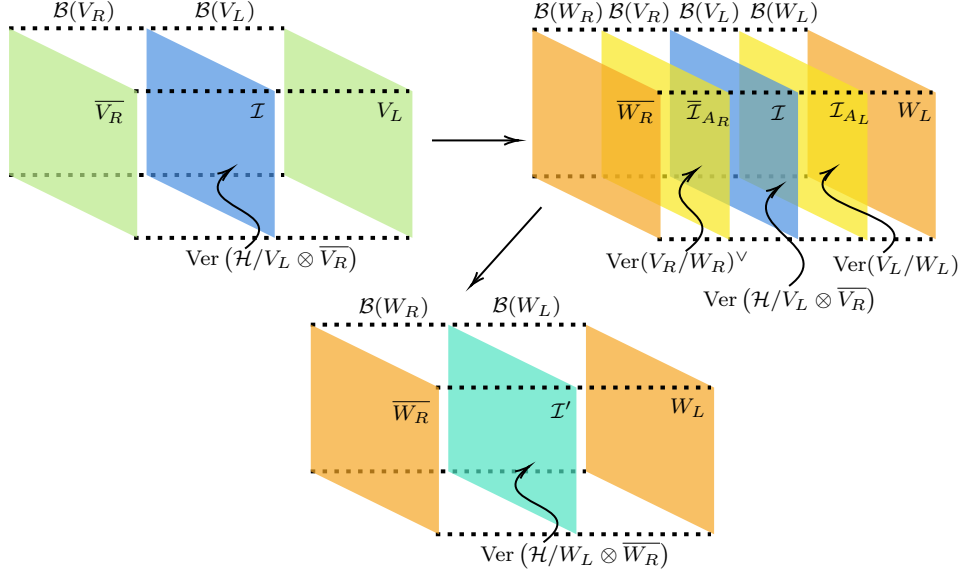}
    \end{center}
    \caption{Top left: The Kapustin-Saulina representation of a CFT built on $V_L$ and $V_R$. Top right: Obtained from top left by applying the SymTFT construction of Figure \ref{fig:SymTFTconformalembedding}. Bottom: Obtained from top right by fusing the interfaces $\mathcal{I}_{A_L}$, $\mathcal{I}$, and $\overline{\mathcal{I}}_{A_R}$. The categories $\Ver(\cdot)$ are categories of topological lines supported on the interfaces which they point to.}\label{fig:combiningchiralsymmetries}
\end{figure}

To make the problem sharper, suppose one is studying a full rational conformal field theory $\mathcal{H}$ built on rational chiral algebras $V_L$ and $V_R$,
\begin{align}\label{eqn:fullCFTHilbertspace}
    \mathcal{H}\cong \bigoplus_{a\in\mathcal{B}(V_L)}\bigoplus_{b\in\mathcal{B}(V_R)}M_{ab} ~(V_L)_a\otimes \overline{(V_R)}_{b}, \ \ \ \ M_{ab}\in\mathbb{Z}_{\geq 0},
\end{align}
and that one has calculated the category $\Ver (\mathcal{H}/V_L\otimes\overline{V_R})$ of topological line operators that preserve the chiral algebras $V_L$ and $\overline{V_R}$. 
\begin{example}{Verlinde lines in diagonal rational CFTs}{}
    In the case that $V_L=V_R$ and $\mathcal{H}$ is the canonical diagonal modular invariant theory built on $V_L=V_R$, then $\Ver (\mathcal{H}/V_L\otimes\overline{V_R})$ is the usual category of Verlinde lines, which is equivalent as a fusion category to $\mathrm{Rep}(V_L)=\mathrm{Rep}(V_R)$.
\end{example} 
\begin{example}{Verlinde lines in general rational CFTs}{}
More generally, thinking of the Hilbert space $\mathcal{H}$ as a Lagrangian algebra object in $\Rep(V_L)\boxtimes\overline{\Rep(V_R)}$, the category $\Ver (\mathcal{H}/V_L\otimes\overline{V}_R)$ is equivalent to the category of $\mathcal{H}$-modules in $\Rep(V_L)\boxtimes\overline{\Rep(V_R)}$, 
\begin{align}
    \Ver (\mathcal{H}/V_L\otimes\overline{V}_R)\cong \big(\Rep(V_L)\boxtimes\overline{\Rep(V_R)}\big)_{\mathcal{H}}.
\end{align}
\end{example}

Suppose that one later realizes that $V_L$ and $V_R$ admit strongly rational conformal subalgebras $W_L\subset V_L$ and $W_R\subset V_R$. Then $W_L\otimes \overline{W_R}$ is a conformal subalgebra of the full operator algebra of $\mathcal{H}$,
\begin{align}
    W_L\otimes\overline{W_R}\subset V_L\otimes \overline{V_R}\subset\mathcal{H},
\end{align}
and one can compute the corresponding symmetry category $\Ver (\mathcal{H}/W_L\otimes \overline{W}_R)$ which preserves $W_L\otimes\overline{W_R}$. There are clearly more topological lines which preserve the chiral algebras $W_L\otimes\overline{W_R}$ than there are that preserve $V_L\otimes\overline{V_R}$,
\begin{align}
    \Ver (\mathcal{H}/W_L\otimes\overline{W_R})\supset \Ver (\mathcal{H}/V_L\otimes\overline{V_R}),
\end{align} 
since $W_L\otimes\overline{W_R}\subset V_L\otimes \overline{V_R}$. However, if one has already done the hard work of computing $\Ver (V_L/W_L)$ and $\Ver (V_R/W_R)$ using the ideas of the previous section,  one prefers not to compute $\Ver (\mathcal{H}/W_L\otimes\overline{W_R})$ from scratch, but rather to reconstruct this category from $\Ver (\mathcal{H}/V_L\otimes\overline{V_R})$ and the two chiral computations of $\Ver (V_L/W_L)$ and $\Ver (V_R/W_R)$. What is the relationship between the four categories $\Ver (\mathcal{H}/W_L\otimes\overline{W_R})$, $\Ver (\mathcal{H}/V_L\otimes\overline{V_R})$, $\Ver (V_L/W_L)$, and $\Ver (V_R/W_R)$?

Following Kapustin and Saulina \cite{Kapustin:2010if}, by thinking of the conformal field theory $\mathcal{H}$ as being built on top of the chiral algebras $V_L$ and $V_R$, we may represent it via the 3D background depicted in the top left of Figure \ref{fig:combiningchiralsymmetries}. Here, $\mathcal{I}$ is a topological interface between the topological quantum field theories $\mathcal{B}(V_L)$ and $\mathcal{B}(V_R)$ supporting $V_L$ and $V_R$ on their boundaries. Its role is to make it possible to lay out the chiral algebras along with their bulks on an interval geometry, and hence to define a 2D CFT by dimensional reduction; operationally, $\mathcal{I}$ dictates how irreducible representations of $V_L$ and $V_R$ are glued together to produce a consistent modular invariant Hilbert space $\mathcal{H}$ and, as such, it encodes information like the pairing matrix $M_{ab}$ in Equation \eqref{eqn:fullCFTHilbertspace}. The category $\Ver (\mathcal{H}/V_L\otimes \overline{V_R})$ of topological line operators of $\mathcal{H}$ which commute with the $V_L\otimes\overline{V_R}$ subalgebra appear in this picture on the worldvolume of the topological interface $\mathcal{I}$. 

On the other hand, we may also think of $\mathcal{H}$ as a conformal field theory built on the smaller chiral algebras $W_L\subset V_L$ and $W_R\subset V_R$. Thus, after replacing $V_L$ with $W_L$ and $V_R$ with $W_R$, etc., one recovers an analogous presentation of the same 2D CFT $\mathcal{H}$, as depicted in the bottom of Figure \ref{fig:combiningchiralsymmetries}. This presentation of the theory exposes the larger symmetry category $\Ver (\mathcal{H}/W_L\otimes \overline{W_R})$ on the worldvolume of $\mathcal{I}'$. 

One can pass between these two presentations of the conformal field theory $\mathcal{H}$ by using the ``SymTFT'' picture of conformal embeddings depicted in Figure \ref{fig:SymTFTconformalembedding}. Indeed, if we apply this to the conformal embeddings $W_L\subset V_L$ and $W_R\subset V_R$ separately, we get a picture as in the top right of Figure \ref{fig:combiningchiralsymmetries}. Here, $A_L$ is the condensable algebra in $\mathcal{B}(W_L)$ defined by $V_L$, and similarly $A_R$ is the condensable algebra in $\mathcal{B}(W_R)$ defined by $V_R$.  If we then fuse the three topological interfaces $\mathcal{I}_{A_L}$, $\mathcal{I}$, and $\overline{\mathcal{I}}_{A_R}$ together, we recover the topological interface $\mathcal{I}'$, i.e.\ 
\begin{align}\label{eqn:relationinterfaces}
    \mathcal{I}'=\mathcal{I}_{A_L}\otimes \mathcal{I}\otimes \overline{\mathcal{I}}_{A_R}.
\end{align}
Recall that the two symmetry categories $\Ver (V_L/W_L)$ and $\Ver (V_R/W_R)$ appear in this picture as the category of line operators supported on the topological interfaces $\mathcal{I}_{A_L}$ and $\mathcal{I}_{A_R}$, respectively. 

Thus, the question we are asking is whether we can relate the category of topological line operators supported on the interface $\mathcal{I}'$ to the category of topological line operators supported on the interfaces $\mathcal{I}_{A_L}$, $\mathcal{I}$, and $\overline{\mathcal{I}}_{A_R}$ which recover $\mathcal{I}'$ when they are fused together.

\begin{figure}
    \begin{center}
\input{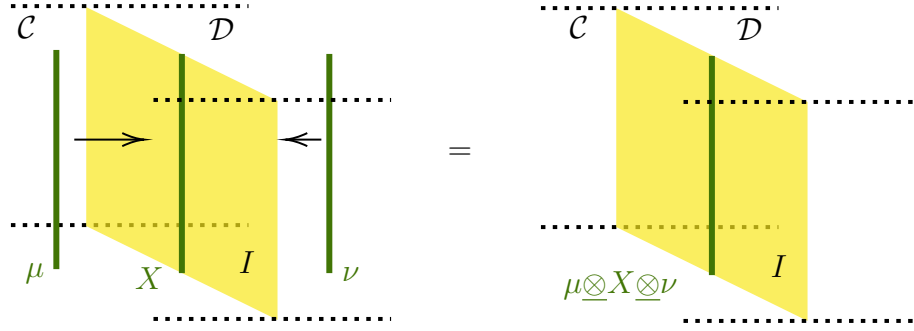}
\caption{The category $\mathbf{L}(\mathcal{I})$ of topological lines supported on the interface $\mathcal{I}$ possesses the structure of a $(\mathcal{C},\mathcal{D})$-bimodule category.}\label{fig:CDbimodule}
    \end{center}
\end{figure}

\subsection{Gluing chiral symmetries using the relative Deligne product}\label{subsec:gluingrelative}

The answer is yes if we use the techniques developed in \cite{Huston:2022utd}. To state the relation, we first note the following.

\begin{claim}{}{}
The category $\mathbf{L}(\mathcal{I})$ of topological line operators supported on an interface $\mathcal{I}$ between two topological orders $\mathcal{C}$ and $\mathcal{D}$ is naturally a $(\mathcal{C},\mathcal{D})$-bimodule category, via the fusion of lines in the bulk onto lines supported on the interface, as in Figure \ref{fig:CDbimodule}.
\end{claim}
We use $\underline{\otimes}$ to denote both tensor products, i.e.\ both maps $\mathcal{C}\times\mathbf{L}(\mathcal{I})\to\mathbf{L}(\mathcal{I})$ and $\mathbf{L}(\mathcal{I})\times \mathcal{D}\to\mathbf{L}(\mathcal{I})$.  
\begin{example}{}{}The category $\Ver (V_L/W_L)$ of topological line operators supported on the interface $\mathcal{I}_{A_L}$ is naturally a $(\mathcal{B}(W_L),\mathcal{B}(V_L))$-bimodule category:  $\mathcal{B}(V_L)$ acts on $\Ver (V_L/W_L)$ simply because it is a subcategory, and $\mathcal{B}(W_L)$ acts through the induction functor from Equation \eqref{eqn:induction}, 
\begin{align}
\begin{split}
    \mu\underline{\otimes}X\underline{\otimes}\nu = I(\mu)\otimes X\otimes \nu, \ \ \ \ \mu\in\mathcal{B}(W_L), \ X\in\Ver (V_L/W_L), \ \nu\in\mathcal{B}(V_L).
\end{split}
\end{align}
Similar comments apply to the categories $\Ver (V_R/W_R)$ and $\Ver (\mathcal{H}/V_L\otimes\overline{V_R})$ of line operators supported on the interfaces $\overline{\mathcal{I}}_{A_R}$ and $\mathcal{I}$, respectively. 
\end{example}

The correct way to combine the three wall categories is to use the \emph{relative} Deligne product, see \cite[Section 3]{etingof2010fusion} for a thorough exposition. The intuition is that the category $\mathbf{L}(\mathcal{I}\otimes\mathcal{J})$ of line operators on a composite interface $\mathcal{I}\otimes\mathcal{J}$ is \emph{not} given by the standard Deligne product $\mathbf{L}(\mathcal{I})\boxtimes \mathbf{L}(\mathcal{J})$. Instead, one must impose certain relations, which are taken care of by the \emph{relative} Deligne product. For example, the standard Deligne product always treats the lines $(X,\mathds{1})$ and $(\mathds{1},Y)$ as distinct, but if there is a bulk line $a \in \mathcal{D}$ such that $X$ is the result of pushing $a$ onto $\mathcal{I}$, and $Y$ is the result of pushing $a$ onto $\mathcal{J}$, then $(X,\mathds{1})$ and $(\mathds{1},Y)$ should be identified in the composite category $\mathbf{L}(\mathcal{I}\otimes\mathcal{J})$. We formalize this below.

We will provide concrete constructions of the relative Deligne product shortly, but for now, we appeal to its more abstract definition through a universal property. Let $\mathcal{M}$ be a right $\mathcal{D}$-module category, $\mathcal{N}$ be a left $\mathcal{D}$-module category, and $\mathcal{A}$ be any Abelian category. We say that a functor $B:\mathcal{M}\boxtimes \mathcal{N}\to \mathcal{A}$ is $\mathcal{D}$-balanced if there is a family of natural isomorphisms $b_{X,D,Y}: B(X\underline{\otimes}D,Y)\xrightarrow{\sim} B(X,D\underline{\otimes}Y)$ which satisfy the following pentagon equation, 
\[
\begin{tikzcd}[column sep=small,row sep=small]
B(X\otimes(D_1\otimes D_2),Y)
  \arrow[rr,"m_{X,D_1,D_2}"]%
  \arrow[d,"b_{X,D_1\otimes D_2,Y}"']
  & &
  B((X\otimes D_1)\otimes D_2,Y)
    \arrow[d,"b_{X\otimes D_1,D_2,Y}"] \\[0.5em]
B(X,(D_1\otimes D_2)\otimes Y)
  & &
  B(X\otimes D_1,D_2\otimes Y)
    \arrow[dl,"b_{X,D_1,D_2\otimes Y}"] \\[0.5em]
& B(X,D_1\otimes(D_2\otimes Y))
    \arrow[ul,"n^{-1}_{D_1,D_2,Y}"']
\end{tikzcd}
\]
where $m$ is the right $\mathcal{D}$-module category associator for $\mathcal{M}$ and $n$ is the left $\mathcal{D}$-module category associator for $\mathcal{N}$. The relative Deligne product is then any Abelian category $\mathcal{M}\boxtimes_{\mathcal{D}}\mathcal{N}$ along with a $\mathcal{D}$-balanced functor $B:\mathcal{M}\boxtimes \mathcal{N}\to \mathcal{M}\boxtimes_{\mathcal{D}}\mathcal{N}$ which satisfies the universal property that, if $F:\mathcal{M}\boxtimes\mathcal{N}\to\mathcal{A}$ is a $\mathcal{D}$-balanced functor, then there is an additive functor $F':\mathcal{M}\boxtimes_{\mathcal{D}}\mathcal{N}\to\mathcal{A}$ satisfying $F=F'\circ B$.

\begin{claim}{(Equation (15) of \cite{Huston:2022utd})}{}
   If $\mathcal{I}$ is an interface between $\mathcal{C}$ and $\mathcal{D}$, and $\mathcal{J}$ is an interface between $\mathcal{D}$ and $\mathcal{E}$, then the category of topological line operators supported on the composite interface $\mathcal{I}\otimes \mathcal{J}$ is furnished by the relative Deligne product,
\begin{align}
    \mathbf{L}(\mathcal{I}\otimes \mathcal{J})\cong \mathbf{L}(\mathcal{I})\boxtimes_{\mathcal{D}}\mathbf{L}(\mathcal{J}).
\end{align}
\end{claim}
Indeed, the $\mathcal{D}$-balanced functor $B:\mathbf{L}(\mathcal{I})\boxtimes\mathbf{L}(\mathcal{J})\to \mathbf{L}(\mathcal{I})\boxtimes_{\mathcal{D}}\mathbf{L}(\mathcal{J})$ which participates in the universal property  physically encodes the result of fusing a line $X$ supported on $\mathcal{I}$ with a line $Y$ supported on $\mathcal{J}$, as in Figure \ref{fig:Dbalfunctor}. The accompanying family $b_{X,D,Y}$ of natural isomorphisms can be unpacked into a collection of numbers $[F_{XDY}^L]_{(Z_1u_1v_1)(Z_2u_2v_2)}$ which behave like F-symbols, see Figure \ref{fig:bnaturalisos}. Indeed, when $\mathcal{C}=\mathcal{D}=\mathcal{E}$, and $\mathcal{I}$ and $\mathcal{J}$ are both taken to be the trivial surface,  then these numbers coincide with the standard F-symbols of $\mathcal{C}$.

\begin{figure}
    
\begin{center}
\input{Figures/Dbalfunctor}
\caption{The physical interpretation of the functor $B:\mathbf{L}(\mathcal{I})\boxtimes\mathbf{L}(\mathcal{J})\to\mathbf{L}(\mathcal{I})\boxtimes_{\mathcal{D}}\mathbf{L}(\mathcal{J})$.}\label{fig:Dbalfunctor}
\end{center}
\end{figure}

\begin{figure}
    \begin{center}
\input{Figures/bnaturalisos}
\caption{The physical interpretation of the natural isomorphisms $b_{X,D,Y}:B(X\underline{\otimes}D,Y)\xrightarrow{\sim}B(X,D\underline{\otimes}Y)$.}\label{fig:bnaturalisos}
    \end{center}
\end{figure}

If we apply these constructions to the problem at hand, we arrive at the main result of this section. Recall that, given a fusion category $\mathcal{C}$, the dual category $\mathcal{C}^{\vee}$ also admits the structure of a fusion category. (See e.g.\ \cite[Remark 2.1.6]{etingof2015tensor}. In our context, where all categories are rigid, $\mathcal{C}^\vee$ is equivalent to $\mathcal{C}^{\mathrm{op}}$ defined in \cite[Definition 2.1.5]{etingof2015tensor}.) Similarly, given a left $\mathcal{C}$-module category $\mathcal{M}$, the dual category $\mathcal{M}^{\vee}$ admits the structure of a right $\CC$-module category (see e.g.\ \cite[Remark 7.1.5]{etingof2015tensor}). 
We then have the following.

\begin{claim}{}{}
Let $\mathcal{H}$ be a conformal field theory with left- and right-moving chiral algebras $V_L$ and $\overline{V_R}$, respectively. The symmetries preserved by conformal subalgebras $W_L\subset V_L$ and $\overline{W_R}\subset \overline{V_R}$ are given by the formula
\begin{align}
    \Ver (\mathcal{H}/W_L\otimes\overline{W_R})\cong \Ver (V_L/W_L)\boxtimes_{\mathcal{B}(V_L)}\Ver (\mathcal{H}/V_L\otimes \overline{V_R})\boxtimes_{\mathcal{B}(V_R)}\Ver (V_R/W_R)^{\vee},
\end{align}
which extends the relationship between the four interfaces in Equation \eqref{eqn:relationinterfaces} to a relation between their corresponding categories of topological line operators.
This equation describes how to glue symmetries of chiral algebra boundary conditions into  symmetries of full conformal field theories. Cf.\ Equation \eqref{eqn:gluingformula} from the introduction for an equivalent formulation.
\end{claim}

For example, if $V_L=V_R=V$ and $W_L=W_R=W$ and $\mathcal{H}$ is the canonical diagonal rational conformal field theory built on $V$, then $\Ver (\mathcal{H}/V\otimes\overline{V})\cong \mathcal{B}(V)$ is the regular $(\mathcal{B}(V),\mathcal{B}(V))$-bimodule category, which behaves like the identity with respect to the relative Deligne product, and hence we find that 
\begin{align}\label{eqn:descrip1}
    \Ver (\mathcal{H}/W\otimes\overline{W})=\Ver (V/W)\boxtimes_{\mathcal{B}(V)}\Ver (V/W)^{\vee}.
\end{align}
Thus, the category of topological line operators of the canonical diagonal rational CFT preserved by a subalgebra $W\otimes \overline{W}$ roughly factorizes into a (relative Deligne) product of a category on the left with a category on the right.

 There is another well-known formula for $\Ver(\mathcal{H}/W\otimes \overline{W})$ to which we can compare Equation \eqref{eqn:descrip1}. Call $A$ the condensable algebra in $\mathcal{B}(W)$ which mediates the conformal extension from $W$ to $V$. If we forget about the fact that it is a condensable algebra (describing 1-form gauging in 3D) and treat it instead as a gaugeable algebra (describing 0-form gauging in 2D), then we can start with the diagonal CFT built on $W$, which has $\mathcal{B}(W)$ as its Verlinde lines, and gauge $A$ to obtain the diagonal CFT built on $V$. From this description, \cite{Fuchs:2002cm,Bhardwaj:2017xup} predicts that
\begin{align}\label{eqn:descrip2}
    \Ver(\mathcal{H}/W\otimes \overline{W})\cong {_{A}}\mathcal{B}(W)_A,
\end{align}
where ${_{A}}\mathcal{B}(W)_A$ is the category of $A$-$A$-bimodules in $\mathcal{B}(W)$. The equivalence of Equation \eqref{eqn:descrip1} and Equation \eqref{eqn:descrip2} follows from the fact that
\begin{align}
    \mathcal{B}(W)_A\boxtimes_{\mathcal{B}(W)_A^{\mathrm{loc}}}\mathcal{B}(W)_A^{\vee}\cong {_{A}}\mathcal{B}(W)_A.
\end{align}

The description in Equation \eqref{eqn:descrip2} has some useful corollaries. For example, thinking of the canonical diagonal CFT built on $V$ as a non-diagonal theory built on the conformal subalgebra $W$, we have that the pairing matrix $M_{\mu\nu}$ from Equation \eqref{eqn:fullCFTHilbertspace} is given by 
\begin{align}
    M_{\mu\nu} = \sum_{a\in\mathcal{B}(V)} B_{a\mu}B_{a^\ast\nu}, \ \ \ \ \mu,\nu \in\mathcal{B}(W),
\end{align}
where the $B_{a\mu}$ describe the decomposition of $V$-modules into $W$-modules (cf.\ Equation \eqref{eqn:Vmoddecomp}).
From \cite[Claim 3]{Ostrik:2001xnt}, the complexified Grothendieck ring of $\Ver(\mathcal{H}/W\otimes \overline{W})$ decomposes into matrix algebras as 
\begin{align}\label{eqn:matrixalgdecomp}
    \mathcal{K}_0(\Ver(\mathcal{H}/W\otimes\overline{W}))\otimes \mathbb{C}\cong \bigoplus_{\mu,\nu \in \mathcal{B}(W)} \mathrm{Mat}_{\mathbb{C}}(M_{\mu\nu}).
\end{align}
In particular, $\Ver(\mathcal{H}/W\otimes\overline{W})$ has commutative fusion rules if and only if $M_{\mu\nu}\leq 1$ for all $\mu,\nu\in\mathcal{B}(W)$. 

There is yet another description of $\Ver(\mathcal{H}/W\otimes\overline{W})$ using the relative Deligne product, namely $\Ver(\mathcal{H}/W\otimes\overline{W})\cong \Ver(V/W)^\vee \boxtimes_{\mathcal{B}(W)} \Ver(V/W)$. We leave a more detailed exploration to our companion paper \cite{grmath}.

\begin{example}{}{cleftamalg}
    Consider a chiral algebra $V$ supported on the boundary of an Abelian Chern-Simons theory, i.e.\ $\Rep(V)\cong \mathrm{Vec}_D^{\sigma,\omega}$. Let $\mathcal{H}$ be the canonical diagonal rational conformal field theory built on top of $V$. Then the fusion category of topological line operators of $\mathcal{H}$ preserved by $V\otimes \overline{V}$ is 
    \begin{align}
        \Ver (\mathcal{H}/V\otimes\overline{V})\cong \mathcal{B}(V)\cong \mathrm{Vec}_D^\omega.
    \end{align}
    Now, suppose $G$ is a finite, \underline{cleft} group of automorphisms of $V$, so that by Example \ref{ex:cleftpt1},
    \begin{align}
        \Ver (V/V^G)\cong \mathrm{Vec}_\Gamma^{\tilde\omega}
    \end{align}
    for some extension $\Gamma$ of $G$ by $D$. We can ask what is the category $\Ver(\mathcal{H}/V^G\otimes\overline{V^G})$ of topological line operators of $\mathcal{H}$ which preserve the factorizing subalgebra $V^G\otimes \overline{V^G}$. 
    
   \hspace{.3in} In this case, the relative Deligne product is related to the amalgamated product of groups. Define $\Gamma^{(2)}=\Gamma\times \Gamma/D$, where $D$ is embedded in $\Gamma\times \Gamma$ diagonally via $d\mapsto (d,d)$. Then
    \begin{align}
    \begin{split}
        \Ver (\mathcal{H}/V^G\otimes\overline{V^G})&\cong \Ver (V/V^G)\boxtimes_{\mathcal{B}(V)}\Ver (V/V^G)^{\vee} \\
        &\cong \mathrm{Vec}_\Gamma^{\tilde\omega}\boxtimes_{\mathrm{Vec}_D^\omega}\mathrm{Vec}_{\Gamma}^{\tilde{\omega}^{\vee}} \\
        &\cong \mathrm{Vec}_{\Gamma^{(2)}}^{\nu},
    \end{split}
    \end{align}
    where here, $\tilde{\omega}^{\vee}(\gamma_1,\gamma_2,\gamma_3) = \omega(\gamma_1,\gamma_2,\gamma_3)^{-1}\in H^3(\Gamma,U(1))$ and $\nu\in H^3(\Gamma^{(2)},U(1))$ is a 3-cocycle which obeys $[q^\ast\nu]=[\tilde{\omega}\times \tilde{\omega}^\vee]\in H^3(\Gamma\times \Gamma,U(1))$, where $q:\Gamma\times \Gamma\to \Gamma^{(2)}$ is the natural quotient map and $q^\ast:H^3(\Gamma^{(2)},U(1))\to H^3(\Gamma\times\Gamma,U(1))$ denotes pullback. 
\end{example}

\subsection{Concrete constructions of the relative Deligne product}\label{subsec:relativedeligne}

\begin{figure}
    \begin{center}
\input{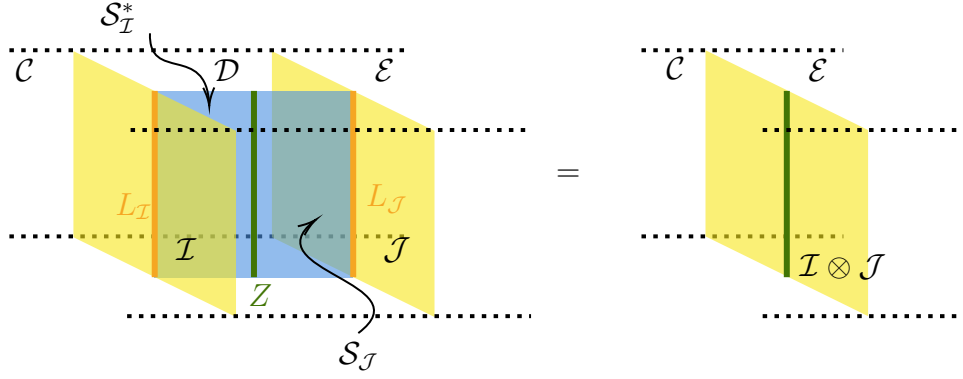}
\caption{The physical interpretation of the equivalence $\mathrm{Fun}_{\mathcal{D}}(\mathbf{L}(\mathcal{I})^{\vee},\mathbf{L}(\mathcal{J}))\cong \mathbf{L}(\mathcal{I}\otimes \mathcal{J}).$ The two blue surfaces are determined abstractly by the structure of $\mathbf{L}(\mathcal{I})^{\vee}$ and $\mathbf{L}(\mathcal{J})$ as $\mathcal{D}$-module categories.}\label{fig:IJsandwich}
    \end{center}
\end{figure}

The description of the relative Deligne product through its universal property is not very useful for computations. Therefore, we now provide some more concrete descriptions.

\paragraph{As a category of $\mathcal{D}$-module functors} One useful fact is that the relative Deligne product is equivalent to a certain category of $\mathcal{D}$-module functors,
\begin{align}\label{eqn:functorrelativedeligne}
    \mathbf{L}(\mathcal{I}\otimes\mathcal{J})\cong \mathrm{Fun}_{\mathcal{D}}(\mathbf{L}(\mathcal{I})^{\vee},\mathbf{L}(\mathcal{J})).
\end{align}
In this realization, the $\mathcal{D}$-balanced functor  which arises in the universal property is given mathematically by 
\begin{align}
\begin{split}
    B:\mathbf{L}(\mathcal{I})\boxtimes\mathbf{L}(\mathcal{J})&\to \mathrm{Fun}_{\mathcal{D}}(\mathbf{L}(\mathcal{I})^{\vee},\mathbf{L}(\mathcal{J})) \\
    (X,Y)&\mapsto \mathrm{Hom}_{\mathbf{L}(\mathcal{I})}(\cdot,X)\otimes Y.
\end{split}
\end{align}

To interpret this physically, we note the following facts. Recall that $\mathcal{D}$-module categories describe topological surfaces in the TQFT defined by $\mathcal{D}$, and so both $\mathbf{L}(\mathcal{I})^{\vee}$ and $\mathbf{L}(\mathcal{J})$ define surfaces in $\mathcal{D}$. In fact, the pinching trick from the previous section (Figure \ref{fig:pinch}) tells us that the surfaces are simply $\mathcal{S}_{\mathcal{I}}\equiv \mathcal{I}^\ast\otimes \mathcal{I}$ and $\mathcal{S}_{\mathcal{J}}\equiv\mathcal{J}^\ast\otimes \mathcal{J}$, respectively. As before, the objects of $\mathbf{L}(\mathcal{I})^{\vee}$ and $\mathbf{L}(\mathcal{J})$ can be thought of as topological lines which bound $\mathcal{S}_{\mathcal{I}}$ and $\mathcal{S}_{\mathcal{J}}$. Both surfaces admit topological junctions $L_{\mathcal{I}}$ and $L_{\mathcal{J}}$ with the interfaces $\mathcal{I}$ and $\mathcal{J}$, respectively, and moreover the category $\mathrm{Fun}_{\mathcal{D}}(\mathbf{L}(\mathcal{I})^{\vee},\mathbf{L}(\mathcal{J}))$ physically encodes the category of topological line interfaces between $\mathcal{S}_{\mathcal{I}}$ and $\mathcal{S}_{\mathcal{J}}$.

Using these ingredients,  we can produce a line on the composite surface $\mathcal{I}\otimes \mathcal{J}$ from a line interface $Z\in \mathrm{Fun}_{\mathcal{D}}(\mathbf{L}(\mathcal{I})^{\vee},\mathbf{L}(\mathcal{J}))$ between the surfaces $\mathcal{S}_{\mathcal{I}}$ and $\mathcal{S}_{\mathcal{J}}$ by sandwiching it with $\mathcal{I}$ and $\mathcal{J}$, as in Figure \ref{fig:IJsandwich}. This furnishes the equivalence in Equation \eqref{eqn:functorrelativedeligne}. In this picture, the image of bounding lines $X\in \mathbf{L}(\mathcal{I})^{\vee}$ and $Y\in\mathbf{L}(\mathcal{J})$ under the $\mathcal{D}$-balanced functor $B:\mathbf{L}(\mathcal{I})\boxtimes \mathbf{L}(\mathcal{J})\to \mathbf{L}(\mathcal{I}\otimes \mathcal{J})$ is simply the interface obtained by fusion, as in Figure \ref{fig:B in Fun description}. The natural isomorphisms $b_{X,D,Y}$ (or equivalently, the F-symbols $[F^L_{XDY}]_{(Z_1u_1v_1)(Z_2u_2v_2)}$) arise in a picture which we do not draw, but which is obtained by generalizing Figure \ref{fig:bnaturalisos}.

\begin{figure}
    \begin{center}
\input{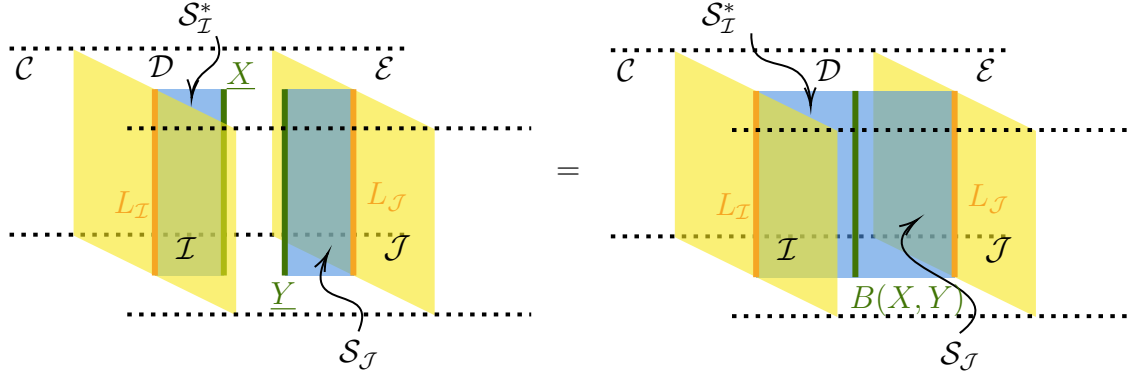}
\caption{The physical interpretation of the $\mathcal{D}$-balanced functor $B:\mathbf{L}(\mathcal{I})\boxtimes \mathbf{L}(\mathcal{J})\to \mathrm{Fun}_{\mathcal{D}}(\mathbf{L}(\mathcal{I})^{\vee},\mathbf{L}(\mathcal{J}))$.}\label{fig:B in Fun description}
    \end{center}
\end{figure}

\paragraph{As a category of modules for an algebra} Another useful perspective on the relative Deligne product is obtained by folding. By folding the left of Figure \ref{fig:IJsandwich} along the line interface $Z$, we see that $\mathrm{Fun}_{\mathcal{D}}(\mathbf{L}(\mathcal{I})^{\vee},\mathbf{L}(\mathcal{J}))$ is equivalent to the category of topological lines on the canonical gapped boundary of $Z(\mathcal{D})\cong \mathcal{D}\boxtimes\overline{\mathcal{D}}$ on which the surface $\mathcal{S}_{\mathcal{I}}\otimes \mathcal{S}_{\mathcal{J}}$ can terminate. Gapped boundary conditions of $Z(\mathcal{D})$ are in one-to-one correspondence with Lagrangian algebras of $Z(\mathcal{D})$, and the canonical gapped boundary of $Z(\mathcal{D})$ corresponds to the algebra 
   \begin{align}
    K_{\mathcal{D}}=\bigoplus_{D\in\mathrm{Irr}(\mathcal{D})} D\boxtimes D^\ast.
\end{align}
Now, in the case that $\mathcal{S}_{\mathcal{I}}$ and $\mathcal{S}_{\mathcal{J}}$ are the trivial surfaces of $\overline{\mathcal{D}}$ and $\mathcal{D}$, respectively, the folded problem reduces to finding the topological lines supported on the canonical gapped boundary $K_{\mathcal{D}}$; this category is well-known to be described by $Z(\mathcal{D})_{K_{\mathcal{D}}}$, i.e.\ the category of $K_{\mathcal{D}}$-modules in $Z(\mathcal{D})$. The same logic which underpins the derivation of this fact shows that, in general, one should compute the category of $K_{\mathcal{D}}$-modules inside of $\mathbf{L}(\mathcal{I})\boxtimes\mathbf{L}(\mathcal{J})$, so that 
\begin{align}
   \mathrm{Fun}_{\mathcal{D}}(\mathbf{L}(\mathcal{I})^{\vee},\mathbf{L}(\mathcal{J}))\cong (\mathbf{L}(\mathcal{I})\boxtimes\mathbf{L}(\mathcal{J}))_{K_{\mathcal{D}}}.
\end{align}
Here, we are using the fact that $\mathbf{L}(\mathcal{I})\boxtimes\mathbf{L}(\mathcal{J})$ is in particular a $\mathcal{D}\boxtimes \overline{\mathcal{D}}$-module category, and so it makes sense to compute the modules of an algebra of $\mathcal{D}\boxtimes\overline{\mathcal{D}}$ inside of it.

Just as in Equation \eqref{eqn:induction}, we have an induction functor (which we now call $B$ instead of $I$, because it coincides with the $\mathcal{D}$-balanced functor arising in the universal property of the relative Deligne product),
\begin{align}
\begin{split}
    B:\mathbf{L}(\mathcal{I})\boxtimes\mathbf{L}(\mathcal{J})&\rightarrow (\mathbf{L}(\mathcal{I})\boxtimes\mathbf{L}(\mathcal{J}))_{K_{\mathcal{D}}} \\
    X &\mapsto (X\otimes K_{\mathcal{D}},\mathrm{id}_X\otimes m)
\end{split}
\end{align}
where $m$ is the multiplication on the algebra $K_{\mathcal{D}}$. The physical interpretation of this functor is obtained by simply folding Figure \ref{fig:B in Fun description}, as in Figure \ref{fig:boundaryinterpretation}. In particular, it can be thought of as a kind of bulk-to-boundary functor obtained by pushing a bounding line of the surface $\mathcal{S}_{\mathcal{I}}\otimes \mathcal{S}_{\mathcal{J}}$ onto the canonical boundary $K_{\mathcal{D}}$. This picture in particular makes it clear that $B$ should enjoy several useful properties: it is a tensor functor, and it enjoys Frobenius reciprocity with restriction. 

The content of this subsection can be summarized by the following.
\begin{claim}{}{}
    The category of line operators on a composite surface is given by any of the following expressions,
\begin{align}
    \mathbf{L}(\mathcal{I}\otimes\mathcal{J})\cong \mathbf{L}(\mathcal{I})\boxtimes_{\mathcal{D}}\mathbf{L}(\mathcal{J})\cong \mathrm{Fun}_{\mathcal{D}}(\mathbf{L}(\mathcal{I})^{\vee},\mathbf{L}(\mathcal{J}))\cong (\mathbf{L}(\mathcal{I})\boxtimes\mathbf{L}(\mathcal{J}))_{K_{\mathcal{D}}}.
\end{align}
\end{claim}

\begin{figure}
    \begin{center}
\input{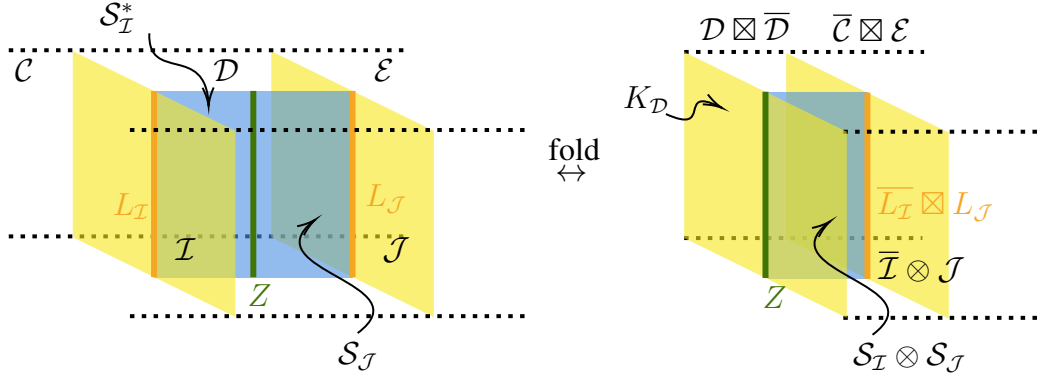}
\caption{The physical interpretation of the equivalence $\mathrm{Fun}_{\mathcal{D}}(\mathbf{L}(\mathcal{I})^{\vee},\mathbf{L}(\mathcal{J}))\xrightarrow{\sim}(\mathbf{L}(\mathcal{I})\boxtimes\mathbf{L}(\mathcal{J}))_{K_{\mathcal{D}}}$.}\label{fig:boundaryinterpretation}
    \end{center}
\end{figure}

Before moving on, we note one more useful formula. Suppose $\mathcal{H}$ is a rational conformal field theory with maximal chiral algebra $V\otimes \overline{V}$, such that 
\begin{align}
    \mathcal{H}\cong \bigoplus_{a\in\mathcal{B}(V)} V_a\otimes \overline{V}_{\phi(a)^\ast},
\end{align}
where $\phi$ is a ribbon auto-equivalence of $\mathcal{B}(V)$. In this case, if $W$ is a rational conformal subalgebra, then 
\begin{align}
    \Ver (\mathcal{H}/W\otimes \overline{W})\cong (\Ver (V/W)\boxtimes \Ver (V/W)^{\vee})_{K_{\mathcal{B}(V)}^\phi}
\end{align}
where now $K_{\mathcal{B}(V)}^\phi$ is a twisted Lagrangian algebra of $\mathcal{B}(V)\boxtimes\overline{\mathcal{B}(V)}$, 
\begin{align}
    K_{\mathcal{B}(V)}^\phi \cong \bigoplus_{a\in \mathcal{B}(V)} a\otimes \phi(a)^\ast.
\end{align}
As an abstract fusion category, $\Ver (\mathcal{H}/W\otimes \overline{W})$ is independent of the precise choice of twist $\phi$ which is chosen. However, for example, the induction functor $\Ver (V/W)\boxtimes \Ver (V/W)^{\vee}\to \Ver (\mathcal{H}/W\otimes\overline{W})$ will depend on $\phi$. 

\subsection{Boundary conditions from folding}\label{subsec:boundaries}

It turns out there is a close relationship between the symmetry structure of a chiral algebra $V$ and boundary conditions of full CFTs built on top of $V$. This subsection is dedicated to describing this connection.

In the previous subsections, we allowed for ``heterotic'' rational CFTs, i.e.\ CFTs for which the left- and right-moving chiral algebras are different. To make the present discussion simpler, let us specialize to the case that the left- and right-moving chiral algebra are the same. To this end, recall from Figure \ref{fig:KS} the definition of the theory $\mathrm{CFT}_{V,\mathcal{S}}$. It is a rational conformal field theory built on the chiral algebra $V$, with left- and right-movers glued using the topological surface $\mathcal{S}$. Let us use the notation 
\begin{align}
    \mathrm{Bdy}^\dagger(\mathrm{CFT}_{V,\mathcal{S}}), \ \ \ \mathrm{Bdy}^\dagger(\mathrm{CFT}_{V,\mathcal{S}}/W)
\end{align}
to denote, respectively, the full category of unitary boundary conditions of the theory $\mathrm{CFT}_{V,\mathcal{S}}$, and the subcategory of unitary boundary conditions of $\mathrm{CFT}_{V,\mathcal{S}}$ which preserve the conformal subalgebra $W\subset V$. 

We say by definition that a boundary condition preserves the subalgebra $W$ if the corresponding boundary state $|B\rangle$ satisfies 
\begin{align}\label{eqn:boundarypreservingW}
   w_n|B\rangle = (-1)^{h_w}\overline{w}_{-n}|B\rangle,
\end{align}
where $w_n$ are the modes of an arbitrary operator $w(z)$ in the chiral algebra $W$, and $h_w$ is the conformal dimension of $w(z)$. 
Note that sometimes one sees a relaxed version of Equation \eqref{eqn:boundarypreservingW} where a boundary condition preserves $W$ only up to a ``gluing automorphism'', 
\begin{align}\label{eqn:gluingaut}
    w_n|B\rangle =(-1)^{h_w} \Omega(\overline{w}_{-n})|B\rangle, \ \ \ \ \Omega\in\mathrm{Aut}(W).
\end{align}
However, in our nomenclature, we would say that such a boundary condition breaks $W$ down to the subalgebra $W^\Omega$ of $\Omega$-invariant operators.

We also remark that, sometimes, the same CFT can be recovered using multiple Kapustin-Saulina surfaces. That is, sometimes, $\mathrm{CFT}_{V,\mathcal{S}}\cong \mathrm{CFT}_{V,\mathcal{S}'}$ for two genuinely different surface operators $\mathcal{S}$ and $\mathcal{S}'$. In this situation, the definition of what it means for a boundary condition to preserve the subalgebra $W$ may depend on which surface operator is used to construct the theory.

\begin{example}{}{}
Let $V=V_{2m}$ be the $c=1$ chiral boson algebra from Example \ref{ex:latticeVOA}. The theory $\mathrm{CFT}_{V,\mathds{1}}$ built using the identity surface $\mathcal{S}=\mathds{1}$ in $U(1)_{2m}$ Chern-Simons theory is the compact boson of radius $R^2=2m$, or radius $R^2=2/m$ by T-duality. The boundary conditions which preserve the $W=\widehat{\mathfrak{u}}(1)$ Kac-Moody subalgebra in the compact boson of radius $R^2=2m$ (resp.\ $R^2=2/m$) are the Dirichlet (resp.\ Neumann) boundary conditions. The Neumann (resp.\ Dirichlet) boundaries preserve the smaller subalgebra $W=\widehat{\mathfrak{u}}(1)^+$ of $\mathbb{Z}_2$ charge-conjugation invariant operators. 

\hspace{.3in} On the other hand, one can construct a theory $\mathrm{CFT}_{V,C}$ using the $\mathbb{Z}_2$ surface $C$ in $U(1)_{2m}$ Chern-Simons theory which sends the Wilson line of charge $\lambda$ to the Wilson line of charge $-\lambda$. This also engineers the compact boson theory of radius $R^2=2m$, i.e.\ $\mathrm{CFT}_{V,\mathds{1}}\cong \mathrm{CFT}_{V,C}$. However, in this description of the theory, at radius $R^2=2m$ (resp.\ $R^2=2/m$), it is the Neumann (resp.\ Dirichlet) boundaries which preserve the $\widehat{\mathfrak{u}}(1)$ Kac-Moody algebra, and the Dirichlet (resp.\ Neumann) boundaries which preserve the $\widehat{\mathfrak{u}}(1)^+$ subalgebra.

\end{example}

The following proposal is a generalization of observations made in \cite{Kapustin:2010if}. Recall that, by definition, $\Ver_{\mathcal{S}}(V/W)$ consists of topological lines on $V$ on which $\mathcal{S}$ can terminate and which commute with the local operators in the subalgebra $W$ (cf.\ Equation \eqref{eqn:VerS(VW)}). 

\begin{claim}{}{claim:VerS(VW)boundaries}
   The folding functor defined by Figure \ref{fig:lines2boundaries},
   \begin{align}
       \partial: \Ver_{\mathcal{S}}(V/W)\to \mathrm{Bdy}^\dagger(\mathrm{CFT}_{V,\mathcal{S}}/W)
   \end{align}
   furnishes an equivalence of categories. See also Figure \ref{fig:squashboundary}.
\end{claim}

The physical argument for Claim \ref{cl:claim:VerS(VW)boundaries} is the folding trick. Namely, one studies a $V$-boundary of $\mathcal{B}(V)$ in the presence of a boundary topological line defect $X \in \Ver_{\mathcal{S}}(V/W)$ on which the bulk surface $\mathcal{S}$ terminates. One imagines this to be a taco --- with $V$ the taco shell and $\mathcal{S}$ the plant-based meat alternative --- and then folds it into the two-dimensional form in which it will ultimately be eaten. Then, away from the folding crease, the theory goes over to $\mathrm{CFT}_{V,\mathcal{S}}$; at the crease, we construct a boundary condition $\partial X$ of $\mathrm{CFT}_{V,\mathcal{S}}$. See Figure \ref{fig:lines2boundaries} (cf.\ also Figure \ref{fig:squashboundary}). 

\begin{figure}
    \begin{center}
        \input{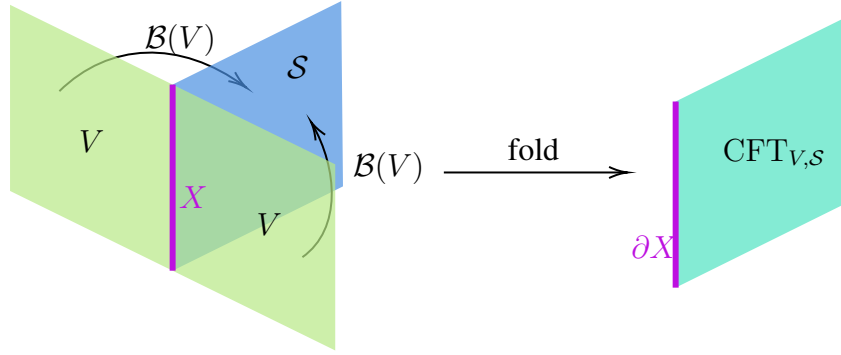}
        \caption{The functor $\partial$ between (non-genuine) lines $X$ on a chiral algebra boundary $V$ and boundary conditions of $\mathrm{CFT}_{V,\mathcal{S}}.$}\label{fig:lines2boundaries}
    \end{center}
\end{figure}

In the special case that one takes $\mathcal{S}$ to be the trivial surface operator and $W=V$, then $\Ver_{\mathcal{S}}(V/W)$ reduces to $\Rep(V)$ and Claim \ref{cl:claim:VerS(VW)boundaries} says that the boundaries of  the canonical diagonal ``charge-conjugation'' modular invariant CFT built on $V$ are in one-to-one correspondence with representations of $V$, as originally showed by Cardy \cite{Cardy:1989ir}. More generally, if $W=V$ but one allows $\mathcal{S}$ to be a non-trivial surface operator (corresponding to gauging an algebra $B\in\Rep(V)$ in the canonical diagonal theory) then the boundaries of $\mathrm{CFT}_{V,\mathcal{S}}$ which preserve $V$ correspond to objects in $\Ver_{\mathcal{S}}(V/W)\cong \Rep(V)_B$, as proved in \cite{Fuchs:2002cm,Kapustin:2010if}.

If we trust Claim \ref{cl:claim:VerS(VW)boundaries} beyond its naive regime of validity by taking $W$ to be the Virasoro algebra, then what we learn is that \emph{every} boundary condition of $\mathrm{CFT}_{V,\mathcal{S}}$ corresponds to some line on $V$ on which $\mathcal{S}$ can terminate.
\begin{claim}{}{claim:bdyextrap}
    Let $\mathcal{S}$ be a topological surface defined by a gaugeable algebra $B$ in $\Rep(V)$. Then there is the following equivalence of categories:
    \begin{align}
        \partial:\mathrm{Sym}^\dagger_{\mathcal{S}}(V)\to  \mathrm{Bdy}^\dagger(\mathrm{CFT}_{V,\mathcal{S}})
    \end{align}
    where $\mathrm{Sym}^\dagger_{\mathcal{S}}(V)\cong \mathrm{Sym}^\dagger(V)_B$ is the full category of topological lines on $V$ on which $\mathcal{S}$ can terminate, and is mathematically described as the category of $B$-modules in $\mathrm{Sym}^\dagger(V)$.
\end{claim}
Let us describe some surprising consequences of this claim in the special case that $\mathcal{S}$ is the trivial surface operator, i.e.\ in the case that we are considering the canonical diagonal rational CFT built on $V$, which we denote $\mathrm{CFT}_V$. We expect most of the claims to generalize. 

First, note that the lines supported on e.g.\ the left-moving chiral algebra embed as a subcategory into the category $\mathrm{Sym}^\dagger(\mathrm{CFT}_V)$ of line operators of the full diagonal theory built on $V$, 
\begin{align}
\begin{split}
    \mathrm{Sym}^\dagger(V) &\subset \mathrm{Sym}^\dagger(\mathrm{CFT}_V)\cong \mathrm{Sym}^\dagger(V)\boxtimes_{\mathcal{B}(V)}\mathrm{Sym}^\dagger(V)^{\vee}, \\
    X &\mapsto B(X,\mathds{1}),
\end{split}
\end{align}
where $B:\mathrm{Sym}^\dagger(V)\boxtimes \mathrm{Sym}^\dagger(V)^{\vee}\to \mathrm{Sym}^\dagger(V)\boxtimes_{\mathcal{B}(V)}\mathrm{Sym}^\dagger(V)^{\vee}$ is the $\mathcal{B}(V)$-balanced functor appearing in the definition of the relative Deligne product, cf.\ Section \ref{subsec:gluingrelative}. 
\begin{example}{}{}
Consider the $SU(2)_1$ WZW model. As will be argued more fully elsewhere, the lines supported on the left-moving chiral algebra are described by $\mathrm{Irr}(\mathrm{Sym}^\dagger(V))=SU(2)$ while the lines of the full CFT are $\mathrm{Irr}(\mathrm{Sym}^\dagger(\mathrm{CFT}_V))=SU(2)_L\times SU(2)_R/\mathbb{Z}_2$, with  $SU(2)$ being identified with $SU(2)_L$ (cf.\ Example \ref{ex:cleftamalg}). The Kapustin-Saulina picture of Figure \ref{fig:KS} shows that the general situation is analogous.
\end{example}

We then have a result which says that, in the canonical diagonal theory $\mathrm{CFT}_V$, every boundary condition can be generated by fusing bulk topological lines onto a single known boundary condition.

\begin{claim}{}{}
    Fix any boundary condition $b$ of $\mathrm{CFT}_V$. Then any other simple boundary condition $b'$ occurs as a summand of $\mathcal{L}\otimes b$, where $\mathcal{L}$ is some topological line in $\mathrm{Sym}^\dagger(\mathrm{CFT}_V)$ and $\mathcal{L}\otimes b$ is the boundary obtained by parallel fusion.
\end{claim}

In fact, we do not even need to consider the full category $\mathrm{Sym}^\dagger(\mathrm{CFT}_V)$ of topological lines in the bulk. It suffices to consider the subcategory $\mathrm{Sym}^\dagger(V)$ of lines coming from e.g.\ the left-moving chiral algebra $V$. Indeed, by Claim \ref{cl:claim:bdyextrap}, there is an equivalence
\begin{align}\label{eqn:linebdyequiv}
    \partial:\mathrm{Sym}^\dagger(V)\to \mathrm{Bdy}^\dagger(\mathrm{CFT}_V)
\end{align}
obtained by the folding trick. 
The category $\mathrm{Bdy}^\dagger(\mathrm{CFT}_V)$ is acted on in the obvious way by the subcategory $\mathrm{Sym}^\dagger(V)\subset \mathrm{Sym}^\dagger(\mathrm{CFT}_V)$ of bulk lines: namely, $\mathrm{Bdy}^\dagger(\mathrm{CFT}_V)$ transforms in the ``regular'' module category of $\mathrm{Sym}^\dagger(V)$, 
\begin{align}
    X\otimes b \cong \partial(X\otimes \partial^{-1}b).
\end{align}Thus, we can always take $\mathcal{L}=\partial^{-1}b'\otimes (\partial^{-1}b)^\ast$, in which case it clearly follows that $b'\subset \mathcal{L}\otimes b$.

We obtain an even stronger statement if we choose $b$ to be identity Cardy boundary condition, i.e.\ $b=\partial\mathds{1}$ the boundary condition obtained by folding the trivial line.

\begin{claim}{}{claim:identityCardybound}
    Let $b=\partial\mathds{1}$ be the identity Cardy boundary condition of $\mathrm{CFT}_V$. Every other simple boundary condition is of the form 
    \begin{align}
        b' = X\otimes \partial\mathds{1},
    \end{align} with $X$ a topological line in $\mathrm{Sym}^\dagger(V)\subset\mathrm{Sym}^\dagger(\mathrm{CFT}_V)$ coming from the left-moving chiral algebra. In particular, in a unitary diagonal rational CFT,
    \begin{align}
        g_{b'}=\dim(X)g_{\partial\mathds{1}} \geq g_{\partial\mathds{1}} = \sqrt{S_{11}},
    \end{align}
    where $S$ is the modular S-matrix of $\Rep(V)$. Hence, the identity Cardy state has the smallest $g$-function amongst all boundaries.
\end{claim}
In \cite{Collier:2021ngi}, the authors observed, by applying numerical bootstrap techniques, that the $g$-function of the identity Cardy boundary condition provides a lower bound for the $g$-function of \emph{stable} branes in several examples of diagonal 2D RCFTs. Claim \ref{cl:claim:identityCardybound} above shows that this bound applies to any boundary condition (i.e.\ not just the stable ones) in any diagonal theory $\mathrm{CFT}_V$.

Before moving on, recall that $\mathrm{Ver}(V/W)$ admits a hypergroup grading, Claim \ref{cl:cla:hypergroupgrading}. Since Claim \ref{cl:claim:VerS(VW)boundaries} identifies $\mathrm{Ver}(V/W)$ with $\mathrm{Bdy}^\dagger(\mathrm{CFT}_V/W)$,  it is natural to ask whether the hypergroup has an interpretation in terms of boundary conditions of $\mathrm{CFT}_V$. 

Indeed, suppose that a topological line $X\in \mathcal{B}(V)_g\subset\Ver(V/W)$ belongs to a graded-component of $\mathrm{Ver}(V/W)$ corresponding to an invertible element $g$ of the hypergroup. Invertible elements of the hypergroup can be thought of as automorphisms of $V$, and we claim that the boundary condition $\partial X$ of $\mathrm{CFT}_V$ obtained from $X$ by folding enjoys the twisted gluing condition from Equation \eqref{eqn:gluingaut} with $\Omega=g$.
 In the case that $r_i$ is a noninvertible element of the hypergroup, one should be able to obtain a richer collection of gluing conditions than have previously been considered, though the equations will likely involve twisted sectors. We leave a more detailed exploration of this question to the future.

\begin{figure}
    \begin{center}
        \input{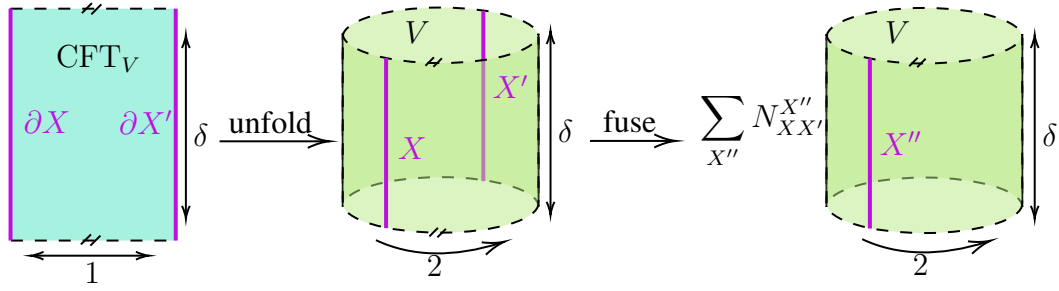}
        \caption{The pictorial interpretation of the relationship Equation \eqref{eqn:annulus} between annulus partition functions of $\mathrm{CFT}_V$ (far left) and twisted-sector partition functions of the chiral algebra $V$ (far right).}\label{fig:annulus}
    \end{center}
\end{figure}

We conclude by noting the following result on annulus partition functions. 

\begin{claim}{}{claim:annulus}
    The annulus partition function ${^{S^1\times I}}Z_{\partial X,\partial X'}(\delta)$ of $\mathrm{CFT}_V$ in the presence of boundary conditions $\partial X,\partial X' \in \mathrm{Bdy}^\dagger(\mathrm{CFT}_V/W)$ is
    \begin{align}\label{eqn:annulus}
        {^{S^1\times I}}Z_{\partial X,\partial X'}(\delta) = \sum_{X''\in \mathrm{Irr}(\Ver(V/W))} N_{XX'}^{X''}~Z_{X''}(\tau),  \ \ \ \ \ \tau=i\frac{\delta}{2},
    \end{align}
    where $N_{XX'}^{X''}$ are the fusion coefficients of $\Ver(V/W)$, and $Z_X(\tau)$ is the graded-dimension of the (twisted) $V$-module $V_X$ in Equation \eqref{eqn:twistedpf}. 
\end{claim}

Note that Claim \ref{cl:claim:annulus} generalizes the classic Cardy result (see e.g.\ \cite[Equations (29), (41)]{Cardy:2004hm}) on annulus partition functions involving boundaries in $\mathrm{Bdy}^\dagger(\mathrm{CFT}_V/V)$, see \cite{Behrend:1999bn,Fuchs:2002cm}. From our perspective it arises from the (un)folding trick, as depicted in e.g.\ Figure \ref{fig:annulus}. 

\section{Examples}\label{sec:examples}

In this section, we provide a number of examples which illustrate the generalities developed in Section \ref{sec:relativesymmetries} and Section \ref{sec:gluing}.

\subsection{A cleft example: invertible symmetries of free chiral boson CFTs}\label{subsec:cleft}

Recall from Example \ref{ex:latticeVOA} that $V_{2m}$ refers to the $c=1$ chiral algebra of the compact free boson CFT of radius $R^2=2m$. It lives on the boundary of $U(1)_{2m}$ Chern-Simons theory, and we write its irreducible representations as $V_{2m,r}$ with $r\in\mathbb{Z}_{2m}$. It admits a conformal subalgebra of the form
\begin{align}\label{eqn:latlatembedding}
    V_{2k}\subset V_{2m}
\end{align}
if and only if $k=mn^2$ for some positive integer $n\in\mathbb{Z}_{>0}$. Indeed, in this case, the vertex operators included in $V_{2mn^2}$ form a closed subset (under the OPE) of the vertex operators included in $V_{2m}$. 

We would like to compute the fusion category 
\begin{align}
    \mathcal{C}_{m\vert n}:=\Ver (V_{2m}/V_{2mn^2})
\end{align} 
of boundary topological line operators which commute with the $V_{2mn^2}$ subalgebra. 

\begin{claim}{}{}
    The category of topological line operators of $V_{2m}$ which preserve $V_{2mn^2}$ is given by 
    \begin{align}
        \mathcal{C}_{m\vert n}\cong \mathrm{Vec}_{\mathbb{Z}_{2mn}}^{\omega_m}
    \end{align}
    where $\omega_m$ is the 3-cocycle in $H^3(\mathbb{Z}_{2mn},U(1))\cong \mathbb{Z}_{2mn}$ corresponding to $m$ units of 't Hooft anomaly.\footnote{See e.g.\ \cite{Lin:2021udi} for an explicit expression.}
\end{claim}

\noindent The rest of this subsection is dedicated to deriving this claim and explaining related results (e.g.\ writing down the twisted sectors and their corresponding partition functions).

A useful fact is that $V_{2mn^2}$ can be obtained from $V_{2m}$ by passing to the invariant operators with respect to a $\mathbb{Z}_n$ automorphism, i.e.\ 
\begin{align}
    V_{2mn^2}=V_{2m}^{\mathbb{Z}_n}.
\end{align}
This $\mathbb{Z}_n$ symmetry is simply the $\mathbb{Z}_n$ subgroup of the $U(1)$ symmetry whose corresponding Noether current is $\partial \phi$, i.e.\ the $U(1)$ whose elements $e^{i\alpha}$ act on vertex operators in $V_{2m}$ as 
\begin{align}
    e^{i\sqrt{2m}\ell \phi(z)}\mapsto e^{i\alpha \ell}e^{i\sqrt{2m}\ell\phi(z)}.
\end{align}
In particular, it follows from the logic of Example \ref{ex:Ghypergroup} that the hypergroup induced by the conformal embedding in Equation \eqref{eqn:latlatembedding} is simply $\mathbb{Z}_n$ itself, 
\begin{align}
    K=\mathbb{Z}_n.
\end{align}
Moreover, by Claim \ref{cl:cla:hypergroupgrading}, $\mathcal{C}_{m\vert n}$ is a $\mathbb{Z}_n$-graded extension of $\mathcal{B}(V_{2m})$, 
\begin{align}
    \mathcal{C}_{m\vert n} \cong \bigoplus_{g\in\mathbb{Z}_n} \mathcal{B}(V_{2m})_g, \ \ \ \ \ \mathcal{B}(V_{2m})_1 = \mathcal{B}(V_{2m}),
\end{align}
where $\mathcal{B}(V_{2m})\cong \mathrm{Vec}_{\mathbb{Z}_{2m}}^{\sigma_m,\omega_m}$ with $\omega_m\in H^3(\mathbb{Z}_{2m},U(1))$ the 3-cocycle corresponding to $m$ units of 't Hooft anomaly.

Note further that this $\mathbb{Z}_n$ is a \emph{cleft} group of automorphisms, i.e.\ each element induces the trivial ribbon autoequivalence on $\mathcal{B}(V_{2m})$ (cf.\ the discussion around Equation \eqref{eqn:inducedautoequiv}). The easiest way to see this is to note that the group of ribbon autoequivalences of a modular tensor category is finite, and  that each element of $\mathbb{Z}_n$ is continuously connected to the identity automorphism in $\mathrm{Aut}(V_{2m})$. In our ``dome'' perspective on the effective hypergroup, Figure \ref{fig:domehypergroup}, this means that each surface $\mathcal{S}_g$ with $g\in\mathbb{Z}_n$ is the trivial surface.

From Example \ref{ex:cleftpt1}, we know that $\mathbb{Z}_n$ being cleft implies that, as a fusion category, $\mathcal{C}_{m\vert n}\cong \mathrm{Vec}_\Gamma^{\tilde\omega}$ for some extension $\Gamma$ of $\mathbb{Z}_n$ by $\mathbb{Z}_{2m}$, and some 3-cocycle $\tilde\omega \in H^3(\Gamma,U(1))$ which restricts to $\omega_m$ on $\mathbb{Z}_{2m}$. Another way to see that $\mathcal{C}_{m\vert n}$ should be pointed (i.e.\ group-like) is to compute its  dimension squared using Equation \eqref{eqn:dimver} and its rank using \eqref{eqn:Brank}, and observe that they are equal. Indeed, the dimension is simply given by 
\begin{align}
    \dim(\mathcal{C}_{m\vert n})^2=\dim(\mathcal{B}(V_{2m}))\dim(\mathcal{B}(V_{2mn^2}))=\sqrt{2m}\times \sqrt{2mn^2}=2mn.
\end{align}
To compute the rank, note that we have the following decomposition of $V_{2m}$-modules into $V_{2k}$-modules by restriction, 
\begin{align}\label{eqn:latlatdecomp}
    V_{2m,r}\cong \bigoplus_{\ell=0}^{n-1}V_{2k,2mn\ell+rn}.
\end{align}
This implicitly fixes the entries of the matrix $B$ appearing in Equation \eqref{eqn:Brank}, and one finds that 
\begin{align}
    \mathrm{rank}(\mathcal{C}_{m\vert n})=2mn.
\end{align}
It follows that $\mathcal{C}_{m\vert n}$ is a pointed fusion category.

To pin down the group $\Gamma$, we can use the induction (i.e.\ bulk-to-wall) functor 
\begin{align}
    I:\mathcal{B}(V_{2mn^2})\to \mathcal{C}_{m\vert n}\cong \mathcal{B}(V_{2mn^2})_A
\end{align} 
from Equation \eqref{eqn:induction} to constrain the fusion rules. Call $\widetilde{\eta}$ the generator of $\mathcal{B}(V_{2mn^2})$ and $\eta$ the generator of $\mathcal{B}(V_{2m})$. Applying the decomposition in Equation \eqref{eqn:latlatdecomp} to the vacuum module corresponding to $r=0$, we learn that the algebra $A$ in $\mathcal{B}(V_{2mn^2})$ which controls the conformal extension from $V_{2mn^2}$ to $V_{2m}$ is given by 
\begin{align}
    A \cong \bigoplus_{\ell=0}^{n-1} \widetilde{\eta}^{2mn\ell}.
\end{align}
In particular, we note that $R(I(\widetilde{\eta}^j))=\bigoplus_{\ell=0}^{n-1}\widetilde{\eta}^{2mn\ell+j}$ (cf.\ Equation \eqref{eqn:ResInd}), where $R:\mathcal{B}(V_{2mn^2})_A\to \mathcal{B}(V_{2mn^2})$ is the restriction functor which acts by forgetting the $A$-module structure of its input (see also Figure \ref{fig:physicalrestriction} for the physical interpretation).

By applying Frobenius reciprocity, Equation \eqref{eqn:Frobeniusreciprocity}, we learn that 
\begin{align}\label{eqn:latlatfrobrep}
\begin{split}
    \mathrm{Hom}_{\mathcal{C}_{m\vert n}}(I(\widetilde{\eta}^i),I(\widetilde{\eta}^j))&=\mathrm{Hom}_{\mathcal{B}(V_{2mn^2})}(\widetilde{\eta}^i,R(I(\widetilde{\eta}^j)))=\delta_{i= j~\mathrm{mod}~2mn}.
\end{split}
\end{align}
Since induction preserves quantum dimensions and in these categories an object is simple if and only if it has quantum dimension $1$,  it follows that 
\begin{align}
    \kappa^i\equiv I(\widetilde{\eta}^i)
\end{align} 
is simple, and Equation \eqref{eqn:latlatfrobrep} implies that 
\begin{align}\label{eqn:latlatfrobrepB}
    B_{\kappa^i,\widetilde{\eta}^j} = \delta_{i=j~\mathrm{mod}~2mn}.
\end{align}
 We can further deduce from \eqref{eqn:latlatfrobrep}  that $\kappa^i\cong \kappa^j$ if and only if $i=j~\mathrm{mod}~2mn$.
Since $\mathcal{C}_{m\vert n}$ has rank $2mn$, it follows that the $\kappa^i$, with $i=0,\dots,2mn-1$, exhaust all the simple lines of the category, 
\begin{align}
    \mathcal{C}_{m\vert n}=\{\kappa^i\mid i=0,\dots,2mn-1\}.
\end{align}Moreover, since induction is a tensor functor (Equation \eqref{eqn:inductiontensorfunctor}), we have that $\kappa^{i}\otimes \kappa^{j}\cong \kappa^{i+j}$, so that 
\begin{align}
    \Gamma \cong \mathbb{Z}_{2mn},
\end{align}
i.e.\ $\kappa$ generates a cyclic group of order $2mn$.

To determine the 3-cocycle $\tilde\omega$, we can use Equation \eqref{eqn:Drinfeldcenterverlinde}, which says that
\begin{align}
    Z(\mathrm{Vec}_\Gamma^{\tilde\omega})\cong \overline{\mathrm{Vec}_{\mathbb{Z}_{2m}}^{\sigma_m,\omega_m}}\boxtimes \mathrm{Vec}_{\mathbb{Z}_{2mn^2}}^{\sigma_{mn^2},\omega_{mn^2}}.
\end{align}
Now, $Z(\mathrm{Vec}_\Gamma^{\tilde\omega})$ is simply the category underlying ($\tilde\omega$-twisted) Dijkgraaf-Witten theory associated to $\Gamma$ \cite{Dijkgraaf:1989pz}, also known as the twisted quantum double model. The fusion rules of the twisted quantum double of a cyclic group were computed in \cite{Coste:2000tq}, where it was found that, up to equivalences induced by automorphisms of $\mathbb{Z}_{2mn}$, only the 3-cocycle $\tilde\omega=\omega_{m}\in H^3(\mathbb{Z}_{2mn},U(1))\cong \mathbb{Z}_{2mn}$ leads to fusion rules of the form $\mathbb{Z}_{2m}\times \mathbb{Z}_{2mn^2}$. Thus, in total we find that 
\begin{align}
    \mathcal{C}_{m\vert n}\cong \mathrm{Vec}_{\mathbb{Z}_{2mn}}^{\omega_m}.
\end{align}

We can plug in the matrix $B$ from Equation \eqref{eqn:latlatfrobrepB} into Equation \eqref{eqn:twisteddecomposition} to determine how the $\kappa^i$ twisted sector of $V_{2m}$ (i.e.\ the Hilbert space of local operators supported at the end of the boundary topological line $\kappa^i$, see Figure \ref{fig:stateopcorrtwisted}) decomposes into representations of $V_{2mn^2}$, 
\begin{align}\label{eqn:kappatwistedsectors}
    (V_{2m})_{\kappa^i}=\bigoplus_{\ell=0}^{n-1}V_{2mn^2,2mn\ell+i}.
\end{align}
In particular, the twisted partition function in Equation \eqref{eqn:twistedpf} is given by 
\begin{align}
    Z_{\kappa^i}(\tau) = \sum_{\ell=0}^n\chi_{2mn^2,2mn\ell+i}(\tau)
\end{align}
where here, 
\begin{align}
    \chi_{2k,r}(\tau)  := \mathrm{Tr}_{V_{2m,r}}q^{L_0-\frac{1}{24}}= \frac{1}{\eta(\tau)}\sum_{n\in\mathbb{Z}}q^{\frac{(r+2mn)^2}{4m}}
\end{align}
are the characters of the chiral boson algebra $V_{2k}$, with $\eta(\tau) = q^{1/24}\prod_{n=1}^\infty(1-q^n)$ the Dedekind-eta function.

Before moving on, we comment in passing that we can also formally consider the category $\Ver(V_{2m}/\widehat{\mathfrak{u}}(1))$ of topological lines of $V_{2m}$ which commute with the $U(1)$ Kac-Moody subalgebra $\widehat{\mathfrak{u}}(1)$. This is a kind of $n\to \infty$ limit of $\mathcal{C}_{m\vert n}$, and it is clear that the simple topological lines will be described by 
\begin{align}
    \mathrm{Irr}(\Ver(V_{2m}/\widehat{\mathfrak{u}}(1))) = \widetilde{U}(1) ,
\end{align}
where $\widetilde{U}(1)$ is a $2m$-fold cover of the hypergroup,
\begin{align}
    K=\widetilde{U}(1)/\mathbb{Z}_{2m} = U(1).
\end{align}
In other words, for each element $g$ of $U(1)$, there are exactly $2m$ $g$-twisted modules.

\subsection{A non-chiral relative example: \texorpdfstring{$\mathbb{Z}_2$}{Z2}-even sector of the Ising CFT}\label{subsec:nonchiralexample}

Even though we have chosen to present the results of Section \ref{sec:relativesymmetries} in terms of chiral algebras for simplicity, there is no conceptual obstruction to applying our techniques more or less ``out-of-the-box'' to any (relative) conformal field theory living at the boundary of a semi-simple bulk TQFT, including those with both left- and right-movers. Indeed, in this more general setting, one still expects that finite symmetries one-to-one correspond with conformal subalgebras of finite index. 

Of course in practice, the easiest subalgebras to discover are typically those of the form $V_L\otimes \overline{V}_R\subset \mathcal{H}$, with $V_L$ containing purely holomorphic operators and $\overline{V}_R$ containing purely anti-holomorphic operators. Working with such tensor factorizing subalgebras is not a serious limitation in rational relative CFTs, as it is expected that any finite index conformal subalgebra $\mathcal{A}\subset\mathcal{H}$ will itself contain a conformal subalgebra of the form $V_L\otimes \overline{V}_R$ with finite index; thus, by the Galois theory of Section \ref{subsec:Galois}, the symmetry category $\Ver (\mathcal{H}/\mathcal{A})$ exposed by an arbitrary conformal subalgebra $\mathcal{A}$ will always be a subcategory of $\Ver (\mathcal{H}/V_L\otimes \overline{V}_R)$ for some $V_L$ and $V_R$.  See Section 7.3 of \cite{Moller:2024xtt} for further discussion.

Let us treat a non-chiral example here in detail. In the remainder of this subsection, we will classify all topological line operators of the $\mathbb{Z}_2$-even sector $\mathcal{A}$ of the Ising CFT, which can be thought of as a relative conformal field theory in its own right living on the boundary of 3D $\mathbb{Z}_2$ gauge theory (i.e.\ the IR limit of the toric code). We achieve this by considering the $c=\sfrac12$ Virasoro conformal subalgebra $L_{\sfrac12}\otimes \overline{L}_{\sfrac12}\subset\mathcal{A}$ and computing $\Ver (\mathcal{A}/L_{\sfrac12}\otimes \overline{L}_{\sfrac12})$; since any topological line operator must commute with the Virasoro subalgebra, $\Ver (\mathcal{A}/L_{\sfrac12}\otimes \overline{L}_{\sfrac12})$ contains all noninvertible symmetries of $\mathcal{A}$.

\begin{claim}{}{claim:Z2evenIsing}
    The full category 
    \begin{align}
        \mathcal{C}:=\Ver(\mathcal{A}/L_{\sfrac12}\otimes \overline{L}_{\sfrac12})
    \end{align} 
    of boundary topological line operators of the $\mathbb{Z}_2$-even sector $\mathcal{A}$ of the Ising CFT has 
    \begin{align}
        \mathrm{Irr}(\mathcal{C}) = \{\mathds{1},e,m,f,\mathcal{N}_+,\mathcal{N}_-\},
    \end{align}
    with the fusion rules of $\mathds{1},e,m,f$ governed by the toric code and the rest given by 
    \begin{align}
    \begin{split}
        &\mathcal{N}_\pm\otimes \mathcal{N}_\pm = \mathds{1}\oplus f, \ \ \ \mathcal{N}_\pm \otimes \mathcal{N}_\mp \cong e\oplus m, \ \ f\otimes \mathcal{N}_\pm \cong \mathcal{N}_\pm \otimes f \cong \mathcal{N}_\pm \\
        &\hspace{.6in}e\otimes \mathcal{N}_\pm \cong \mathcal{N}_\pm \otimes e \cong \mathcal{N}_\mp , \ \ \ m\otimes \mathcal{N}_\pm \cong \mathcal{N}_\pm \otimes m \cong \mathcal{N}_\mp.
        \end{split}
    \end{align}
    The effective hypergroup which acts faithfully on genuine local operators of $\mathcal{A}$ is $K=\mathbb{Z}_2$.
\end{claim}

\noindent Before calculating $\mathcal{C}$, we note that everything we say here applies also to many examples of chiral algebras. For example, the chiral algebra $\mathfrak{spin}(16)_1$ can also be thought of as a relative CFT living at the boundary of 3D $\mathbb{Z}_2$ gauge theory. Furthermore, the conformal embedding 
\begin{align}
    L_{\sfrac12}\otimes \mathfrak{spin}(15)_1 \subset \mathfrak{spin}(16)_1,
\end{align}
induces an action of the category $\mathcal{C}$ from Claim \ref{cl:claim:Z2evenIsing} on $\mathfrak{spin}(16)_1$, since this conformal embedding behaves identically to the embedding $L_{\sfrac12}\otimes \overline{L}_{\sfrac12}\subset\mathcal{A}$ at the level of category theory/TQFT.

Denote the anyons of $\mathbb{Z}_2$ gauge theory (which we henceforth label $\mathrm{TC}$) as $\mathds{1},e,m,f$. Let $L_{\sfrac12}(h)$ be the representation of the $c=\sfrac12$ Virasoro algebra with conformal dimension $h$, and abbreviate $[h_1,h_2]:=L_{\sfrac12}(h_1)\otimes \overline{L_{\sfrac12}(h_2)}$. We label the $9$ bulk anyons of the $\mathrm{Ising}\boxtimes\overline{\mathrm{Ising}}$ modular tensor category as $(a,b)$, with $a,b\in \{1,\epsilon,\sigma\}$. This modular tensor category supports the relative conformal field theory  $L_{\sfrac12}\otimes \overline{L_{\sfrac12}}$ on its boundary. 

The full Hilbert space $\mathcal{H}$ of the Ising CFT decomposes into representations of the Virasoro algebras as 
\begin{align}
    \mathcal{H}=[0,0]\oplus [\sfrac12,\sfrac12]\oplus[\sfrac1{16},\sfrac1{16}].
\end{align}
The $\mathbb{Z}_2$-odd operators of the Ising CFT are those living in the $[\sfrac1{16},\sfrac1{16}]$ summand, and so the Hilbert space of local operators of the $\mathbb{Z}_2$-even sector decomposes into Virasoro representations as
\begin{align}
    \mathcal{A} = [0,0]\oplus [\sfrac12,\sfrac12].
\end{align}
The $\mathbb{Z}_2$-odd operators of the full Ising CFT manifest in the theory $\mathcal{A}$ as boundary local operators living at the end of the bulk $e$ line of $\mathbb{Z}_2$ gauge theory 
\begin{align}\label{eqn:edecomposition}
    \mathcal{A}_e=[\sfrac1{16},\sfrac{1}{16}].
\end{align}
Meanwhile, the even and odd $\mathbb{Z}_2$-twisted operators of the full Ising CFT become the Hilbert spaces of boundary local operators living at the end of the bulk $m$ and $f$ lines, respectively. These are known to decompose as 
\begin{align}
    \mathcal{A}_m=[\sfrac1{16},\sfrac1{16}]', \ \ \ \ \mathcal{A}_f = [0,\sfrac12]\oplus [\sfrac12,0],
\end{align}
where the prime $'$ indicates that, although $\mathcal{A}_e$ and $\mathcal{A}_m$ are identical as $L_{\sfrac12}\otimes \overline{L}_{\sfrac12}$-modules, they possess different structures as $\mathcal{A}$-modules. 

Before computing $\mathcal{C}$, let us first compute the effective hypergroup which grades it. First, note that the effective hypergroups induced by the embeddings $\mathcal{A}\subset \mathcal{H}$ and $L_{\sfrac12}\otimes\overline{L}_{\sfrac12}\subset\mathcal{H}$ are $\mathbb{Z}_2$ and $K_{R}$, respectively, where $K_R$ is the hypergroup induced by the fusion ring $R$ of the $\mathbb{Z}_2$ Tambara-Yamagami category $\mathrm{TY}_+(\mathbb{Z}_2)$ (see Appendix \ref{app:hypergroups}). Indeed, this follows simply because $\mathcal{A}$ is the $\mathbb{Z}_2$ even sector of $\mathcal{H}$, and $L_{\sfrac12}\otimes \overline{L}_{\sfrac12}$ is the symmetric sector of $\mathcal{H}$ with respect to the full $\mathrm{TY}_+(\mathbb{Z}_2)$ Verlinde symmetry. Thus, it follows from comments in Section \ref{subsec:hypergroup} that the effective hypergroup induced by the embedding $L_{\sfrac12}\otimes \overline{L}_{\sfrac12}\subset\mathcal{A}$ is $K_R\sslash \mathbb{Z}_2$. It is an elementary exercise to show that this quotient is $\mathbb{Z}_2$ again, so that $\mathcal{C}$ is a $\mathbb{Z}_2$-crossed braided extension of $\mathcal{B}(\mathcal{A})\cong \mathrm{TC}$, 
\begin{align}
   \mathcal{C}\cong \bigoplus_{g\in\mathbb{Z}_2}\mathcal{B}(\mathcal{A})_g, \ \ \ \ \ \ \mathcal{B}(\mathcal{A})_0=\mathcal{B}(\mathcal{A}),
\end{align}
with the identity graded component being given by the toric code.

The decompositions of the sectors $\mathcal{A}_a$ into representations of $L_{\sfrac12}\otimes \overline{L}_{\sfrac12}$ allow us to compute, via Equation \eqref{eqn:Brank}, that the rank of $\mathcal{C}$ is $6$. Since the identity component is the rank 4 toric code, the rank of the non-identity component must be 
\begin{align}
    \mathrm{rk}(\mathcal{B}(\mathcal{A})_1)=2.
\end{align}
Let us call the two topological lines which span the non-identity graded component $\mathcal{N}_\pm$.

To discover the properties of these lines, we compute the induction functor 
\begin{align}
    I:\mathrm{Ising}\boxtimes\overline{\mathrm{Ising}}\to \mathcal{C}
\end{align} 
from Equation \eqref{eqn:induction}, i.e.\ the bulk-to-wall functor from the $\mathrm{Ising}\boxtimes \overline{\mathrm{Ising}}$ TQFT to the topological interface $\mathcal{I}_A$ obtained by 1-form gauging in half of spacetime using the condensable algebra $A=(1,1)\oplus (\epsilon,\epsilon)$. Let us in particular consider the two (potentially non-simple) topological lines $I(\sigma,1)$ and $I(1,\sigma)$ on the wall. Observing that
\begin{align}\label{eqn:preinductionIsingex}
    A\otimes (\sigma,1)=(\sigma,1)\oplus (\sigma,\epsilon), \ \ \ \ A\otimes(1,\sigma)=(1,\sigma)\oplus(\epsilon,\sigma),
\end{align}
Frobenius reciprocity, Equation \eqref{eqn:Frobeniusreciprocity}, allows us to conclude that these two lines are simple, because
\begin{align}
    \mathrm{Hom}_{\mathcal{C}}(I(\sigma,1),I(\sigma,1))=\mathrm{Hom}_{\mathrm{Ising}\boxtimes\overline{\mathrm{Ising}}}((\sigma,1),(\sigma,1)\oplus(\sigma,\epsilon))=1,
\end{align}
and similarly for $I(1,\sigma)$. On the other hand, by applying Frobenius reciprocity again, we find that these two lines do not admit a topological junction between them, 
\begin{align}
    \mathrm{Hom}_{\mathcal{C}}(I(\sigma,1),I(1,\sigma))=\mathrm{Hom}_{\mathrm{Ising}\boxtimes\overline{\mathrm{Ising}}}((\sigma,1),(1,\sigma)\oplus(\epsilon,\sigma))=0,
\end{align}
and hence are non-isomorphic. Further applications of Frobenius reciprocity lead us to find that they are not isomorphic to any of the lines in the toric code subcategory, e.g.\ 
\begin{align}
    \mathrm{Hom}_{\mathcal{C}}(I(\sigma,1),e)=\mathrm{Hom}_{\mathrm{Ising}\boxtimes\overline{\mathrm{Ising}}}((\sigma,1),(\sigma,\sigma))=1,
\end{align}
where we have used Equation \eqref{eqn:edecomposition} to determine the restriction of the $e$ line. Thus, the two lines in the non-identity graded component of the Verlinde category are 
\begin{align}
    \mathcal{N}_+ = I(\sigma,1), \ \ \ \ \ \mathcal{N}_-=I(1,\sigma)
\end{align}
and, from Equation \eqref{eqn:preinductionIsingex}, the corresponding Hilbert spaces of local operators at their endpoints decompose into $L_{\sfrac12}\otimes\overline{L}_{\sfrac12}$-modules as
\begin{align}
    \mathcal{A}_{\mathcal{N}_+}\cong [\sfrac1{16},0]\oplus [\sfrac1{16},\sfrac12], \ \ \ \ \ \  \mathcal{A}_{\mathcal{N}_-}\cong [0,\sfrac1{16}]\oplus [\sfrac12,\sfrac{1}{16}].
\end{align}
For completeness, we calculate the bulk-to-wall map on the rest of the simples of $\mathrm{Ising}\boxtimes\overline{\mathrm{Ising}}$, which can be achieved by repeated applications of Frobenius reciprocity:
\begin{align}
\begin{split}
    &\hspace{.85in}I(1,1)=I(\epsilon,\epsilon)=\mathds{1}, \ \ \ \ \ I(\epsilon,1)=I(1,\epsilon) = f, \\
    &I(\sigma,1)=I(\sigma,\epsilon)=\mathcal{N}_+, \ \ \ \ \ I(1,\sigma)=I(\epsilon,\sigma)=\mathcal{N}_-, \ \ \ \ \ I(\sigma,\sigma)=e\oplus m.
\end{split}
\end{align}

Using that induction is a tensor functor, Equation \eqref{eqn:inductiontensorfunctor}, we can straightforwardly compute the fusion rules of these lines. For example,
\begin{align}
\begin{split}
    &\mathcal{N}_+\otimes\mathcal{N}_+=I((\sigma,1)\otimes(\sigma,1))=I((1,1)+(\epsilon,1))=\mathds{1}\oplus f, \\
    &\mathcal{N}_-\otimes \mathcal{N}_-=I((1,\sigma)\otimes (1,\sigma))=I((1,1)\oplus (1,\epsilon))=\mathds{1}\oplus f, \\
    &\hspace{.2in}\mathcal{N}_+\otimes \mathcal{N}_-\cong I((\sigma,1)\otimes (1,\sigma))=I(\sigma,\sigma)=e\oplus m .
\end{split}
\end{align}
Continuing in this way, one recovers the rest of the fusion rules of the $\mathbb{Z}_2$-crossed braided extension of the toric code described in Section I.2 of \cite{Barkeshli:2014cna}, 
\begin{align}
\begin{split}
    &\hspace{1in}f\otimes \mathcal{N}_\pm\cong \mathcal{N}_\pm \otimes f \cong \mathcal{N}_\pm, \\
    &e\otimes \mathcal{N}_\pm \cong \mathcal{N}_\pm \otimes e \cong \mathcal{N}_\mp, \ \ \ m\otimes \mathcal{N}_\pm \cong \mathcal{N}_\pm \otimes m \cong \mathcal{N}_\mp.
\end{split}
\end{align}
Now, $\mathcal{B}(\mathcal{A})_1$ is a rank-2 module category of the toric code. In particular, when we inflate the lines $\mathcal{N}_\pm$ into strips as in Figure \ref{fig:pinch}, we will not find that they are attached to the identity surface (corresponding to the regular module category of rank 4), but rather a non-trivial bulk surface. In this case, the only option with the correct $\mathbb{Z}_2$ fusion rule is the electric-magnetic duality surface, called $S_\psi$ in \cite{Roumpedakis:2022aik}. See Figure \ref{fig:Np} for a depiction.

\begin{figure}
    \centering
    \input{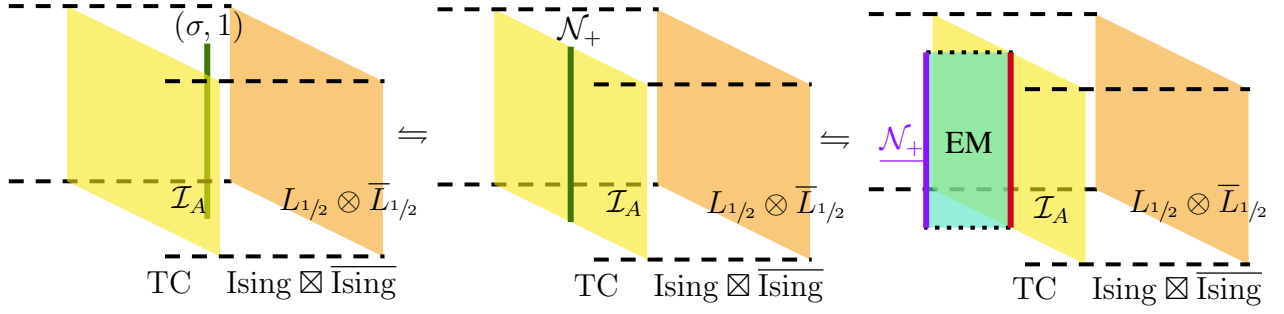}
    \caption{The topological line $\mathcal{N}_+$ (middle) can be obtained by pushing the line $(\sigma,1)$ in $\mathrm{Ising}\boxtimes\overline{\mathrm{Ising}}$ (left) onto the interface $\mathcal{I}_A$. Pulling $\mathcal{N}_+$ further into the bulk $\mathbb{Z}_2$ gauge theory $\mathrm{TC}$ leads to a line $\underline{\mathcal{N}_+}$ attached to a  $\mathbb{Z}_2$ surface which implements electric-magnetic duality (right).}\label{fig:Np}
\end{figure}

We can also determine the action of the lines $\mathcal{N}_\pm$ on the Hilbert spaces $\mathcal{A}_{\mathds{1}}$, $\mathcal{A}_e$, $\mathcal{A}_m$, $\mathcal{A}_f$. This can be done using Equation \eqref{eqn:inducedaction}, where we find 
\begin{align}
    \widehat{\mathcal{N}}_\pm \cdot \mathcal{O} = \begin{cases}
        \mathcal{O}, & \mathcal{O}\in [0,0]\subset\mathcal{A}, \\
        -\mathcal{O},&\mathcal{O}\in [\sfrac12,\sfrac12]\subset\mathcal{A}, \\
        \pm \mathcal{O}, & \mathcal{O}\in[0,\sfrac12]\subset \mathcal{A}_f \\
        \mp \mathcal{O}, & \mathcal{O}\in [\sfrac12,0]\subset\mathcal{A}_f,\\ 
        0, & \mathcal{O}\in \mathcal{A}_e,\mathcal{A}_m.
    \end{cases}
\end{align}
The action of $\widehat{\mathcal{N}}_\pm$ can also be viewed as the action of a dome of EM surface on (non-genuine) boundary local operators, as in Figure \ref{fig:autV}.
The involution induced on $\mathcal{A}$ by this action is essentially the Miyamoto involution ``of $\sigma$-type'' studied in \cite{miyamoto1996griess}.

We conclude by noting that, by Claim \ref{cl:claim:0gauging}, since $\mathcal{A}$ consists of states invariant under the action of $\mathbb{Z}_2$ on the Ising CFT $\mathcal{H}$, the category $\mathcal{C}$ could have also been computed as a $\mathbb{Z}_2$-equivariantization, 
\begin{align}
    \mathcal{C}\cong \mathrm{TY}_+(\mathbb{Z}_2)^{\mathbb{Z}_2}, 
\end{align}
where, at the level of permuting anyons, the $\mathbb{Z}_2$ acts trivially on $\mathrm{TY}_+(\mathbb{Z}_2)$ (although it acts non-trivially in the more subtle data of the action described in Section \ref{subsec:topman}). We will describe an equivariantization calculation in detail in the next section.

\subsection{A noninvertible hypergroup: Haagerup categories and chiral CFTs}\label{subsec:haagerup}

We turn to our first example of a noninvertible hypergroup action, which is closely related to the Haagerup fusion category. We begin with some motivation.

There are two closely related open questions in 2D CFT and the theory of vertex operator algebras.
\begin{enumerate}[label=\arabic*)]
\item Does every unitary fusion category $\CC$ arise as a symmetry category of a unitary 2D conformal field theory with a unique vacuum?
\item Does every modular tensor category $\mathcal{B}$ arise as the representation category of a strongly rational vertex operator algebra?
\end{enumerate}
The common belief is that both questions possess a positive answer,\footnote{In fact, the answer to the second question being yes whenever $\mathcal{B}=Z(\mathcal{C})$ implies that the answer to the first question is yes, as we will illustrate below in Claim \ref{cl:cla:haagerupchiralCFT} below.} though it appears quite difficult to obtain a proof with present methods. The state of the art has been to try to obtain evidence in favor of this common belief by demonstrating that it is true when $\mathcal{C}$ and $\mathcal{B}$ are taken to be ``exotic'' categories. 

The prototypical example of an ``exotic'' fusion category is the Haagerup fusion category $\textsl{Hg}$ \cite{Haagerup1994,AsaedaHaagerup1999,Izumi:2001mi}, whose simple objects are 
\begin{align}\label{eqn:Haagerupsimples}
    \mathds{1},~~\alpha,~~\alpha^2,~~\rho,~~\alpha\rho,~~\alpha^2\rho,
\end{align}
and whose fusion rules are 
\begin{align}\label{eqn:hgfusion}
    \alpha^3=\mathds{1}, ~~~\alpha\rho=\rho\alpha^2,~~~\rho^2=\mathds{1}\oplus \rho\oplus\alpha\rho\oplus\alpha^2\rho.
\end{align}
See e.g.\ \cite{Izumi:2001mi,Evans:2010yr,grossman2012quantum,Evans:2015zga,Huang:2020lox} for the $F$-symbols.\footnote{There are 4 fusion categories with fusion rules given by Equation \eqref{eqn:hgfusion}. One is the $\textsl{Hg}$ fusion category. One is the Grossman-Snyder category, which is Morita equivalent to $\textsl{Hg}$ and can be thought of as the dual category obtained when one gauges the $\mathbb{Z}_3$ symmetry in $\textsl{Hg}$. The remaining two are non-unitary, and can be obtained by applying Galois operations to the two unitary categories.} One is therefore led to study the following two problems, inspired by the more general questions posed above. 
\begin{enumerate}[label=\arabic*)]
    \item Discover a unitary 2D conformal field theory with a unique vacuum which possesses a faithful action of $\textsl{Hg}$ by symmetries (see e.g.\ \cite{Huang:2021ytb,Huang:2021nvb,Vanhove:2021zop,Bottini:2025hri,Gang:2023ggt,Hung:2025gcp} for results in this direction).
    \item Find a rational chiral algebra/strongly rational vertex operator algebra $V^{\textsl{Hg}}$ whose category of representations is the Drinfeld center of the Haagerup fusion category, $\Rep(V^{\textsl{Hg}})\cong Z(\textsl{Hg})$ (see e.g.\ \cite{Evans:2010yr,gannon2019reconstruction}). Equivalently, find a gapless chiral boundary condition for a TQFT of the form $(\mathcal{B},c)=(Z(\textsl{Hg}),8n)$.
\end{enumerate}  
We will not attempt to solve these problems here. However, we will explain how our theory of generalized symmetries for relative CFTs leads to some reformulations.

Our first comment is that these two problems are related by symmetry/subalgebra duality, Claim \ref{cl:claim:sym/subduality}. To state the relation, recall that a purely chiral 2D CFT is a 2D CFT with only left-movers which lives at the boundary of an invertible TQFT. Equivalently, it is a rational chiral algebra whose only irreducible representation is itself (a.k.a. a holomorphic vertex operator algebra). Examples include the monster CFT $V^\natural$ and the chiral $(E_8)_1$ Wess-Zumino-Witten model.

\begin{claim}{}{cla:haagerupchiralCFT}
    There exists a purely chiral 2D CFT $V$ with central charge $c=8n$ and  $\mathcal{C}$ symmetry if and only if there exists a rational chiral algebra $V^{\mathcal{C}}$ with central charge $c=8n$ and  $\Rep(V^{\mathcal{C}})=Z(\mathcal{C})$.
\end{claim} 

\noindent See \cite{Burbano:2021loy,Rayhaun:2023pgc,Moller:2024xtt} for related discussions. We remark that there could of course exist CFTs which carry a $\mathcal{C}$ symmetry while not being purely chiral, but we will not discuss this possibility here. 

\begin{figure}
    \begin{center}
        \input{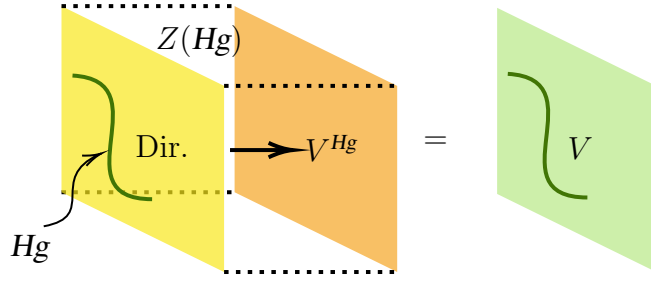}
    \end{center}
    \caption{A rational chiral algebra $V^{\textsl{Hg}}$ with $\Rep(V^{\textsl{Hg}})=Z(\textsl{Hg})$ exists if and only if there is a purely chiral 2D CFT $V$ with $\textsl{Hg}$ symmetry.}\label{fig:haagerupsandwich}
\end{figure}

\begin{proof}Suppose one is given $V^{\CC}$. The modular tensor category $\Rep(V^{\CC})\cong Z(\CC)$ admits a canonical Lagrangian algebra $L$ (corresponding to the ``Dirichlet'' boundary condition) satisfying
\begin{align}
    Z(\CC)_L\cong \CC, \ \ \ \ Z(\CC)_L^{\mathrm{loc}}\cong \textsl{Vec}.
\end{align} 
Thus, following the discussion in Section \ref{subsec:boundarylines}, if one uses $V^{\CC}$ as the ``physical boundary condition'' for a SymTFT construction as in Figure \ref{fig:haagerupsandwich}, then one obtains a conformal extension $V^{\CC}\subset V$ to a purely chiral 2D CFT $V$ which is acted on by $\CC$. 

In the reverse direction, if one is handed a chiral 2D CFT $V$ with a $\CC$ action, then one may pass to the conformal subalgebra $V^{\CC}$ of operators in $V$ which are transparent to the topological lines in $\CC$ (cf.\ Figure \ref{fig:VCdef}), which is guaranteed to satisfy $\Rep(V^{\CC})\cong Z(\CC)$.
\end{proof}

The minimal possibility is that $V^{\textsl{Hg}}$ exists at $c=8$. In \cite{Evans:2010yr}, it was shown that, if $V^{\textsl{Hg}}$ indeed exists at $c=8$, then there are two possibilities for its graded-dimension, 
\begin{align}
    \mathrm{Tr}_{V^{\textsl{Hg}}}q^{L_0-\sfrac13}=q^{-\sfrac13}(1+(6+13\gamma)q+(120+78\gamma)q^2+(956+351\gamma)q^3+\cdots),
\end{align}
with $\gamma=0$ or $\gamma=1$. Furthermore, Claim \ref{cl:cla:haagerupchiralCFT} would predict that the chiral $(E_8)_1$ Wess-Zumino-Witten model, by virtue of being the unique purely chiral CFT with $c=8$, possesses a $\textsl{Hg}$ symmetry \cite{Moller:2024xtt} (again, assuming that $V^{\textsl{Hg}}$ actually exists). One could then use the results of \cite{Choi:2024tri} (e.g.\ Equation (6.9) of op.\ cit.), or the generalizations presented in Section \ref{subsec:symmetryresolvedpf}, to write down the $\textsl{Hg}$-twisted genus-1 partition functions of chiral $(E_8)_1$ as linear combinations of the characters computed in \cite{Evans:2010yr}, with coefficients given in terms of half-braiding data associated with $\textsl{Hg}$. Some of this analysis was done in \cite{Albert:2025umy}.

Now, $\textsl{Hg}$ possesses a $\textsl{Vec}_{\mathbb{Z}_3}$ subcategory. Given a candidate purely chiral CFT $V$ which is believed to possess a $\textsl{Hg}$ symmetry, there are only finitely many ways, up to conjugation, that the $\mathbb{Z}_3$ subcategory could act. Finding a $\textsl{Hg}$ action on $V$ with the $\mathbb{Z}_3$ subcategory acting in a chosen way is equivalent to finding a particular noninvertible symmetry of $V^{\mathbb{Z}_3}$, the subalgebra of $V$ consisting of the $\mathbb{Z}_3$-invariant operators. In more detail, we have the following claim.

\begin{figure}
    \begin{center}
        \input{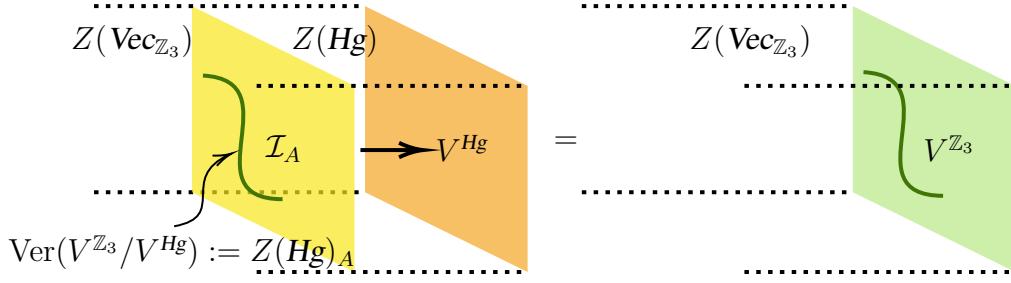}
        \caption{The existence of a Haagerup chiral algebra $V^{\textsl{Hg}}$ implies a noninvertible symmetry of $V^{\mathbb{Z}_3}$, where $V$ is a purely chiral CFT and $V^{\mathbb{Z}_3}$ is the subalgebra of neutral operators in $V$ with respect to a $\mathbb{Z}_3$ symmetry.}\label{fig:VZ3sandwich}
    \end{center}
\end{figure}

\begin{claim}{}{cla:Hg3}
    Assume a chiral CFT $V$ admits an action by $\textsl{Hg}$, and consider the conformal embedding $V^{\textsl{Hg}}\subset V^{\mathbb{Z}_3}$ of the $\textsl{Hg}$-invariant subalgebra into the $\mathbb{Z}_3$-invariant subalgebra. Then $V^{\textsl{Hg}}$ can be obtained from $V^{\mathbb{Z}_3}$ by passing to the states which are neutral under an action of the hypergroup 
    \begin{align}
        K_{\textsl{Hg}}\sslash\mathbb{Z}_3 = K_{\Ver(V^{\mathbb{Z}_3}/V^{\textsl{Hg}})}\sslash K_{Z(\textsl{Vec}_{\mathbb{Z}_3})}  = \{\mathds{1},a\}
    \end{align}
    whose multiplication rule is given in Equation \eqref{eqn:haghyp}. Furthermore, $V^{\mathbb{Z}_3}$ defines a gapless chiral boundary condition of 3D $\mathbb{Z}_3$ gauge theory, and there are $10$ simple topological line operators supported on this boundary which commute with the $V^{\textsl{Hg}}$ subalgebra,
    \begin{align}
    \begin{split}
        &\hspace{.3in}\Ver\left(V^{\mathbb{Z}_3}/V^{\textsl{Hg}}\right) = \mathcal{B}(V^{\mathbb{Z}_3})\oplus \mathcal{B}(V^{\mathbb{Z}_3})_a, \\
       & \mathcal{B}(V^{\mathbb{Z}_3})\cong Z(\textsl{Vec}_{\mathbb{Z}_3}), \ \ \ \ \ \mathrm{Irr}(\mathcal{B}(V^{\mathbb{Z}_3})_a)=\{\mathcal{L}\},
    \end{split}
    \end{align}
    $9$ of which are lines coming from anyons in the bulk.
    Letting $g$ denote an arbitrary element of $Z(\textsl{Vec}_{\mathbb{Z}_3})$, the fusion rules of $\Ver(V^{\mathbb{Z}_3}/V^{\textsl{Hg}})$, in addition to those of $\mathbb{Z}_3$ gauge theory, are
    \begin{align}
        g\otimes \mathcal{L} \cong \mathcal{L}\otimes g \cong \mathcal{L} , \ \ \ \ \ \ \ \ \mathcal{L}^2\cong 9\cdot \mathcal{L}\oplus \bigoplus_{g\in \mathrm{Irr}(Z(\textsl{Vec}_{\mathbb{Z}_3}))} g\, .
    \end{align}
    In particular, $\Ver(V^{\mathbb{Z}_3}/V^{\textsl{Hg}})$ is a near-group category of type $\mathbb{Z}_3\times\mathbb{Z}_3+9$ (see e.g.\ \cite{evans2014near} or \cite[Example 12.13]{izumi2015cuntz}).
\end{claim}
\noindent For example, take $V$ to be the chiral $(E_8)_1$ WZW model. It possesses a unique non-anomalous $\mathbb{Z}_3$ symmetry up to conjugation (see e.g.\ \cite{Burbano:2021loy}). If it further possessed a $\textsl{Hg}$ symmetry, then the $\mathbb{Z}_3$ subcategory would have to act according to this unique $\mathbb{Z}_3$, and Claim \ref{cl:cla:Hg3} would predict that $V^{\mathbb{Z}_3}\cong \widehat{\mathfrak{su}}(3)_1\otimes \widehat{(\mathfrak{e}_6)}_1$ admits an action by the hypergroup $K_{\textsl{Hg}}\sslash\mathbb{Z}_3$ such that the neutral operators describe a chiral algebra whose representation category is $Z(\textsl{Hg})$. Further, $\widehat{\mathfrak{su}}(3)_1\otimes\widehat{(\mathfrak{e}_6)}_1$ would admit a topological line operator $\mathcal{L}$ which, together with the lines coming from the bulk $\mathbb{Z}_3$ gauge theory, generates a near-group fusion category $\mathbb{Z}_3\times\mathbb{Z}_3+9$.

The rest of this subsection is dedicated to deriving this claim, and explaining related calculations.

Let us begin by determining the hypergroup. This can be easily achieved by applying the discussion around Equation \eqref{eqn:consecutiveembeddings} to the consecutive conformal embeddings $V^{\textsl{Hg}}\subset V^{\mathbb{Z}_3}\subset V$. Since $V^{\textsl{Hg}}$ and $V^{\mathbb{Z}_3}$ are obtained from $V$ as the fixed points of $\textsl{Hg}$ and $\mathbb{Z}_3$, respectively, it follows that $V^{\textsl{Hg}}$ can be obtained from $V^{\mathbb{Z}_3}$ as the fixed points of an action of the double coset hypergroup $K_{\textsl{Hg}}\sslash \mathbb{Z}_3$. This hypergroup was actually already computed in \cite[Lemma 4.26]{Bischoff:2016jmy}, where it was found that $K_{\textsl{Hg}}\sslash\mathbb{Z}_3=\{\mathds{1},a\}$ with
\begin{align}\label{eqn:haghyp}
   a\star a = \frac{1}{w}\cdot \mathds{1} + \frac{w-1}{w} \cdot a , \ \ \ \ \ \ w= \frac{11+3\sqrt{13}}{2}.
\end{align}
Thus, the problem of finding a Haagerup chiral algebra reduces to finding a suitable action of the hypergroup in Equation \eqref{eqn:haghyp} on a chiral algebra $V^{\mathbb{Z}_3}$ on the boundary of 3D $\mathbb{Z}_3$ gauge theory.

Let us next calculate the full category $\mathrm{Ver}\left( V^{\mathbb{Z}_3}/V^{\textsl{Hg}}\right)$ of topological line operators of $V^{\mathbb{Z}_3}$ which preserve the conformal subalgebra $V^{\textsl{Hg}}$. To this end, we note that the modular data of $Z(\textsl{Hg})$, and therefore also the fusion rules by Verlinde's formula, are given in \cite{Evans:2010yr}. We label the simple objects of $Z(\textsl{Hg})$ as $\mathds{1},X_1,\dots,X_{11}$, following the basis of op.\ cit. In this basis, the conformal extension $V^{\textsl{Hg}}\subset V^{\mathbb{Z}_3}$ is mediated by the condensable algebra
\begin{align}
    A=\mathds{1}\oplus X_1,
\end{align} 
in $Z(\textsl{Hg})$. Thus, by Claim \ref{cl:claim:verdef}, we have that 
\begin{align}
    \mathrm{Ver}\left( V^{\mathbb{Z}_3}/V^{\textsl{Hg}}\right)\cong Z(\textsl{Hg})_A,
\end{align}
where $Z(\textsl{Hg})_A$ is the category of $A$-modules.

We already have access to a very large subcategory of $Z(\textsl{Hg})_A$, namely the subcategory of \emph{local} $A$-modules,
\begin{align}
    Z(\textsl{Hg})_A^{\mathrm{loc}}\cong \Rep(V^{\mathbb{Z}_3})\cong Z(\textsl{Vec}_{\mathbb{Z}_3})\cong (A_2,1)\boxtimes \overline{(A_2,1)},
\end{align} 
where $(A_2,1)\equiv \Rep(SU(3)_1)$ is described in \cite[Section 5.3.3]{Rowell:2007dge}. Let us parametrize the simple objects of $Z(\textsl{Hg})_A^{\mathrm{loc}}$, which have $\mathbb{Z}_3\times\mathbb{Z}_3$ fusion rules,  as $(\eta^i,\eta^j)$.

It will be useful in what follows to compute the restriction functor, 
\begin{align}
    R:Z(\textsl{Hg})_A\to Z(\textsl{Hg}),
\end{align}
and its adjoint, the induction functor
\begin{align}
    I:Z(\textsl{Hg})\to Z(\textsl{Hg})_A,
\end{align}
given in Equation \eqref{eqn:restriction} and Equation \eqref{eqn:induction}, respectively. (See also Figure \ref{fig:physicalrestriction} and Figure \ref{fig:induction} for the physical interpretations of restriction and induction.) Let us parametrize how $R$ acts on the local modules in $Z(\textsl{Hg})_A^{\mathrm{loc}}\subset Z(\textsl{Hg})_A$ as
\begin{align}
    R(\eta^i,\eta^j) = \bigoplus_{k=0}^{11} B_{(i,j),k}X_k.
\end{align}
In the language of chiral algebras, the $B_{(i,j),k}$ are the multiplicities describing how irreducible representations of $V^{\mathbb{Z}_3}$ decompose into irreducible representations of $V^{\textsl{Hg}}$. 
By imposing that the $B_{(i,j),k}$ must intertwine modular data, Equation \eqref{eqn:intertwinemodular}, and using the known modular data of $Z(\textsl{Vec}_{\mathbb{Z}_3})\cong Z(\textsl{Hg})_A^{\mathrm{loc}}$ and $Z(\textsl{Hg})$, we can completely fix the action of restriction on local modules, $R\vert_{Z(\textsl{Hg})_A^{\mathrm{loc}}}$. In particular, one finds that 
\begin{align}\label{eqn:haagR}
\begin{split}
    &R(\mathds{1},\mathds{1}) = \mathds{1}\oplus X_1, \ \ \ R(\mathds{1},\eta)=R(\mathds{1},\eta^2) = X_5, \ \ \ R(\eta,\mathds{1})=R(\eta^2,\mathds{1})=X_4 \\
    &\hspace{.6in} R(\eta,\eta)=R(\eta^2,\eta^2) = X_2, \ \ \ R(\eta,\eta^2)=R(\eta^2,\eta) = X_3.
\end{split}
\end{align}
If we feed Equation \eqref{eqn:haagR} into Equation \eqref{eqn:Brank}, we can determine the rank of symmetry category we are after to be 
\begin{align}
    \mathrm{rk}\Ver\left( V^{\mathbb{Z}_3}/V^{\textsl{Hg}}\right) = \sum_{i,j=0}^2\sum_{k=0}^{11} (B_{(i,j),k})^2 = 10.
\end{align}
That is, in addition to the 9 simple lines in $Z(\textsl{Hg})_A^{\mathrm{loc}}$ which can be pulled into the bulk, there is one simple topological line operator which is trapped on the gapless chiral boundary defined by $V^{\mathbb{Z}_3}$, and cannot be pulled into the bulk $\mathbb{Z}_3$ gauge theory. Let us call this line $\mathcal{L}$. 

To gain more information about $\mathcal{L}$, let us finish the calculation of the restriction functor  by calculating $R(\mathcal{L})$. As an intermediate step, it will be helpful to compute the image of $X_{11}$ under the bulk-to-wall functor $I$. By using Frobenius reciprocity, Equation \eqref{eqn:Frobeniusreciprocity}, we observe that 
\begin{align}
    \mathrm{Hom}(I(X_{11}),(\eta^i,\eta^j))=\mathrm{Hom}(X_{11},R(\eta^i,\eta^j))=0,
\end{align}
for all $i,j$.
In particular, one finds that no simple line of $Z(\textsl{Hg})_A^{\mathrm{loc}}$ appears in the decomposition of $I(X_{11})$ into simple lines of $Z(\textsl{Hg})_A$, and therefore $I(X_{11})=n\cdot \mathcal{L}$. To compute the multiplicity $n$, we simply apply Frobenius reciprocity again, 
\begin{align}
    \mathrm{Hom}(I(X_{11}),I(X_{11}))&= \mathrm{Hom}(X_{11},R(I(X_{11})))=\mathrm{Hom}(X_{11},A\otimes X_{11})=1,
\end{align}
where we have used Equation \eqref{eqn:ResInd} for the second equality, and the fusion rules of $Z(\textsl{Hg})$ to deduce that $A\otimes X_{11}\cong \bigoplus_{k=1}^{11} X_k$ and obtain the third equality. This implies that $I(X_{11})$ is simple and therefore that $n=1$. From this, it follows that
\begin{align}
    \mathrm{Hom}(X_k,R(\mathcal{L})) = \mathrm{Hom}(X_k,R(I(X_{11})))=\mathrm{Hom}(X_k,A\otimes X_{11})=1-\delta_{k,0},
\end{align}
i.e.\ $R(\mathcal{L})\cong \bigoplus_{k\neq 0}X_k$. This completes the calculation of the restriction functor. 

We note in passing that, since induction preserves quantum dimensions, 
\begin{align}
    \dim(\mathcal{L})=\dim(X_{11})=\frac32(3+\sqrt{13}),
\end{align}
where we have read off the quantum dimension of $X_{11}$ from the modular S matrix of $Z(\textsl{Hg})$. 

One immediate physical application of the calculation of $R(\mathcal{L})$ is that we can compute the spectrum of the twisted-sector Hilbert space $V^{\mathbb{Z}_3}_{\mathcal{L}}$. Indeed, using Equation \eqref{eqn:twistedpfB}, we can express the twisted partition function $Z_{\mathcal{L}}(\tau)$ in terms of the characters of $V^{\textsl{Hg}}$, 
\begin{align}
    \mathrm{Tr}_{V^{\mathbb{Z}_3}_{\mathcal{L}}}q^{L_0-1} \equiv Z_{\mathcal{L}}(\tau) = \sum_{k=1}^{11} \mathrm{Tr}_{V^{\textsl{Hg}}_{X_k}}q^{L_0-1},
\end{align}
where we have written $V^{\textsl{Hg}}_{X_k}$ for the module of $V^{\textsl{Hg}}$ realized by boundary local operators on which the bulk line $X_k$ can end, cf.\ Figure \ref{fig:gaplesschiralboundary}. 

Next, let us compute the fusion rules in $\Ver(V^{\mathbb{Z}_3}/V^{\textsl{Hg}})$. First of all, we note that it is a $K_{\textsl{Hg}}\sslash \mathbb{Z}_3$-graded extension of $Z(\textsl{Vec}_{\mathbb{Z}_3})$, i.e.\ 
\begin{align}
    \mathrm{Ver}(V^{\mathbb{Z}_3}/V^{\textsl{Hg}}) \cong \mathcal{B}(V^{\mathbb{Z}_3})\oplus \mathcal{B}(V^{\mathbb{Z}_3})_a
\end{align}
where 
\begin{align}
    \mathcal{B}(V^{\mathbb{Z}_3})\cong Z(\textsl{Vec}_{\mathbb{Z}_3}), \ \ \ \ \  \mathrm{Irr}(\mathcal{B}(V^{\mathbb{Z}_3})_a)=\{\mathcal{L}\}.
\end{align}
From Equation \eqref{eqn:hypergroupgrading}, it immediately follows that 
\begin{align}
    (\eta^i,\eta^j)\otimes\mathcal{L}\cong \mathcal{L}\otimes (\eta^i,\eta^j) \cong \mathcal{L}.
\end{align}
On the other hand, we can compute the remaining fusion rule of $\mathcal{L}$ with itself using that induction is a tensor functor, Equation \eqref{eqn:inductiontensorfunctor},
\begin{align}
    \mathcal{L}\otimes \mathcal{L} \cong I(X_{11})\otimes I(X_{11})\cong I(X_{11}\otimes X_{11}) \cong \bigoplus_{k\neq 1,6} I(X_k) \cong 9\cdot \mathcal{L}\oplus \bigoplus_{i,j=0}^2 (\eta^i,\eta^j),
\end{align}
where, in order to determine $I(X_k)$ for all $k$, we have used that induction is the transpose of restriction, which was computed earlier. One can confirm that this fusion rule respects the grading by $K_{\textsl{Hg}}\sslash \mathbb{Z}_3$ in the sense of Equation \eqref{eqn:hypergroupgrading}. Indeed if one defines $\pi_{\mathds{1}}$ and $\pi_a$ to be the projections onto the two graded components, one recovers the structure constants of the hypergroup, 
\begin{align}
    \frac{\dim(\pi_{\mathds{1}}(\mathcal{L}^2))}{\dim(\mathcal{L}^2)} = \frac{1}{w}, \ \ \ \ \ \frac{\dim(\pi_{a}(\mathcal{L}^2))}{\dim(\mathcal{L}^2)} = \frac{w-1}{w},
\end{align}
where $w$ was defined in Equation \eqref{eqn:haghyp}. 

We note that the action of $\Ver(V^{\mathbb{Z}_3}/Z^{\textsl{Hg}})$ on the modules of $V^{\mathbb{Z}_3}$ can be computed easily. For example, the action of the $\mathbb{Z}_3$ toric code on the irreducible modules $V^{\mathbb{Z}_3}_{(\eta^i,\eta^j)}$ of $V^{\mathbb{Z}_3}$ can be obtained by the standard Verlinde-like formula in Equation \eqref{eqn:actionpullingbulk} (see also Figure \ref{fig:encirclingaction}) which is calculated by dragging $(\eta^i,\eta^j)$ into the bulk $\mathbb{Z}_3$ toric code, 
\begin{align}
    \widehat{(\eta^i,\eta^j)}\cdot \mathcal{O} = \frac{S^{Z(\textsl{Vec}_{\mathbb{Z}_3})}_{(i,j)(i',j')}}{S^{Z(\textsl{Vec}_{\mathbb{Z}_3})}_{(0,0)(i',j')}}\mathcal{O}= \omega^{ii'-jj'}\mathcal{O}, \ \ \ \ \ \ \text{if } \mathcal{O}\in V^{\mathbb{Z}_3}_{(\eta^{i'},\eta^{j'})},
\end{align}
where $\omega=\exp(2\pi \sqrt{-1}/3)$. Similarly, the action of $\mathcal{L}=I(X_{11})$ can be determined by dragging it into the bulk \emph{in the other direction} (i.e.\ into the bulk $Z(\textsl{Hg})$) and using Equation \eqref{eqn:inducedaction}, 
\begin{align}
    \widehat{\mathcal{L}}\cdot \mathcal{O} = \frac{S^{Z(\textsl{Hg})}_{11,k}}{S_{0,k}^{Z(\textsl{Hg})}}\mathcal{O}, \ \ \ \ \ \ \text{if } \mathcal{O}\in B_{(i,j)k}V^{\textsl{Hg}}_k \subset V^{\mathbb{Z}_3}_{(i,j)}.
\end{align}
For example, unpacking the action of $\mathcal{L}$ on the vacuum module $V^{\mathbb{Z}_3}\cong V^{\textsl{Hg}}\oplus V^{\textsl{Hg}}_{X_1}$ gives 
\begin{align}\label{eqn:Haagaction}
    \widehat{\mathcal{L}}\cdot \mathcal{O} = \begin{cases}
        \frac32(3+\sqrt{13})\mathcal{O}, &\text{ if }\mathcal{O}\in V^{\textsl{Hg}}\subset V^{\mathbb{Z}_3} \\
        \frac32(3-\sqrt{13})\mathcal{O}, & \text{ if }\mathcal{O}\in V^{\textsl{Hg}}_{X_1}\subset V^{\mathbb{Z}_3}.
    \end{cases}
\end{align}
Using the dome perspective of Section \ref{subsec:surfaces}, the action in \eqref{eqn:Haagaction} can equivalently be represented as in the left of Figure \ref{fig:domehypergroup} via a surface operator of $\mathbb{Z}_3$ gauge theory obtained by higher-gauging $\mathbb{Z}_3\times\mathbb{Z}_3$ with some choice of discrete torsion $\psi \in H^2(\mathbb{Z}_3\times\mathbb{Z}_3,U(1))$ \cite{Roumpedakis:2022aik}. It would be interesting to determine the correct choice of $\psi$.

Before moving on, let us comment on the generalization of Claim \ref{cl:cla:Hg3} to untwisted $\mathbb{Z}_{n}$ Haagerup-Izumi categories $\textsl{Hg}_{n}$, with $n$ odd. These have simple objects 
\begin{align}
    \mathrm{Irr}(\textsl{Hg}_{n}) \cong \{\alpha^i\mid i=0,\dots,n-1\}\cup \{\alpha^i\rho\mid i=0,\dots,n-1\}, 
\end{align}
with fusion rules 
\begin{align}
    \alpha^i\otimes \alpha^j \cong \alpha^{i+j},  \ \ \ \ \  \alpha^i\otimes \rho \cong \rho\otimes \alpha^{-i} , \ \ \ \ \ \rho^2\cong \mathds{1}\oplus \bigoplus_{k=0}^{n-1}\alpha^k\rho,
\end{align}
where we identify $\alpha^i\cong \alpha^{i+n}$. The category $\textsl{Hg}$ corresponds to the case $n=3$. These fusion categories are known to exist for $n\leq 29$ \cite{Izumi:2001mi,Evans:2010yr,Huang:2020lox,budinski2021exotic}, and conjectured to exist for all odd $n$. We then have the following.
\begin{claim}{}{}
    Assume a chiral CFT $V$ admits an action by $\textsl{Hg}_{n}$, with $n$ an odd integer. Then $V^{\textsl{Hg}_{n}}$ can be obtained from $V^{\mathbb{Z}_{n}}$ by passing to the states which are neutral under an action of the hypergroup 
    \begin{align}
        K_{\textsl{Hg}_{n}}\sslash \mathbb{Z}_{n} = K_{\mathrm{Ver}(V^{\mathbb{Z}_{n}}/V^{\textsl{Hg}_{n}})} \sslash K_{Z(\textsl{Vec}_{\mathbb{Z}_{n}})}=\{\mathds{1},a\}
    \end{align}
    with multiplication rule given by
    \begin{align}
        a\star a = \frac1w \mathds{1} + \frac{w-1}{w} a, \ \ \ \ w = \frac{2+n^2+n\sqrt{n^2+4}}{2}.
    \end{align}
    Furthermore, $V^{\mathbb{Z}_{n}}$ defines a gapless chiral boundary condition of $\mathbb{Z}_n$ gauge theory, and there are $n^2+1$ simple topological line operators supported on this boundary which commute with the $V^{\textsl{Hg}_{n}}$ subalgebra, 
    \begin{align}
    \begin{split}
        &\hspace{.3in}\Ver(V^{\mathbb{Z}_{n}}/V^{\textsl{Hg}_{n}})=\mathcal{B}(V^{\mathbb{Z}_n})\oplus \mathcal{B}(V^{\mathbb{Z}_n})_a, \\
        &\mathcal{B}(V^{\mathbb{Z}_n})\cong Z(\textsl{Vec}_{\mathbb{Z}_{n}}), \ \ \ \ \ \mathrm{Irr}(\mathcal{B}(V^{\mathbb{Z}_n})_a)= \{\mathcal{L}\},
    \end{split}
    \end{align}
    $n^2$ of which are lines coming from anyons in the bulk. Letting $g$ denote an arbitrary element of $Z(\textsl{Vec}_{\mathbb{Z}_{n}})$, the fusion rules of $\Ver(V^{\mathbb{Z}_{n}}/V^{\textsl{Hg}_{n}})$, in addition to those of the $\mathbb{Z}_{n}$ toric code, are 
    \begin{align}
        g\otimes \mathcal{L}\cong \mathcal{L}\otimes g \cong \mathcal{L}, \ \ \ \ \ \mathcal{L}^2\cong n^2\cdot \mathcal{L}\oplus \bigoplus_{g\in Z(\textsl{Vec}_{\mathbb{Z}_{n}})}g.
    \end{align}
    That is, $\Ver(V^{\mathbb{Z}_{n}}/V^{\textsl{Hg}_{n}})$ is a near-group fusion category of the form $\mathbb{Z}_{n}\times\mathbb{Z}_{n}+n^2$.
\end{claim}

\begin{proof}
    The structure of the hypergroup which acts on $V^{\mathbb{Z}_{n}}$ follows by combining the discussion around Equation \eqref{eqn:consecutiveembeddings} with \cite[Lemma 4.26]{Bischoff:2016jmy}, as we did for the special case $n=3$ previously. By Claim \ref{cl:cla:hypergroupgrading}, the hypergroup grades the Verlinde category, 
    \begin{align}\label{eqn:haggrading}
        \Ver(V^{\mathbb{Z}_n}/V^{\textsl{Hg}_n}) = \mathcal{B}(V^{\mathbb{Z}_n})\oplus \mathcal{B}(V^{\mathbb{Z}_n})_a,
    \end{align}
    with $\mathcal{B}(V^{\mathbb{Z}_n})= Z(\textsl{Vec}_{\mathbb{Z}_n})$.

    In the special case that $n=3$, we constructed the Verlinde category by direct calculation. Let us provide a slicker way of computing $\Ver(V^{\mathbb{Z}_{n}}/V^{\textsl{Hg}_{n}})$ for general $n$. We use the fact (cf.\ Section \ref{subsec:topman}) that $\Ver(V^{\mathbb{Z}_n}/V^{\textsl{Hg}_n})$ can be obtained as the $\mathbb{Z}_n$-equivariantization of $\textsl{Hg}_n$, where the generator of $\mathbb{Z}_n$ acts by conjugation by $\alpha$, i.e.\ $X\mapsto \alpha \otimes X \otimes \alpha^{-1}$. 

    Under this action of $\mathbb{Z}_n$, the invertible objects of $\textsl{Hg}_n$ are invariant, while the $\rho_i$ are all mapped into each other. Therefore, in the equivariantization, the invertibles each split into $n$ lines, 
    \begin{align}
        (\alpha^i,j), \ \ \ \ i,j=0,\dots,n-1,
    \end{align}
    where $j$ indexes the different choices of equivariant structures. On the other hand, the $\alpha^i\rho$ all combine into a single object 
    \begin{align}
        \mathcal{L}\equiv \left(\bigoplus_{k=0}^{n-1}\alpha^k\rho,u \right).
    \end{align}
    In particular, there are $n^2+1$ simple lines in $\Ver(V^{\mathbb{Z}_n}/V^{\textsl{Hg}_n})$. 
    
    The quantum dimensions follow from the fact that the induction functor Equation \eqref{eqn:inductionGequiv}, 
    \begin{align}
        I:\Ver(V^{\mathbb{Z}_n}/V^{\textsl{Hg}_n})\to \Ver(V/V^{\textsl{Hg}_n})\cong \textsl{Hg}_n,
    \end{align}
    which sends $(\alpha^i,j)\mapsto \alpha^i$ and $\mathcal{L}\mapsto \bigoplus_{k=0}^{n-1}\alpha^k\rho$, preserves quantum dimensions. In particular, we can conclude that the lines $(\alpha^i,j)$ all have quantum dimension $1$ and hence are invertible, while the line $\mathcal{L}$ has quantum dimension
    \begin{align}
       \dim(\mathcal{L})= \sum_{k=0}^{n-1} \dim(\alpha^k\rho) = \frac12 (n^2+n\sqrt{4+n^2}).
    \end{align}
At the same time, we know from Equation \eqref{eqn:haggrading} that $\Ver(V^{\mathbb{Z}_n}/V^{\textsl{Hg}_n})$ contains $Z(\textsl{Vec}_{\mathbb{Z}_n})$ as a fusion subcategory. All the lines in $Z(\textsl{Vec}_{\mathbb{Z}_n})$ are invertible and have $\mathbb{Z}_n\times \mathbb{Z}_n$ fusion rules, so the lines $(\alpha^i,j)$ must correspond to this subcategory.
The grading by the hypergroup then forces the fusion rules $(\alpha^i,j)\otimes \mathcal{L} \cong \mathcal{L}\otimes (\alpha^i,j)\cong \mathcal{L}$.

All that is left is to compute the fusion rule of $\mathcal{L}$ with itself. We can parametrize its general form as 
\begin{align}
    \mathcal{L}^2 \cong m_{\mathcal{L}}\cdot \mathcal{L} \oplus \bigoplus_{g\in Z(\textsl{Vec}_{\mathbb{Z}_n})}m_g\cdot g.
\end{align}
The fusion coefficient $m_{\mathcal{L}}$ follows from the hypergroup grading. Indeed, we have 
\begin{align}
    \frac{w-1}{w} = \frac{\dim(\pi_a(\mathcal{L}^2))}{\dim(\mathcal{L}^2)} = \frac{m_{\mathcal{L}}}{\dim(\mathcal{L})} \implies m_{\mathcal{L}}=n^2.
\end{align}
We also learn from the Frobenius reciprocity equation $N_{ab}^c=N_{c\bar b}^a$ (cf.\ Equation \eqref{eqn:Frobrep2}), using the self-duality of $\mathcal{L}$, that
\begin{align}
    m_g = N_{\mathcal{L}\mathcal{L}}^g = N_{g\mathcal{L}}^{\mathcal{L}}=1.
\end{align}
This completes the calculation of $\Ver(V^{\mathbb{Z}_n}/V^{\textsl{Hg}_n})$.
\end{proof}

\subsection{Boundaries from folding: \texorpdfstring{$(G_2)_1$}{(G2)1} WZW model}\label{subsec:G21}

In this section, we will illustrate how boundary conditions of the full $G_{2,1}$ WZW model  can be constructed by folding along topological line operators of its chiral algebra, $\widehat{\mathfrak{g}}_{2,1}$. In the process, we will be led to study a hypergroup with ``non-integerizable'' structure constants. More specifically, we will establish the following.

\begin{claim}{}{}
 (a) The category $\Ver\left(\widehat{\mathfrak{g}}_{2,1}/\widehat{\mathfrak{su}}(2)_{28}\right)$ of topological line operators supported on the $\widehat{\mathfrak{g}}_{2,1}$ chiral algebra which preserve its $\widehat{\mathfrak{su}}(2)_{28}$ subalgebra is the $E_8$ fusion category $\mathcal{C}_{E_8}$, whose fusion rules are recorded in Table \ref{tab:E8}. \\\vspace{-.15in}
 
 (b) The effective hypergroup which acts faithfully on $\widehat{\mathfrak{g}}_{2,1}$ has rank-4 and its noninvertible multiplication rule is recorded in Table \ref{tab:E8hyp} (see also \cite{Rie22}). \\\vspace{-.15in}
 
(c) There is a one-to-one correspondence between topological lines in $\mathcal{C}_{E_8}$ and boundary conditions of the full $G_{2,1}$ WZW model which preserve the $\widehat{\mathfrak{su}}(2)_{28}$ chiral subalgebra. We give a recipe for calculating the annulus partition functions in the presence of these boundary conditions. 

\end{claim}

Our starting point is the well-known conformal embedding 
\begin{align}\label{eqn:su(2)28confembed}
    \widehat{\mathfrak{su}}(2)_{28}\subset \widehat{\mathfrak{g}}_{2,1}.
\end{align}
A closely related fact is that, in the ADE classification of CFTs built on top of the $\widehat{\mathfrak{su}}(2)_k$ chiral algebras \cite{Cappelli:1987xt}, the exceptional theory corresponding to $E_8$ is actually the $G_{2,1}$ WZW model. We will study the symmetries of the $\widehat{\mathfrak{g}}_{2,1}$ chiral algebra induced by the conformal embedding in Equation \eqref{eqn:su(2)28confembed}, revisiting some of the calculations carried out in \cite{Rie22} with our more physical perspective in hand.

Let $A$ be the condensable algebra in $\mathcal{B}(\widehat{\mathfrak{su}}(2)_{28})$ (i.e.\ in $SU(2)_{28}$ Chern-Simons theory) which mediates the conformal extension in Equation \eqref{eqn:su(2)28confembed}. It is known that the category of $A$-modules in $\mathcal{B}(\widehat{\mathfrak{su}}(2)_{28})$ is given by the commutative $E_8$-fusion category \cite{izumi1991application} (see also \cite{edie2020classifying}), 
\begin{align}
    \Ver\left(\widehat{\mathfrak{g}}_{2,1}/\widehat{\mathfrak{su}}(2)_{28}\right):= \mathcal{B}(\widehat{\mathfrak{su}}(2)_{28})_A\cong \mathcal{C}_{E_8},
\end{align}
whose simple objects we parametrize as 
\begin{align}
    \mathrm{Irr}(\mathcal{C}_{E_8})=\{\mathds{1},W,\mathds{1}_a,W_a,\mathds{1}_b,W_b,\mathds{1}_c,W_c\},
\end{align}
and whose fusion rules are given in Table \ref{tab:E8}.
Our interpretation is that $\mathcal{C}_{E_8}$ is a category of topological line operators supported on the $\widehat{\mathfrak{g}}_{2,1}$-boundary of $G_{2,1}$ Chern-Simons TQFT.

Let us develop the alternative perspective on $\Ver\left(\widehat{\mathfrak{g}}_{2,1}/\widehat{\mathfrak{su}}(2)_{28}\right)\cong \mathcal{C}_{E_8}$ afforded by Section \ref{subsec:hypergroup} and Section \ref{subsec:surfaces}. That is, we wish to view $\mathcal{C}_{E_8}$ instead as a category of line operators which can appear at the boundary of the surface $\mathcal{S} = \mathcal{I}_A^\ast\otimes \mathcal{I}_A$.\footnote{Similar calculations appear in \cite{Mulevicius:2020tgg} for the conformal embedding of $\widehat{\mathfrak{su}}(2)_{10}\subset \widehat{\mathfrak{spin}}(5)_1$, see also \cite{Kirillov:2001ti} for prior related work.} Without any work, we can immediately conclude that the surface in question must be a direct sum of four copies of the identity surface, 
\begin{align}
    \mathcal{S} \cong \mathcal{S}_{1} \oplus \mathcal{S}_a\oplus \mathcal{S}_b\oplus \mathcal{S}_c, \ \ \ \ \ \mathcal{S}_k \cong  \mathds{1} \text{ for }k\in\{1,a,b,c\}.
\end{align}
Indeed, $G_{2,1}$ Chern-Simons theory (a.k.a.\ the Fibonacci TQFT) has no non-trivial surface operators, so every simple surface appearing in the decomposition of $\mathcal{S}$ must be the identity surface. On the other hand, we know that precisely four copies of the identity surface operator must appear because $\mathcal{C}_{E_8}$ has rank $8$, while each copy of the identity surface supports $2$ lines at its boundary (corresponding to the fact that the Fibonacci TQFT has 2 anyons).

We recall that the category $\mathcal{B}(V)$ on general grounds appears as a subcategory of $\Ver\left(V/W\right)$ for any conformal embedding $W\subset V$. In $\mathcal{C}_{E_8}$, we can readily identify $\{\mathds{1},W\}$ as precisely the lines which generate the bulk $\mathcal{B}(\widehat{\mathfrak{g}}_{2,1})$ subcategory. From the fusion rule 
\begin{align}
    W \otimes \mathds{1}_k \cong W_k, \ \ \ \ W\otimes W_k \cong \mathds{1}_k\oplus W_k, \ \ \ \ \ k\in \{1,a,b,c\},
\end{align}
we see that each pair $\{\mathds{1}_k,W_k\}$ transforms in the regular module category with respect to $\mathcal{B}(\widehat{\mathfrak{g}}_{2,1}).$ Thus, we can identify $\{\underline{\mathds{1}_k},\underline{W_k}\}$ as precisely the lines which should be thought of as appearing on the boundary of the surface $\mathcal{S}_k$, where we recall that we use the notation $\underline{X}$ to denote the image of $X$ under the pinching functor described in Figure \ref{fig:pinch}.

We can think of the $T$-junction of Figure \ref{fig:Tjunction} as decomposing as
\begin{align}
    T\cong \bigoplus_{k,k',k''\in \{1,a,b,c\}} {_{k''}}t_{kk'}, \ \ \ \ \ \ \ \ {_{k''}}t_{kk'}\in\mathcal{B}(\widehat{\mathfrak{g}}_{2,1}),
\end{align}
where ${_{k''}}t_{kk'}$ is the bulk topological line which appears at the junction of the three surfaces $\mathcal{S}_k$, $\mathcal{S}_{k'}$, $\mathcal{S}_{k''}$. We can determine the ${_{k''}}t_{kk'}$ by imposing that the $T$-junction reproduces the fusion rules of the $\mathcal{C}_{E_8}$ category under the correspondence of Figure \ref{fig:Tfusion}. Indeed, we require that the fusion rule on boundary lines of the $\mathcal{S}$-surface defined by the $T$-junction,
\begin{align}
    \underline{X_k}\otimes_T\underline{Y_{k'}}:=\bigoplus_{k''\in \{1,a,b,c\}} \underline{(X\otimes {_{k''}}t_{kk'}\otimes Y )_{k''}}, \ \ \ \ \ X,Y\in \{\mathds{1},W\},
\end{align}
agrees with the fusion rule on $\mathcal{C}_{E_8}$ after using the pinching functor of Figure \ref{fig:pinch}. This is most straight-forwardly achieved by taking $X,Y=\mathds{1}$, in which case we learn that ${_{k''}}t_{kk'}$ is the unique collection of lines making the following equation true:
\begin{align}\label{Tdef}
    \mathds{1}_k \otimes \mathds{1}_{k'} \cong \bigoplus_{k''}({_{k''}}t_{kk'})_{k''}.
\end{align}
For example, we read off from the fusion rule $\mathds{1}_a\otimes \mathds{1}_b\cong \mathds{1}_a\oplus W_c$ that 
\begin{align}\label{eqn:exTjunc}
{_{a}}t_{ab}=\mathds{1}, \ \ \ \  {_{c}}t_{ab}=W, \ \ \ \ {_{1}}t_{ab}={_{b}}t_{ab}=0.
\end{align}
The rest of the components of $T$ can be read off straightforwardly from Table \ref{tab:E8}.

We can also determine the $L$ junction defined implicitly in Figure \ref{fig:pinch}. Again, we can decompose this as 
\begin{align}
    L\cong \bigoplus_{k\in \{1,a,b,c\}} L_k
\end{align}
and determine the components $L_k$.  Each $L_k$ is a choice of topological line operator on the interface $\mathcal{I}_A$ (i.e.\ in the $\mathcal{C}_{E_8}$ fusion category) on which the bulk surface $\mathcal{S}_k$ terminates, which we can fix by imposing the equation in Figure \ref{fig:Tmove}, which says that 
\begin{align}
    L_k\otimes L_{k'} \cong \bigoplus_{k''\in\{1,a,b,c\}} {_{k''}}t_{kk'}\otimes L_{k''}.
\end{align}
Reinterpreting Equation \eqref{Tdef} slightly, 
\begin{align}
    \mathds{1}_k\otimes \mathds{1}_{k'}\cong  \bigoplus_{k''} ({_{k''}}t_{kk'})_{k''}\cong \bigoplus_{k''} {_{k''}}t_{kk'}\otimes \mathds{1}_{k''}
\end{align}
we learn that we can take $L_k=\mathds{1}_k$.

We can compute the rescaled hypergroup structure constants $\mathbf{P}_{ij}^k$ from Figure \ref{fig:domefusion}. Indeed, Figure \ref{fig:domefusion} simply identifies them, in the special case that all the $\mathcal{S}_k$ are identity surfaces, with the quantum dimensions of the $T$-junction lines, 
\begin{align}
    \mathbf{P}_{ij}^k = \dim({_{k}}t_{ij}).
\end{align}
For example, Equation \eqref{eqn:exTjunc} allows us to conclude that 
\begin{align}
    \mathbf{P}_{ab}^a=1, \ \ \ \ \  \mathbf{P}_{ab}^c=\frac{1+\sqrt{5}}{2}, \ \ \  \ \ \mathbf{P}_{ab}^1=\mathbf{P}_{ab}^b=0.
\end{align}
The rest of the hypergroup structure constants are reported in Table \ref{tab:E8hyp}. One can straightforwardly verify that the hypergroup basis can be rescaled (cf.\ Equations \eqref{eqn:rescaledhypergroup} and \eqref{eqn:rescaledhypergrouprelation}) so as to obey the stochastic normalization $\sum_{k}P_{ij}^k = 1$ and reproduce the structure constants computed in Table 1 of \cite{Rie22}. 

We note that this hypergroup has ``non-integerizable'' structure constants. It suffices to show this for the subhypergroup spanned by $\{\mathbf{r}_1,\mathbf{r}_b\}$.  We ask if there is a rescaling $r'_b=\alpha \mathbf{r}_b$ for some $\alpha$ such that the structure constants 
\begin{align}
    r'_b\star r'_b = \alpha^2\cdot \mathbf{r}_1+\alpha(1+\varphi)\cdot r'_b
\end{align}
can be made integers. It is straightforward to see that this would force $(1+\varphi)^2 = \frac72+\frac32 \sqrt{5}$ to be a rational number, which is clearly false.

Finally, we note that, by the folding trick of Figure \ref{fig:lines2boundaries}, every line in the $\mathcal{C}_{E_8}$ fusion category can be thought of as furnishing a conformal boundary condition of the full $G_{2,1}$ Wess-Zumino-Witten CFT. The boundaries obtained from the lines $\mathds{1}$ and $W$ are the standard Cardy boundary conditions of the theory. The conformal embedding Equation \eqref{eqn:su(2)28confembed} induces 6 additional topological lines which can be folded to yield 6 boundary conditions which break the full $\widehat{\mathfrak{g}}_{2,1}$ chiral algebra down to its $\widehat{\mathfrak{su}}(2)_{28}$ subalgebra. Furthermore, Claim \ref{cl:claim:annulus} relates annulus partition functions in the presence of these boundaries to twisted partition functions of the chiral algebra (see Figure \ref{fig:annulus}), which in turn can be expressed in terms of characters of $\widehat{\mathfrak{su}}(2)_{28}$. 

Let us give an example. We consider the boundary conditions $\partial\mathds{1}_a$ and $\partial \mathds{1}_b$ obtained by folding the topological lines $\mathds{1}_a$ and $\mathds{1}_b$ on the $\widehat{\mathfrak{g}}_{2,1}$ chiral algebra. From Table \ref{tab:E8} we see that $\mathds{1}_a\otimes \mathds{1}_b\cong \mathds{1}_a\oplus W_c$, so Claim \ref{cl:claim:annulus} says that the annulus partition function of the full $G_{2,1}$ WZW model with $\partial \mathds{1}_a$ imposed at one end of the interval and $\partial\mathds{1}_b$ imposed at the other is given by
\begin{align}
    {^{S^1\times I}}Z_{\partial\mathds{1}_a,\partial \mathds{1}_b}(\delta) = Z_{\mathds{1}_a}(\tau)+Z_{W_c}(\tau),
\end{align}
with $\tau=i\frac{\delta}{2}$ and $Z_X(\tau)$ the $X$-twisted partition function of the $\widehat{\mathfrak{g}}_{2,1}$ chiral algebra. Equation \eqref{eqn:twistedpfB} then allows us to further reduce this to an expression involving characters of $\widehat{\mathfrak{su}}(2)_{28}$. Indeed, from \cite{Kirillov:2001ti}, we can extract the restriction functor $R:\mathcal{C}_{E_8}\to \mathcal{B}(\widehat{\mathfrak{su}}(2)_{28})$, which is given by
\begin{align}
\begin{split}
R(\mathds{1})&=0\oplus 10\oplus 18 \oplus 28 \\
R(W)&=6\oplus 12 \oplus 16\oplus 22 \\
R(\mathds{1}_a) &= 1\oplus 9 \oplus 11 \oplus 17 \oplus 19 \oplus 27 \\
R(W_a) &= 5\oplus 7\oplus 11\oplus 13 \oplus 15\oplus 17\oplus 21\oplus 23 \\
R(\mathds{1}_b) &= 2\oplus 8 \oplus 10 \oplus 12 \oplus 16\oplus 18 \oplus 20 \oplus 26 \\
R(W_b) &= 4\oplus 6\oplus 8\oplus 10 \oplus 12 \oplus 2\cdot 14 \oplus 16\oplus 18 \oplus 20 \oplus 22\oplus 24 \\
R(\mathds{1}_c) &=5\oplus 9\oplus 13 \oplus 15\oplus 19\oplus 23 \\
R(W_c) &= 3\oplus 7\oplus 9\oplus 11\oplus 13\oplus 15\oplus 17\oplus 19\oplus 21\oplus 25.
\end{split}
\end{align}
Plugging this in, we then find that 
\begin{align}
\begin{split}
   & {^{S^1\times I}}Z_{\partial\mathds{1}_a,\partial \mathds{1}_b}(\delta)= \\
    &\hspace{.3in}(\chi_1+\chi_3+\chi_7+2\chi_9+2\chi_{11}+\chi_{13}+\chi_{15}+2\chi_{17}+2\chi_{19}+\chi_{21}+\chi_{25}+\chi_{27})(\tau),
\end{split}
\end{align}
where $\chi_j(\tau)$ is a character of $\widehat{\mathfrak{su}}(2)_{28}$ with $j=0,\dots,28$, 
see e.g.\ \cite{Eberhardt:2019WZW} for explicit expressions.

\subsection{An absolute example: symmetries of \texorpdfstring{$SU(3)_1$}{SU(3)1} WZW model from gluing}\label{subsec:SU(3)1}

In this subsection, we glue together certain symmetries of the $\widehat{\mathfrak{su}}(3)_1$ chiral algebra (those that preserve its $\widehat{\mathfrak{su}}(2)_4$ chiral subalgebra) to produce topological line defects of the full $SU(3)_1$ WZW model beyond its Verlinde lines. 

\begin{claim}{}{}
The full $SU(3)_1$ WZW model admits a fusion category $\mathcal{C}$ of topological line operators,
\begin{align}
    \mathrm{Irr}(\mathcal{C}) = \{1,r,r^2,s_0,s_1,s_2,\mathcal{D}_L,\mathcal{D}_R\},
\end{align}
which preserve its left- and right-moving $\widehat{\mathfrak{su}}(2)_4$ chiral subalgebras.
The lines $1$, $r$, $r^2$, $s_0$, $s_1$, $s_2$ have the fusion rules of the symmetric group $S_3$, and the remaining lines satisfy
\begin{align}
\begin{split}
    &\hspace{.7in}\mathcal{D}_{L,R}^2\cong 1\oplus r\oplus r^2, \ \ \ \mathcal{D}_L\otimes \mathcal{D}_R\cong s_0\oplus s_1\oplus s_2 \\
    &r^i\otimes \mathcal{D}_{L,R}\cong \mathcal{D}_{L,R}\otimes r^i \cong\mathcal{D}_{L,R}, \ \ \ \ s_i\otimes \mathcal{D}_{L,R} \cong \mathcal{D}_{L,R} \otimes s_i \cong \mathcal{D}_{R,L}.
\end{split}
\end{align}
\end{claim}

Let us start with the famous conformal embedding 
\begin{align}
    \widehat{\mathfrak{su}}(2)_4 \subset \widehat{\mathfrak{su}}(3)_1.
\end{align}
This is a $\mathbb{Z}_2$ simple current extension, which means in particular that $SU(3)_1$ Chern-Simons theory can be obtained from $SU(2)_4$ Chern-Simons theory by gauging a $\mathbb{Z}_2^{(1)}$ one-form symmetry. Thus, we obtain a half-space gauging interface $\mathcal{I}$ between these two TQFTs. 

From the SymTFT picture of conformal embeddings, Figure \ref{fig:SymTFTconformalembedding}, we see that the category of topological line defects on $\mathcal{I}$ can be interpreted as a category $\Ver(\widehat{\mathfrak{su}}(3)_1/\widehat{\mathfrak{su}}(2)_4)$ of topological line defects supported on the $\widehat{\mathfrak{su}}(3)_1$ chiral algebra boundary of $SU(3)_1$ Chern-Simons theory. It is known \cite{Huston:2022utd} that this is a Tambara-Yamagami category \cite{Tambara:1998vmj}, 
\begin{align}
   \Ver(\widehat{\mathfrak{su}}(3)_1/\widehat{\mathfrak{su}}(2)_4)\cong \mathrm{TY}_-(\mathbb{Z}_3),
\end{align}
whose simples we denote $\{1,\eta,\eta^2,\mathcal{D}\}$ and whose fusion rules are 
\begin{align}
    \eta^i\otimes \eta^j = \eta^{i+j}, \ \ \ \ \ \eta^i\otimes \mathcal{D} \cong \mathcal{D} \otimes \eta^i \cong \mathcal{D}, \ \ \ \ \ \mathcal{D}\otimes \mathcal{D}\cong 1\oplus \eta\oplus \eta^2.
\end{align}
The $-$ indicates that the Frobenius-Schur indicator of the duality line $\mathcal{D}$ is $-1$. 

\begin{figure}
    \begin{center}
        \input{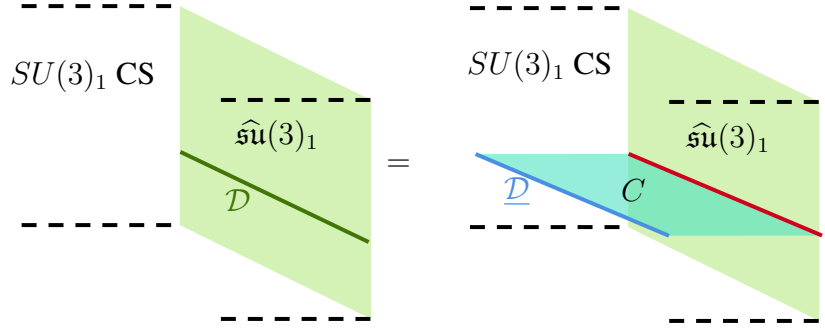}
        \caption{The duality line $\mathcal{D}$ supported on the $\widehat{\mathfrak{su}}(3)_1$ chiral algebra is trailed by a charge conjugation surface $C$ when it is dragged into the bulk.}\label{fig:SU(3)1CC}
    \end{center}
\end{figure}

The lines $\eta^i$ can be freely dragged into the bulk without leaving behind any topological surface, as they are simply bulk Wilson lines of $SU(3)_1$ Chern-Simons theory. (Equivalently, we can say that they are attached to the identity surface $\mathds{1}$.) On the other hand, if we pull $\mathcal{D}$ into the bulk, as in Figure \ref{fig:SU(3)1CC}, then we find that it is attached to a $\mathbb{Z}_2$ charge conjugation topological surface $C$ (cf.\ Figure \ref{fig:pinch}). Indeed, $SU(3)_1$ Chern-Simons theory has two topological surfaces, corresponding to the fact that there are two $\Vec_{\mathbb{Z}_3}$-module categories, and $\mathcal{D}$ transforms in the non-trivial $\Vec_{\mathbb{Z}_3}$ module category, so it leaves behind a non-trivial surface operator when it is dragged into the bulk. Putting these facts together, we find that the surface $\mathcal{S}=\mathcal{I}^\ast\otimes \mathcal{I}$ from Figure \ref{fig:pinch} is $\mathcal{S}=\mathds{1}\oplus C$.

\begin{figure}
    \begin{center}
        \input{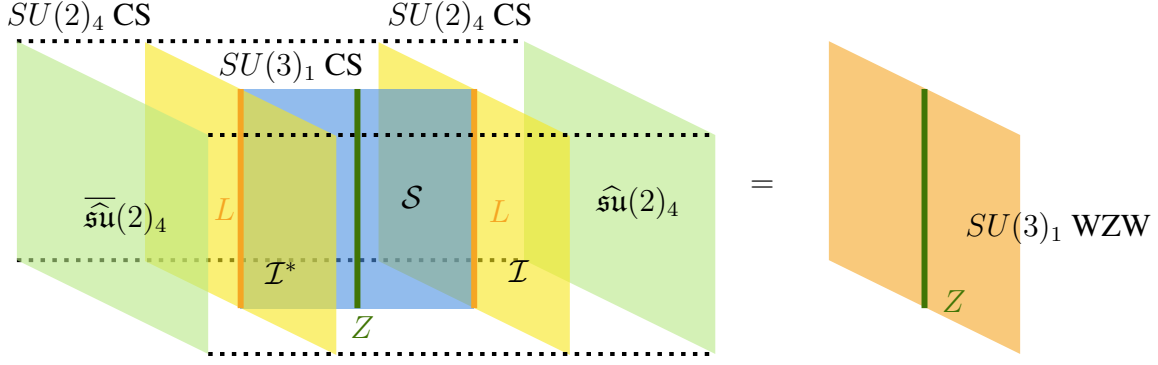}
        \caption{The topological lines $Z$ of the $SU(3)_1$ WZW model which preserve the left- and right-moving $\widehat{\mathfrak{su}}(2)_4$ subalgebras can be identified with topological lines on the surface $\mathcal{S}=\mathds{1}\oplus C$.}\label{fig:SU(3)1lines}
    \end{center}
\end{figure}

Now, Equation \eqref{eqn:descrip1} asserts that the category of topological line defects of the full $SU(3)_1$ WZW model which preserve its left- and right-moving $\widehat{\mathfrak{su}}(2)_4$ chiral subalgebras is given by the relative Deligne product,
\begin{align}\label{eqn:SU(3)1Ver}
    \mathcal{C}\equiv \Ver\left(SU(3)_1/\widehat{\mathfrak{su}}(2)_4\otimes \overline{\widehat{\mathfrak{su}}(2)}_4\right)\cong \mathrm{TY}_-(\mathbb{Z}_3)\boxtimes_{\Vec_{\mathbb{Z}_3}}\mathrm{TY}_-(\mathbb{Z}_3)^{\vee}.
\end{align}
Let us start by calculating how many simple objects there are in this category. We do this using the physical interpretation of the relative Deligne product afforded by Figure \ref{fig:IJsandwich}, which is reproduced in our special case in Figure \ref{fig:SU(3)1lines}. It says that the lines $Z$ in the category \eqref{eqn:SU(3)1Ver} can be identified with lines on the surface $\mathcal{S}=\mathds{1}\oplus C$. There are four sectors: lines from $\mathds{1}\to \mathds{1}$, lines from $C\to C$, line interfaces from $\mathds{1}\to C$, and line interfaces from $C\to\mathds{1}$. Clearly the lines from $\mathds{1}\to \mathds{1}$ are simply the $\mathbb{Z}_3$ worth of bulk lines, which we denote $1, r, r^2$. Similarly, by folding, lines on the surface $C$ are the same as lines which bound the surface $C\otimes C^\ast = \mathds{1}$, so there are again 3 simple lines, which we call $s_0,s_1,s_2$. Finally, there is a unique line $\mathcal{D}_L$ from $C\to 1$, and similarly a unique line $\mathcal{D}_R$ from $1\to C$. (This is because $C$, when viewed as a $\Vec_{\mathbb{Z}_3}$-module category, has rank-1.) In total, we learn that
\begin{align}
    \mathrm{Irr}(\mathcal{C})=\{1,r,r^2,s_0,s_1,s_2,\mathcal{D}_L,\mathcal{D}_R\}.
\end{align}

Let us determine the fusion rules. We do this using the fact that the balancing functor from Figure \ref{fig:B in Fun description},
\begin{align}
    B:\mathrm{TY}_-(\mathbb{Z}_3)\boxtimes \mathrm{TY}_-(\mathbb{Z}_3)^{\vee} \to \mathrm{TY}_-(\mathbb{Z}_3)\boxtimes_{\Vec_{\mathbb{Z}_3}} \mathrm{TY}_-(\mathbb{Z}_3)^{\vee} ,
\end{align}
is a tensor functor. It clearly maps 
\begin{align}
\begin{split}
    &B(\eta^i,\eta^j) = r^{i+j}, \ \ \ \ B(\eta^i,\mathcal{D})=\mathcal{D}_R, \ \ \ \ B(\mathcal{D},\eta^j)=\mathcal{D}_L \\
    &\hspace{1in}B(\mathcal{D},\mathcal{D}) = s_0\oplus s_1\oplus s_2.
\end{split}
\end{align}
The first three of these identifications are clear. For the last one, note that $B(\mathcal{D},\mathcal{D})$ must be a line of quantum dimension $\dim(\mathcal{D})^2=3$ on the charge conjugation surface $C$. Furthermore, because $B$ is a tensor functor, we have 
\begin{align}
    B(\eta^i,1)\otimes B(\mathcal{D},\mathcal{D}) \cong B(\eta^i\otimes\mathcal{D},\mathcal{D})\cong B(\mathcal{D},\mathcal{D}),
\end{align}
so it must be invariant under fusion by $r$, which generates the bulk $\mathbb{Z}_3$. In particular, this leaves $B(\mathcal{D},\mathcal{D})\cong s_0\oplus s_1\oplus s_2$ as the only option. 

Now, using the fact that $B$ is a tensor functor, we derive the fusion rules. The following are a few simple examples:
\begin{align}
\begin{split}
    \mathcal{D}_L^2\cong B(\mathcal{D},1)^2\cong B(1\oplus \eta\oplus \eta^2,1) \cong 1\oplus r\oplus r^2 \\
    \mathcal{D}_L\otimes \mathcal{D}_R \cong B(\mathcal{D},1)\otimes B(1,\mathcal{D})\cong B(\mathcal{D},\mathcal{D})\cong s_0\oplus s_1\oplus s_2, \\
    r^i\otimes \mathcal{D}_L \cong B(\eta^i,1)\otimes B(\mathcal{D},1)\cong B(\eta^i\otimes \mathcal{D},1) \cong B(\mathcal{D},1) \cong \mathcal{D}_L,
    \end{split}
\end{align}
with similar derivations when one replaces $\mathcal{D}_L$ with $\mathcal{D}_R$. The fusions involving the $s_i$ and the duality lines are only slightly more indirect. For example, one calculates that
\begin{align}
\begin{split}
&(s_0\oplus s_1\oplus s_2)\otimes \mathcal{D}_L \cong B(\mathcal{D},\mathcal{D})\otimes B(\mathcal{D},1) \cong \bigoplus_{i=0}^2B(\eta^i,\mathcal{D})\cong 3\cdot \mathcal{D}_R \\
& \hspace{1in}\implies s_i\otimes \mathcal{D}_L\cong \mathcal{D}_R.
\end{split}
\end{align}

The last fusion rules one must fix are those between the $r_i$ and the $s_j$. These lines are all invertible and  generate a group of order $6$, i.e.\ either $\mathbb{Z}_6$ or $S_3$. A quick way to distinguish between these two possibilities is to note that $\mathcal{C}$ must have non-commutative fusion rules, and hence in fact $r_i$ and $s_j$ must generate $S_3$. Indeed, this follows from Equation \eqref{eqn:matrixalgdecomp}. Thinking of the $SU(3)_1$ WZW model as a non-diagonal CFT with respect to the chiral subalgebra $W=\widehat{\mathfrak{su}}(2)_4$,  it is known \cite{Cappelli:1987xt} that the modular invariant pairing matrix $M_{\mu\nu}$ (cf.\ Equation \eqref{eqn:fullCFTHilbertspace}) satisfies $M_{11}=2$, where we are labeling representations of $\widehat{\mathfrak{su}}(2)_4$ by their spin $\mu=0,\sfrac12,\dots,2$. Thus,  Equation \eqref{eqn:matrixalgdecomp} says that the complexified Grothendieck ring $\mathcal{K}_0(\mathcal{C})\otimes \mathbb{C}$ contains the noncommutative algebra $ \mathrm{Mat}_{\mathbb{C}}(2)$ of $2\times 2$ matrices as a subalgebra, and hence $\mathcal{C}$ must have non-commutative fusion rules. 

\subsection{An infinite example: all symmetries of the \texorpdfstring{$\widehat{\mathfrak{u}}(1)$}{u(1)} Kac-Moody algebra}\label{subsec:Heisenberg}

Most of the machinery and examples discussed so far have been about rational chiral algebras and their finite symmetries. We conclude by determining the full infinite category of topological lines of one of the simplest irrational theories: the $\widehat{\mathfrak{u}}(1)$ Kac-Moody algebra, a.k.a.\ the Heisenberg VOA, which can be thought of as a gapless chiral boundary condition of $\mathbb{R}$ Chern-Simons theory. 

The mathematical underpinnings of infinite symmetries are still the subject of ongoing research (see e.g.\ \cite{Freed:2009qp,Marin-Salvador:2025stc,Delmastro:2025ksn,Stockall:2025ngz,Jia:2025vrj,Jia:2026vcr,Jia:2026tsl,Jia:2026bbo} for recent papers), and the results we present here are obtained by formally applying various theorems (which, strictly speaking, only hold in the rational/finite setting) outside their naive regime of validity. Nevertheless, we feel that the basic picture we paint here is correct and should be amenable to eventual mathematical rigor.

Our approach to the symmetries of the Heisenberg VOA relies on two observations. 
\begin{enumerate}[label=\arabic*)]
    \item The full category $\mathrm{Sym}^\dagger(\widehat{\mathfrak{su}}(2)_1)$ of unitary topological line operators of the $\widehat{\mathfrak{su}}(2)_1$ chiral algebra has elementary objects in one-to-one correspondence with elements of $SU(2)$, 
    \begin{align}
        \mathrm{Irr}(\mathrm{Sym}^\dagger(\widehat{\mathfrak{su}}(2)_1)) = SU(2).
    \end{align}
    \item The $\widehat{\mathfrak{u}}(1)$ Kac-Moody algebra can be obtained from $\widehat{\mathfrak{su}}(2)_1$ by passing to the fixed points of the $SO(2)$ Cartan of its unitary automorphism group $\mathrm{Aut}^\dagger(\widehat{\mathfrak{su}}(2)_1)\cong SO(3)$, 
    \begin{align}
        \widehat{\mathfrak{u}}(1) \cong \widehat{\mathfrak{su}}(2)_1^{SO(2)}.
    \end{align}
    In other words, $\widehat{\mathfrak{u}}(1)$ can be obtained by starting with $SU(2)_1$ Chern-Simons theory in the presence of its WZW boundary condition and performing a ``flat'' $SO(2)$ gauging of the bulk/boundary system.
\end{enumerate}

Regarding the first observation, $\widehat{\mathfrak{su}}(2)_1$ of course possesses \emph{at least} an $SU(2)$ worth of lines obtained by exponentiating its spin-1 currents. It will be argued for more fully elsewhere that these in fact are all of its topological lines (see \cite{Fuchs:2007tx,Bischoff:2022fxf,Moller:2024xtt} for similar statements). For now, we note that symmetry/subalgebra duality, Claim \ref{cl:claim:sym/subduality}, already gives reasonable intuition for why this should be true: under the correspondence between symmetry categories and conformal subalgebras, the full category of unitary topological line operators is expected to correspond to the Virasoro subalgebra, and it is not difficult to see that
\begin{align}
    \widehat{\mathfrak{su}}(2)_1^{SU(2)}\cong \mathrm{Vir}_{c=1}.
\end{align}
That is, the subalgebra of invariant operators in $\widehat{\mathfrak{su}}(2)_1$ with respect to $SU(2)$ is precisely the Virasoro subalgebra, a property which is expected only to be true of the maximal category of unitary lines.

The importance of the second observation is that it allows us to use the general yoga of Section \ref{subsec:topman} to determine how the lines of $\widehat{\mathfrak{su}}(2)_1$ mutate into lines of $\widehat{\mathfrak{u}}(1)$. Indeed, formally, we expect to be able to obtain the lines of $\widehat{\mathfrak{u}}(1)$ from those of $\widehat{\mathfrak{su}}(2)_1$ via $SO(2)$ equivariantization, 
\begin{align}\label{eqn:U(1)equiv}
    \mathrm{Sym}^\dagger(\widehat{\mathfrak{u}}(1))\cong \mathrm{Sym}^\dagger(\widehat{\mathfrak{su}}(2)_1)^{SO(2)},
\end{align}
where the $SO(2)\subset SO(3)$ acts via conjugation on $SU(2)$. In particular, we can formally apply \cite[Theorem 4.1]{burciu2013fusion}, which explains how to compute the $G$-equivariantization of a pointed fusion category. This leads to the following.

\begin{claim}{}{}
    The full category of unitary topological line operators of the Abelian current algebra $\widehat{\mathfrak{u}}(1)$, which lives at the boundary of $\mathbb{R}$ Chern-Simons theory, has the following collection of simple objects, 
    \begin{align}
        \mathrm{Irr}(\mathrm{Sym}^\dagger(\widehat{\mathfrak{u}}(1)) = \{L_a\mid a\in\mathbb{C}, ~ |a|<1\}\cup \{M_\lambda \mid \lambda \in \mathbb{R}\}
    \end{align}
    where $\dim(M_\lambda)=1$ and $\dim(L_a)=\infty$. These lines obey the fusion rules 
    \begin{align}\label{eqn:diskfusion}
    \begin{split}
       & M_\lambda\otimes M_{\lambda'}\cong M_{\lambda+\lambda'}, \ \ \ M_\lambda \otimes L_a \cong L_a\otimes M_\lambda \cong L_{ae^{ 2\pi i \lambda}}, \\
        &\hspace{.6in}  L_a\otimes L_b \cong \int_{\theta\in[0,2\pi)}^\oplus[d\theta]~ L_{ab-r(a,b)e^{i\theta}},
    \end{split}
   \end{align}
    where $r(a,b)=\sqrt{(1-|a|^2)(1-|b|^2)}$, and we define 
    \begin{align}
        L_{e^{2\pi i \lambda}}\cong \bigoplus_{\lambda' \in \lambda +\mathbb{Z}}M_{\lambda'}.
    \end{align}
    The duals are given by 
    \begin{align}
        L_a^\ast = L_{\bar a}, \ \ \ \ \ \ M_{\lambda}^\ast = M_{-\lambda}.
    \end{align}The lines $M_\lambda$ can be pulled into the bulk and identified with the Wilson lines of $\mathbb{R}$ Chern-Simons theory. The conformal dimension of the lightest boundary local operator on which $M_\lambda$ terminates is $h=\lambda^2$. The lines $L_a$ are trapped on the boundary and cannot be pulled into the bulk. 

    \hspace{.3in}The effective hypergroup is 
    \begin{align}
        SO(3)\sslash SO(2) \cong [-1,1]
    \end{align}
    with multiplication described by Equation \eqref{eqn:continuousmult}. The category $\mathrm{Sym}^\dagger(\widehat{\mathfrak{u}}(1))$ is graded by this hypergroup, and  the simple lines in the graded-component $\mathcal{B}(\widehat{\mathfrak{u}}(1))_t$ are given by
    \begin{align}
        \mathrm{Irr}(\mathcal{B}(\widehat{\mathfrak{u}}(1))_t) = \begin{cases} \{L_a \mid |a|^2=\frac{t+1}{2} \}, & \text{ if }t<1 \\
        \{M_\lambda \mid \lambda\in \mathbb{R}\} , & \text{ if }t=1.
        \end{cases}
    \end{align}
\end{claim}

\noindent We are unsure whether the $\int[d\theta]$ in Equation \eqref{eqn:diskfusion} should be thought of as a direct integral using the Haar measure on $U(1)$ (and if so, with what normalization), or as an infinite direct sum using the discrete topology. We leave this interesting question to the future.

Let us start by computing the topological lines of the $\widehat{\mathfrak{u}}(1)$ chiral algebra. Applying the prescription of \cite[Theorem 4.1]{burciu2013fusion} to Equation \eqref{eqn:U(1)equiv}, we learn that the topological lines should be labeled by pairs $([g],\rho)$, where $[g]$ is an orbit of the $SO(2)$ action on $SU(2)$, and $\rho$ is a representation of the stabilizer $G_{[g]}$.\footnote{In general we would need to consider projective representations of $G_{[g]}$, but since $G_{[g]}\subset SO(2)$ in this example, all representations of $G_{[g]}$ are in fact linear.} 

To state what the orbits are, let 
\begin{align}
    g(a,\psi) = \left(\begin{array}{cc} a & \sqrt{1-|a|^2}e^{i\psi} \\ -\sqrt{1-|a|^2}e^{-i\psi} & \bar a\end{array}\right)
\end{align}
be an arbitrary element of $SU(2)$. It is clear that the action of $e^{i\zeta} \in SO(2)$ on an element of $SU(2)$ is 
\begin{align}
    g(a,\psi) \xmapsto{e^{i\zeta}} g(a,\psi+\zeta).
\end{align}
It follows that there are two kinds of $SO(2)$ orbits:
\begin{align}
    [g(a,\psi)]= \begin{cases}
        \{g(a,\chi)\mid \chi\in[0,2\pi)\}, & \text{ if }|a|<1 , \\
        \{g(a,0)\}, & \text{ if }|a|=1,
    \end{cases}
\end{align}
and the corresponding stabilizer groups are 
\begin{align}
    G_{[g(a,\psi)]} = \begin{cases}
        \{1\}, & \text{ if }|a|<1 , \\
        SO(2), & \text{ if }|a|=1.
    \end{cases}
\end{align}
Thus, elements $g(a,\psi)$ of $SU(2)$ with a fixed $a$ all merge into a single line in the $SO(2)$ equivariantization if $|a|<1$, but split into $\mathbb{Z}$ many lines if $|a|=1$. In particular, we have the following spectrum of simple lines, 
\begin{align}
    \mathrm{Irr}(\mathrm{Sym}^\dagger(\widehat{\mathfrak{su}}(2)_1)^{SO(2)}) = \{ L_a \mid a\in \mathbb{C}, |a|<1\} \cup \{X_{a,n}\mid |a|=1, n\in\mathbb{Z}\},
\end{align}
where, in the parametrization of simple objects as pairs $([g],\rho)$ with $\rho$ an irreducible representation of $G_{[g]}$, we have
\begin{align}\label{eqn:pairsgrho}
\begin{split}
    L_a&\sim ([g(a,0)],\rho=1),  \ \ \ \ \ a\in \mathbb{C}, \ |a|<1,\\
    X_{a,n} &\sim ([g(a,0)],\rho=n), \ \ \ \ \ a\in\mathbb{C}, \ |a|=1, \ n\in\mathbb{Z},
\end{split}
\end{align}
where $\rho=n$ denotes the $SO(2)$ representation with charge $n$. It will be useful in what follows to repackage 
\begin{align}
    M_\lambda \equiv X_{e^{2\pi i \lambda},\lfloor\lambda \rfloor}, \ \ \ \ \lambda \in \mathbb{R}.
\end{align}
It is clear that the $M_\lambda$ are obtained by pushing the bulk Wilson lines of $\mathbb{R}$ Chern-Simons theory onto the chiral boundary, whereas the $L_a$ are topological lines which are genuinely trapped on the boundary.

Next, we study the twisted sector Hilbert spaces. To this end, it is useful to compute the induction and restriction functors. Induction follows from Equation \eqref{eqn:inductionGequiv}, which in the present situation says that\footnote{Again, there is the question of whether the $\int [d\psi]$ should be taken to be a direct integral using the Haar measure on $U(1)$, or an infinite direct sum using the discrete topology. We remain agnostic about the precise prescription.} 
\begin{align}
\begin{split}
    &I:\mathrm{Sym}^\dagger(\widehat{\mathfrak{u}}(1)) \cong \mathrm{Sym}^\dagger(\widehat{\mathfrak{su}}(2)_1)^{SO(2)} \to \mathrm{Sym}^\dagger(\widehat{\mathfrak{su}}(2)_1), \\
    &\hspace{.7in}L_a\mapsto \int_{\oplus} [d\psi] g(a,\psi), \ \ \ \ M_\lambda \mapsto g(e^{2\pi i \lambda},0).
\end{split}
\end{align}
Since induction preserves quantum dimension, this immediately implies that $\dim(M_\lambda)=1$ and $\dim(L_a)=\infty$.
Restriction is obtained by taking the transpose, so 
\begin{align}\label{eqn:resU1}
\begin{split}
    &R:\mathrm{Sym}^\dagger(\widehat{\mathfrak{su}}(2)_1)\to \mathrm{Sym}^\dagger(\widehat{\mathfrak{u}}(1)) \cong \mathrm{Sym}^\dagger(\widehat{\mathfrak{su}}(2)_1)^{SO(2)}  \\
    &\hspace{.4in}g(a,\psi) \mapsto L_a \text{ for }|a|<1, \ \ \ \ g(e^{2\pi i \lambda},0)\mapsto \bigoplus_{n\in\mathbb{Z}} M_{\lambda+n}.
\end{split}
\end{align}

Restriction determines how (twisted) representations of $\widehat{\mathfrak{su}}(2)_1$ decompose into (twisted) representations of $\widehat{\mathfrak{u}}(1)$. For instance, Equation \eqref{eqn:resU1} tells us that the Hilbert space $\widehat{\mathfrak{u}}(1)_{L_a}$ of local operators on the $\widehat{\mathfrak{u}}(1)$ boundary of $\mathbb{R}$ Chern-Simons theory on which the boundary topological line $L_a$ can terminate is isomorphic (as a twisted $\widehat{\mathfrak{u}}(1)$-module) to 
\begin{align}
    \widehat{\mathfrak{u}}(1)_{L_a} \cong (\widehat{\mathfrak{su}}(2)_1)_{g(a,\psi)}, \ \ \ \text{ for any }\psi.
\end{align}
Here, $(\widehat{\mathfrak{su}}(2)_1)_{g(a,\psi)}$ is the Hilbert space of local operators on the WZW boundary of $SU(2)_1$ Chern-Simons theory on which the boundary topological line $g(a,\psi)$ can terminate.
In particular, using Appendix A of \cite{Gaberdiel:2011aa}, we can write down the twisted partition function (cf.\ Example  \ref{ex:twisteddim}),
\begin{align}\label{eqn:twistedcharacterheisenbergvoa}
   Z_{L_a}(\tau)\equiv \mathrm{Tr}_{\widehat{\mathfrak{u}}(1)_{L_a}} q^{L_0-1/24} = \mathrm{Tr}_{(\widehat{\mathfrak{su}}(2)_1)_{g(a,\psi)}}q^{L_0-1/24} = \frac{1}{\eta(\tau)}\sum_{n\in\mathbb{Z}}q^{(-\chi+2n)^2/4}
\end{align}
where $\chi = \frac 1\pi\arccos \Re(a)$ and $\eta(\tau)=q^{1/24}\prod_{n=1}^\infty (1-q^n)$ is the Dedekind-eta function. In the special case that $a=0$, the twisted-module $\widehat{\mathfrak{u}}(1)_{L_0}$ coincides with the canonical $\mathbb{Z}_2$-charge conjugation twisted module of the Heisenberg VOA.

On the other hand, the Hilbert space $\widehat{\mathfrak{u}}(1)_{M_\lambda}$ of boundary local operators on which the bulk Wilson line $M_\lambda$ of $\mathbb{R}$ Chern-Simons theory can terminate is just a standard module of the $\widehat{\mathfrak{u}}(1)$ chiral algebra, with conformal dimension $h(\widehat{\mathfrak{u}}(1)_{M_\lambda})=\lambda^2$ and graded-dimension given by 
\begin{align}
    Z_{M_\lambda}(\tau)\equiv \mathrm{Tr}_{\widehat{\mathfrak{u}}(1)_{M_\lambda}}q^{L_0-1/24} = \frac{q^{\lambda^2}}{\eta(\tau)}.
\end{align}
The restriction functor predicts that 
\begin{align}
    (\widehat{\mathfrak{su}}(2)_1)_{g(e^{2\pi i\lambda},0)}\cong \bigoplus_{n\in\mathbb{Z}} \widehat{\mathfrak{u}}(1)_{M_{\lambda+n}}
\end{align}
as $\widehat{\mathfrak{u}}(1)$ representations. One can interpret the summand $\widehat{\mathfrak{u}}(1)_{M_{\lambda+n}}$ in the decomposition as the subspace of $(\widehat{\mathfrak{su}}(2)_1)_{g(e^{2\pi i \lambda},0)}$ with $SO(2)$-charge $\lambda+n$.

One can also determine the fusion rules of the lines. Plugging in the identifications of Equation \eqref{eqn:pairsgrho} into the formulae of \cite[Theorem 4.1]{burciu2013fusion} leads to a straight-forward derivation of the fusion rules presented in Equation \eqref{eqn:diskfusion}.

Finally, let us discuss the effective hypergroup. (Similar calculations were carried out in \cite{Bischoff:2020iac}.) By the general logic around Equation \eqref{eqn:consecutiveembeddings}, it follows that it is 
\begin{align}
K\cong SO(3) \sslash SO(2).
\end{align}
These double cosets can be identified with the interval $K=[-1,1]$ with quotient map given by 
\begin{align}
    \begin{split}
        SO(3)&\to SO(3)\sslash SO(2) \\
        R&\mapsto R_{33}
    \end{split}
\end{align}
where $R_{33}$ is the $(3,3)$ component in the $3\times 3$ matrix $R\in SO(3)$. By the standard formula \eqref{eqn:gpdoublecoset} for the multiplication, one obtains that 
\begin{align}\label{eqn:continuousmult}
    \delta_{t}\star \delta_{s} = \frac{1}{2\pi} \int_0^{2\pi} \delta_{{ts-\sqrt{(1-s^2)(1-t^2)}\cos\theta}}d\theta,
\end{align}
where $t,s\in[-1,1]\cong K$. This hypergroup has a $\mathbb{Z}_2=\{\pm 1\}\subset K$ subgroup which acts via charge conjugation on $\widehat{\mathfrak{u}}(1)$.

The category $\mathrm{Sym}^\dagger(\widehat{\mathfrak{u}}(1))$ is graded by $K$. Calling $\mathcal{B}(\widehat{\mathfrak{u}}(1))_t$ the graded component labeled by $t\in [-1,1]\cong K$, one has that the simple lines organize into the graded components according to
\begin{align}
    L_a\in \mathcal{B}(\widehat{\mathfrak{u}}(1))_{2|a|^2-1}, \ \ \ \ M_\lambda \in \mathcal{B}(\widehat{\mathfrak{u}}(1))_1.
\end{align}
For example, we see that $L_0$ is the unique topological line in the graded-component $-1\in K$ corresponding to charge conjugation, which reproduces our earlier assertion that $\widehat{\mathfrak{u}}(1)_{L_0}$ is the unique $\mathbb{Z}_2$ charge conjugation twisted module.

Finally, let us comment that, using the folding trick from Figure \ref{fig:lines2boundaries}, one can convert the boundary topological lines $L_a$ and $M_\lambda$ to boundary conditions $\partial L_a$ and $\partial M_\lambda$ of the diagonal CFT built on top of the $\widehat{\mathfrak{u}}(1)$ chiral algebra, which is none other than the $c=1$ free noncompact boson. The  $\partial M_\lambda$ can be identified with Dirichlet boundary conditions, $\partial L_0$ corresponds to the Neumann boundary condition, and the $\partial L_a$ with $a\neq 0$ can be thought of as analogs of Janik's exceptional boundaries \cite{Janik:2001hb} for the noncompact boson. To the best of our knowledge, this latter class of boundary conditions has not been considered in the literature previously. It would be interesting to study them in more detail.

\section{Future Directions}\label{sec:future}

We conclude by highlighting some possible future directions.

\begin{enumerate}[label=\arabic*)]
    \item The main concrete method for constructing a noninvertible symmetry of a chiral algebra $V$ is to find a conformal subalgebra and apply Claim \ref{cl:claim:verdef}. It would be interesting to rigorously define a notion of noninvertible symmetry which does not make any reference to conformal subalgebras. This would allow one to use symmetries to construct new chiral algebras via the orbifold procedure. One promising approach would be to translate what is known about this problem for conformal nets \cite{Bischoff:2016jmy,Bischoff:2022fxf} to the setting of vertex operator algebras, using the expanding dictionary between the two \cite{Carpi:2015fga,Carpi:2023onx,henriques2025every}. See \cite{gannonriesen} for work in this direction.
    \item There should exist a notion of ``hypergroup-crossed braided tensor category'' which generalizes the notion of $G$-crossed braided tensor category, see e.g.\ \cite{Rie22} for some discussion. Such a construction would describe topological phases enriched by noninvertible symmetry, and also the algebraic structure formed by categories of twisted-modules of a chiral algebra. Making this mathematically precise is work in progress \cite{mrrw}.
\item It would be desirable to drop various adjectives we have assumed in this paper, like rational, unitary, finite, bosonic, etc. See e.g.\ Section \ref{subsec:Heisenberg} for an example in this direction.
\end{enumerate}

\section*{Acknowledgements}
We would like to sincerely thank Anirudh Deb for his participation in earlier stages of this project. BR is grateful to Sven Möller, Ingo Runkel, and Yifan Wang, as well as Yichul Choi and Ho Tat Lam, for inspiring collaborations on related projects. We also thank Meng Cheng, Matthias Gaberdiel, Andrew Riesen, Nathan Seiberg, Sahand Seifnashri, Nikita Sopenko, and Roberto Volpato for helpful discussions. B.R.~acknowledges support from the Leinweber Foundation, the Sivian Fund, the U.S.\ National Science Foundation (NSF) under grant PHY-2210533, and the Department of Energy (DOE) under grant DE-SC0009988. BR also thanks the Yang Institute for Theoretical Physics and the Simons Center for Geometry and Physics where much of this work was carried out. The research of TG was partially supported by NSERC (Canada).

\appendix

\section{Hypergroups}\label{app:hypergroups}

In this appendix, we summarize some basic facts about hypergroups which are used throughout the main text. For more details see  \cite{Bischoff:2016jmy,Rie22,hempel2023hypergroups}.

Hypergroups can roughly be thought of as generalizations of fusion rings obtained by relaxing  integrality of the structure constants. For a finite set $K=\{r_i\}_{i\in I}$, let $\mathrm{Conv}(K)$ be the set of formal linear combinations $\sum_{i\in I} C_i r_i$ with the $C_i$ given by real, non-negative numbers. A \emph{hypergroup} is a finite set $K$ equipped with a multiplication $\cdot:K\times K\to\mathrm{Conv}(K)$, 
\begin{align}
    r_i\cdot r_j = \sum_{k\in I} P_{ij}^k r_k, \ \ \ \ P_{ij}^k\geq 0, \ \ \ \ \sum_{k\in I} P_{ij}^k =1,
\end{align}
which makes $\mathbb{R}[K]$ an associative algebra with unit given by a distinguished element $r_0\in K$. Further, a hypergroup should have an involution $i\mapsto i^\ast$, required to define a unital antiautomorphism of $\mathbb{R}[K]$, satisfying that $P_{ij}^0=P_{ji}^0$ and $P_{ij}^0>0$ if and only if $j=i^\ast$. Think of $P_{ij}^k$ as the fraction of $r_i\cdot r_j$ that lies in the channel $r_k$.

The combination 
\begin{align}
    w_i=\frac{1}{P_{ii^\ast}^0}
\end{align}
is called the \emph{weight} of $r_i$, and the weight of the hypergroup is defined as the sum of the weights of all of its elements, 
\begin{align}\label{eqn:weighthypergroup}
    D(K) = \sum_{i\in I}w_i.
\end{align}
We will also have use for a distinguished element of $\mathrm{Conv}(K)$ known as the \emph{Haar element}, 
\begin{align}\label{eqn:Haarelement}
    \omega_K = \frac{1}{D(K)}\sum_{i\in I}w_i k_i,
\end{align}
which has the defining properties that $\omega_K^\ast =\omega_K = \omega_K^2$ and $r_i\cdot \omega_K = \omega_K\cdot r_i=\omega_K$ for all $i\in I$.

The easiest example of a hypergroup is a finite group $G$, for which $\cdot$ is given by the group multiplication, $\ast$ is given by the inverse map, and $r_0$ is given by the identity element. 

More generally, any fusion ring $R$ always defines an associated hypergroup $K_R$, though not every hypergroup arises in this way.
Recall that a free, unital, associative $\mathbb{Z}$-algebra $R$ is a \emph{fusion ring} if there exists a finite basis $\{x_i\}_{i\in I}$ containing the unit such that 
\begin{align}
    x_i\cdot x_j = \sum_{k\in I} N_{ij}^k x_k, \ \ \ \ N_{ij}^k\in \mathbb{Z}_{\geq 0}.
\end{align}
There should further exist an involution $i\mapsto i^\ast$ inducing a unital anti-automorphism of $R$ and satisfying $N_{ij}^0=N_{ji}^0=\delta_{ij^\ast}$.

To see how a fusion ring induces a hypergroup, recall that the Frobenius-Perron dimension $\mathsf{d}_{x_i}$ of $x_i$ is by definition the largest real eigenvalue of the matrix $(N_{ij}^k)_{j,k\in I}$. These positive numbers satisfy the following:
\begin{align}
    \mathsf{d}_{x_i\cdot x_j}=\mathsf{d}_{x_i} \mathsf{d}_{x_j}, \ \ \ \ \  \mathsf{d}_{x_i}=\mathsf{d}_{x_{i^\ast}}, \ \ \ \ \  \mathsf{d}_{x_0}=1.
\end{align} The hypergroup induced by $R$ is obtained by setting 
\begin{align}
    K_R := \left\{ r_i:=\frac{x_i}{\mathsf{d}_{x_i}}\right\}_{i\in I},
\end{align}
letting the multiplication be the one coming from $R$,  
\begin{align}
    r_i\cdot r_j := \sum_{k\in I}P_{ij}^k r_k, \ \ \ \ \ \ P_{ij}^k=\frac{\mathsf{d}_{x_k}}{\mathsf{d}_{x_i}\mathsf{d}_{x_j}}N_{ij}^k,
\end{align}
and defining $\ast$ in the obvious way. Then $P_{ij}^k$ denotes the weighted fraction of $x_i\cdot x_j$ that lies in the $x_k$ channel, where the weight is given by the Frobenius-Perron dimension.

Note that, for a hypergroup induced by a fusion ring, the Frobenius-Perron dimensions can be detected as $P_{ii^\ast}^0=\mathsf{d}_{x_i}^{-2}$. Thus, we see that a hypergroup is induced from a fusion ring if and only if 
\begin{align}
    P_{ij}^k \sqrt{\frac{P_{kk^\ast}^0}{P_{ii^\ast}^0P_{jj^\ast}^0}}\in\mathbb{Z}_{\geq 0},
\end{align}
for all $i,j,k\in I$.

It turns out that all hypergroups which act on vertex operator algebras arise as quotients of hypergroups induced by fusion rings.  Let $H$ be a subhypergroup of $K$, defined to be a subset of $K$ which contains the unit $r_0$ of $K$ and is closed under addition, multiplication $\cdot$ and the involution $\ast$. We define the double coset represented by $r_i$ to be
\begin{align}
    Hr_iH=\{ r_j\in K \mid r_j \in h\cdot r_i \cdot h' \text{ for some }h,h'\in H\},
\end{align}
where, for $r_i\in K$ and $m=\sum_{j} C_{j}r_{j}\in\mathrm{Conv}(K)$, we say that $r_i\in m$ if $C_i>0$. We write $K\sslash H$ for the set of all double $H$-cosets in $K$.

Define $r_i\sim_H r_j$ when $r_j\in Hr_iH$. This is an equivalence relation with equivalence classes $[r_i]:=Hr_iH$. One can show that $r_i\sim_H r_j$ if and only if $\omega_H\cdot r_i \cdot \omega_H = \omega_H \cdot r_j\cdot  \omega_H$. The involution on $K\sslash H$ is $[r_i]^\ast=[r_{i^\ast}]$ and structure constants are given by 
\begin{align}
   \overline{P}_{[i][j]}^{[k]}=\sum_{\substack{k'\in I \\ r_{k'}\sim_Hr_k}}P_{ij}^{k'}\,,
\end{align}
the fraction of $r_i\cdot r_j$ lying in any channel equivalent to $r_k$. Another characterization of the structure constants of $K\sslash H$ can be obtained using the Haar element,
\begin{align}
    [r_i]\cdot [r_j] = q_H(r_i\cdot \omega_H\cdot r_j)
\end{align}
where $q_H$ is defined by extending the assignment $r_i\mapsto [r_i]$ linearly to a map $\mathbb{R}[K]\to \mathbb{R}[K\sslash H]$. This map satisfies
\begin{align}
    \mathbb{R}[\{r_0\}]\to \mathbb{R}[H]\hookrightarrow \mathbb{R}[K]\twoheadrightarrow \mathbb{R}[K \sslash H] \to\mathbb{R}[\{[r_0]\}]\,.
\end{align}

An important example is the double cosets of one finite group inside another (cf.\ Example \ref{ex:cosine}). When $H$ is normal, the quotient construction $G\sslash H$ simply recovers the standard quotient group $G/H$. However, when $H$ is not normal, one obtains a non-trivial hypergroup with multiplication given by 
\begin{align}\label{eqn:gpdoublecoset}
    [g]\cdot [g'] = \frac{1}{|H|}\sum_{h\in H} [ghg'],
\end{align}
and with involution given by $[g]^\ast = [g^{-1}].$

More generally, one may apply the quotient construction to hypergroups $K_S\subset K_R$ induced by fusion rings $S\subset R$. In that case the structure constant $\overline{P}_{[i][j]}^{[k]}$ is the weighted fraction of $x_i\cdot x_j$ lying in any channel $x_{k'}$ equivalent to $x_{k}$. This is the hypergroup grading discussed in e.g.\ Claim \ref{cl:cla:hypergroupgrading}. 

Suppose that $R$ is represented on a vector space $V$, i.e.\  there is a unital algebra homomorphism $\rho:\mathbb{C}[K_R]\to \mathrm{End}_{\mathbb{C}}(V)$. If $R=\{x_i\}_{i\in I}$ has  kernel given by $S=\{x_i\}_{i\in J}$ in the sense that $\rho(x_i)=\mathsf{d}_{x_i} \mathrm{Id}_V$ for all $i\in J$, then the action of $K_R$ descends to an action $\tilde{\rho}$ of the corresponding quotient hypergroup,
\begin{align}
\begin{split}
    \tilde{\rho}:K_R\sslash K_S&\to \mathrm{End}_{\mathbb{C}}(V) \\
    [r_i]&\mapsto \rho(r_i).
\end{split}
\end{align}
When $R$ and $S$ arise from the Grothendieck rings of a fusion category $\mathcal{C}$ and one of its fusion subcategories $\mathcal{D}$, respectively, we use $K_{\mathcal{C}}\sslash K_{\mathcal{D}}$ to denote the quotient hypergroup $K_R\sslash K_S$.

\begin{landscape}
\begin{table}[p]
\section{The \texorpdfstring{$E_8$}{E8} fusion ring and associated hypergroup}\label{app:E8}
\vspace{.2in}
\centering
\footnotesize
\[
\begin{array}{c|cccccccc}
\otimes & \mathds{1} & W & \mathds{1}_{a} & W_{a} & \mathds{1}_{b} & W_{b} & \mathds{1}_{c} & W_{c}\\
\midrule
\mathds{1} & \mathds{1} & W & \mathds{1}_{a} & W_{a} & \mathds{1}_{b} & W_{b} & \mathds{1}_{c} & W_{c}\\
W & & \mathds{1} \oplus W & W_{a} & \mathds{1}_{a} \oplus W_{a} & W_{b} & \mathds{1}_{b} \oplus W_{b} & W_{c} & \mathds{1}_{c} \oplus W_{c}\\
\mathds{1}_{a} & & & \mathds{1} \oplus \mathds{1}_{b} & W \oplus W_{b} & \mathds{1}_{a} \oplus W_{c} & W_{a} \oplus \mathds{1}_{c} \oplus W_{c} & W_{b} & \mathds{1}_{b} \oplus W_{b}\\
W_{a} & & & & \mathds{1} \oplus W \oplus \mathds{1}_{b} \oplus W_{b} & W_{a} \oplus \mathds{1}_{c} \oplus W_{c} & \mathds{1}_{a} \oplus W_{a} \oplus \mathds{1}_{c} \oplus 2W_{c} & \mathds{1}_{b} \oplus W_{b} & \mathds{1}_{b} \oplus 2W_{b}\\
\mathds{1}_{b} & & & & & \mathds{1} \oplus \mathds{1}_{b} \oplus W_{b} & W \oplus \mathds{1}_{b} \oplus 2W_{b} & W_{a} \oplus W_{c} & \mathds{1}_{a} \oplus W_{a} \oplus \mathds{1}_{c} \oplus W_{c}\\
W_{b} & & & & & & \mathds{1} \oplus W \oplus 2\mathds{1}_{b} \oplus 3W_{b} & \mathds{1}_{a} \oplus W_{a} \oplus \mathds{1}_{c} \oplus W_{c} & \mathds{1}_{a} \oplus 2W_{a} \oplus \mathds{1}_{c} \oplus 2W_{c}\\
\mathds{1}_{c} & & & & & & & \mathds{1} \oplus W_{b} & W \oplus \mathds{1}_{b} \oplus W_{b}\\
W_{c} & & & & & & & & \mathds{1} \oplus W \oplus \mathds{1}_{b} \oplus 2W_{b}
\end{array}
\]
\caption{The fusion rules of the $E_8$ fusion ring.}\label{tab:E8}
\end{table}

\begin{table}[ht]
\centering
\begin{tabular}{c|cccc}
$\star$ & $\tilde{1}$ & $\tilde{a}$ & $\tilde{b}$ & $\tilde{c}$ \\
\midrule
$\tilde{1}$ & $\tilde{1}$ & $\tilde{a}$ & $\tilde{b}$ & $\tilde{c}$ \\
$\tilde{a}$ & $\tilde{a}$ & $\tilde{1}+\tilde{b}$ & $\tilde{a}+\varphi \tilde{c}$ & $\varphi \tilde{b}$ \\
$\tilde{b}$ & $\tilde{b}$ & $\tilde{a}+\varphi \tilde{c}$ & $\tilde{1}+(1+\varphi)\tilde{b}$ & $\varphi \tilde{a}+\varphi \tilde{c}$ \\
$\tilde{c}$ & $\tilde{c}$ & $\varphi \tilde{b}$ & $\varphi \tilde{a}+\varphi \tilde{c}$ & $\tilde{1}+\varphi \tilde{b}$
\end{tabular}
\caption{The rescaled multiplication of the hypergroup $K_{\mathcal{C}_{E_8}}\sslash K_{\mathrm{Fib}}$ with $\varphi=\frac{1+\sqrt{5}}{2}$.}
\label{tab:E8hyp}
\end{table}
\end{landscape}

\bibliographystyle{ytphys}
\bibliography{main}

\end{document}